%% file: main_SSRN.tex
\documentclass[11pt, twoside]{article}
\usepackage[utf8]{inputenc}
\usepackage{amsmath}
\usepackage{amssymb}
\usepackage{amsthm}
\usepackage{enumitem, multirow, booktabs}
\usepackage[margin = 1in]{geometry}
\usepackage{secdot} 
\usepackage{authblk}
\usepackage{changepage, booktabs}
\usepackage{hyperref}
\usepackage{url, diagbox}
\usepackage{mathrsfs}
\usepackage{amsthm}
\usepackage[normalem]{ulem}
\usepackage{enumitem}
\usepackage{comment}
\usepackage{natbib}
\usepackage{babel}
\usepackage{graphicx}
\usepackage{bm} 
\usepackage{cancel} 
\usepackage{caption}
\usepackage[table,xcdraw]{xcolor}
\usepackage{titletoc}
\usepackage{mathtools}
\usepackage{graphicx}
\usepackage{parskip}
\usepackage{algpseudocode}
\usepackage[ruled,vlined]{algorithm2e}
\usepackage{setspace}
\usepackage{float}
\usepackage{placeins}
\usepackage{natbib}

\DeclareCaptionFormat{custom}
{%
    \textbf{#1#2}\textit{\small #3}
}
\graphicspath{ {./images/} }

\renewcommand{\phi}{\varphi}

\theoremstyle{definition}

\newtheorem{remark}{Remark}[section]

\newtheorem{proposition}{Proposition}[section]

\renewcommand{\tilde}{\widetilde}

\allowdisplaybreaks

\newtheorem{corollary}{Corollary}[section]
\makeatother
\usepackage{thmtools}
\usepackage[framemethod=TikZ]{mdframed}
\mdfsetup{skipabove=1em,skipbelow=0em}

\usepackage{cleveref}
\theoremstyle{definition}

\title{Proactive Inpatient Bed Requests for \\ Emergency Department Admissions\\
{\small \textit{Submitted to Manufacturing \& Service Operations Management}}}
\author[1]{Qian Cheng}
\author[2]{Nilay Tan\i k Argon}
\author[2]{Aniruddhan Ganesaraman}
\author[2]{Serhan Ziya}

\affil[1]{United Airlines}
\affil[2]{Department of Statistics and Operations Research, University of North Carolina at Chapel Hill}

\begin{document}

\maketitle

\let\thefootnote\relax\footnotetext{\textit{Emails:} \href{mailto:registerericcheng@gmail.com}{registerericcheng@gmail.com}, \href{mailto:nilay@unc.edu}{nilay@unc.edu}, \href{mailto:aniruddhan_ganesaraman@unc.edu}{aniruddhan\_ganesaraman@unc.edu}, \href{mailto:ziya@unc.edu}{ziya@unc.edu}}


\subsection*{Abstract.} {\bf Problem definition:} Emergency department (ED) boarding occurs when patients admitted to the hospital remain in the ED while waiting for inpatient beds. Boarding is viewed as a major driver of ED crowding and has been associated with poor patient outcomes. We propose a framework to help EDs reduce boarding time and total length of stay by using information about current ED patients and hospital bed availability to proactively request inpatient beds before admission decisions are finalized. 

{\bf Methodology/results:} We formulate the problem as a Markov decision process in which predictions of each patient's probability of hospital admission and time to disposition are aggregated to guide early inpatient bed requests. This formulation leads to three data-driven policies based on approximate dynamic programming, reinforcement learning, and a newsvendor-type approach. Using a simulation model based on data from a large ED, we evaluate these policies across a wide range of settings. The simulation study shows that proactive aggregate bed requests can reduce average boarding times for admitted patients by 30–70\% and average length of stay for all ED patients by 6–15\%, while creating only modest idle time for prepared inpatient beds. The newsvendor heuristic provides the most attractive tradeoff between ED performance and inpatient bed idle time, whereas the reinforcement learning-based heuristic produces smoother bed-request patterns when stability in downstream hospital processes is especially important. 

{\bf Managerial implications:} Our work shows how EDs can use prediction tools to make proactive, ED-level bed-request decisions that improve ED operations while helping managers balance reductions in ED delays against limited idle time for prepared inpatient beds. Our findings also illustrate the value of evaluating both simple myopic heuristics and more sophisticated reinforcement learning-based approaches in this operational problem, since each can offer distinct advantages depending on the performance measures and implementation constraints most important to managers.

\textbf{Keywords.} Emergency department, Markov decision processes, reinforcement learning.

\section{Introduction}
\cite{american2011definition} defines a {\em boarded patient} as one who remains in the emergency department (ED) after being admitted or placed under observation status, but has not yet been transferred to an inpatient or observation unit. In most EDs and hospitals, boarded patients may remain in the ED for hours after the admission decision, simply because hospital beds and care teams are not immediately available. This process of arranging inpatient care can be lengthy, particularly during certain hours of the day. 

{\em ED boarding}, which refers to the presence of boarded patients in the ED, has been linked to poor patient outcomes as well as subpar operational performance. From the patient perspective, individuals generally dislike being held in the ED once they no longer need emergency care \citep{viccellio2013patients}. In addition, several studies have shown that extended ED boarding is associated with a higher risk of adverse health outcomes \citep{chalfin2007impact, singer2011association}. From an operational perspective, boarding has been identified as a major contributor to ED crowding \citep{hoot2008systematic,sartini2022,pearce2023}. Boarded patients continue to occupy ED beds, consume ED resources, and delay the treatment of other ED patients \citep{bair2010impact, white2013boarding}. Boarding is also costly; \cite{Canellas2024} show that caring for a boarding patient in the ED can cost up to twice as much as treating a comparable patient in an inpatient setting.

\begin{figure}[ht!]
	\begin{center}
		\includegraphics[scale=0.2]{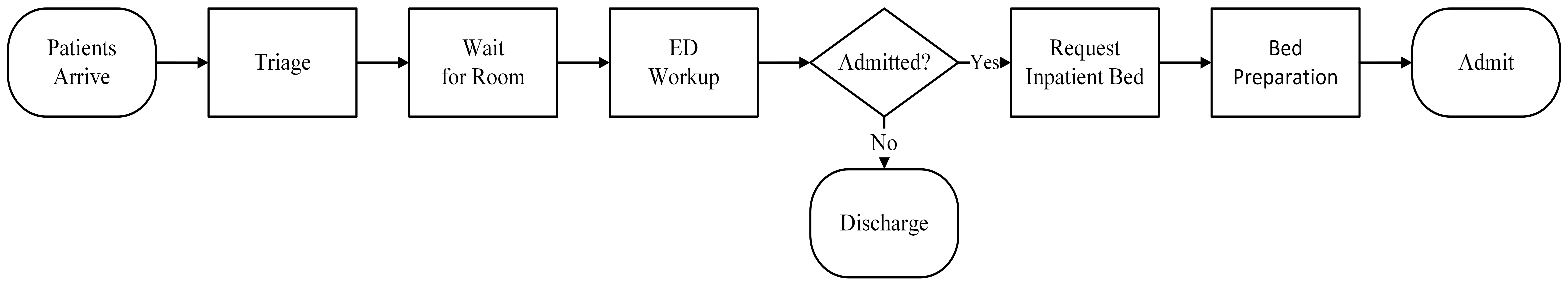}
		\caption{Patient flow in a typical ED.}
		\label{fig:patientworkflow} 
	\end{center}
\end{figure}
\vspace{-10pt}

To understand how boarding contributes to ED crowding, we first take a closer look at the  patient flow in a typical ED as depicted in Figure~\ref{fig:patientworkflow}. The process  begins with triage, where a nurse evaluates and prioritizes patients based on a rapid assessment of their condition upon arrival. This process assigns patients to categories such as the five-level Emergency Severity Index (ESI), ranging from 1 (most acute) to 5 (least acute), which is the most commonly used triage system in the U.S. \citep{gilboy2012emergency}. After triage,  patients are taken to a waiting area if all ED beds are occupied, and are then brought into the treatment area according to a priority mechanism largely based on their triage levels and times of arrival. In most cases, a patient's workup in the ED begins only after the patient is assigned an ED bed. With the exception of patients who leave without being seen or are transferred to another facility, most patients are either discharged or admitted once their ED workup is completed. According to the current practice, the ED places a request for a hospital bed only after a patient has received an admission decision. The transfer preparation process, which includes identifying inpatient beds and care teams as well as completing all necessary steps for the start of patients' transfer, can be time consuming, even when a hospital bed is readily available. If a bed is available at the time of the request, the patient’s boarding time is essentially equal to the transfer preparation time. If no bed is available, the boarding time is even longer.

In this paper, we propose a framework to reduce ED crowding caused by prolonged boarding times by modifying the inpatient bed request process. The core idea is that a patient’s likelihood of hospital admission can be predicted with reasonable accuracy as early as the triage stage, and these predictions can be used to initiate bed requests earlier than in current practice. A similar mechanism was proposed by \cite{chen2022using}, in which admission probabilities are estimated at triage and a bed is requested immediately for patients deemed highly likely to be admitted.  That work demonstrates that individualized early bed requests can reduce  the average length of stay in the ED without generating excessive false early bed requests.  

With this work, we consider an alternative approach to \cite{chen2022using}. {\em Rather than requesting beds for individual patients, we focus on ED-level decision making and generate aggregate bed requests periodically}. Specifically, we envision a decision mechanism that, at regular intervals (e.g., every hour), observes the current ED population, estimates each patient's admission probability and time-to-disposition decision distribution, and uses these estimates together with hospital bed availability to determine how many additional beds should be requested. By aggregating individual-level admission predictions, this approach reduces the impact of uncertainty in single-patient predictions while incorporating both ED operational conditions and hospital capacity. Moreover, discussions with hospital and ED managers suggest that aggregate bed requests may be better aligned with operational practice. Aggregate requests can provide hospital management with useful information about the current and anticipated state of the ED, and  initiating bed preparation earlier may be more acceptable when requests are not tied to specific individual patients~\citep{linthicum2018improving}.

An important assumption underlying the proposed framework is that inpatient beds are largely interchangeable across patients, which may not hold for highly specialized hospital units. This assumption is most appropriate for patients admitted to general medicine units. In our partner healthcare system, which includes more than 15 hospitals with EDs across North Carolina, U.S.A., the proportion of ED-admitted patients assigned to “General Medicine” services was the highest among all service categories during the calendar year 2019, ranging from 16\% at an academic hospital with many specialty inpatient units to 70\% in a large non-academic hospital. Furthermore, several studies recommend acute medical units or centralized admitting units as mechanisms for improving the flow of ED-admitted patients \citep{reid2016effectiveness, moloney2005impact}.

The key contributions of this study are threefold: (i) we propose a generalized decision-making framework that uses patient admission probability and time-to-disposition decision predictions to generate ED bed requests with the goal of balancing operational concerns from both the ED and hospital perspectives; (ii) we develop multiple heuristic policies  grounded in Markov decision processes (MDPs) and reinforcement learning; and (iii) we conduct a high-fidelity simulation study based on patient data from an actual ED to quantify the impact of these policies on key operational performance metrics and generate insights into which policies to use under what conditions. Our simulation study shows that the proposed data-driven framework has the potential to substantially reduce the average ED boarding times for admitted patients as well as the average length of stay for all ED patients, while requiring only modest average idle time for prepared inpatient beds. These benefits persist even when the simulated ED experiences a winter pandemic scenario or when only a fraction of inpatient beds are eligible for early bed requests.


\section{Literature Review}
\label{S2:LR}

In recent years, there has been growing interest in applying statistical and machine learning methods to predict hospital admissions and LoS for ED patients, often with the goal of improving patient flow~\citep{Shin2025}. Studies of individual-level admission prediction have used traditional statistical models, such as regression~\citep{sun2011predicting, peck2012predicting, mehrotra2017, lee2020prediction}, as well as more advanced machine learning methods, including neural networks and natural language processing~\citep{zhang2017prediction, roquette2020prediction, somanchi2022predict, kishore2023}. Similarly, prior studies have applied regression, machine learning, and classification methods to predict ED LoS and waiting times in the form of point estimates, percentiles, or interval-based probabilities~\citep{ang2016accurate, azari2015imbalanced, rahman2020using, kadri2023}.

The primary motivation for admission prediction tools is to identify patients likely to be admitted early in their ED stay, allowing transfer preparation to begin proactively by notifying inpatient units in advance. While this idea is intuitively appealing and has shown promise in simulation studies~\citep{peck2014characterizing}, structured decision-making models to guide implementation remain limited. To the best of our knowledge, only three studies have systematically addressed this question. \cite{qiu2015cost} formulate the individual bed-request decision as a newsvendor problem, focusing on the cost-effective timing of requests when admission likelihood and LoS distributions are known. \cite{lee2021proactive} assume that beds are requested for all predicted admissions and use a Markov process to evaluate system performance under varying prediction accuracies. Most closely related to our work, \cite{chen2022using} propose a framework in which an early bed request is made at triage for individual patients based on their predicted admission probability and the current ED occupancy level, with the goal of reducing average ED LoS while limiting false requests. Our work departs from these studies by shifting the decision from the patient level to the system level: instead of deciding whether to request a bed for each patient at triage, we determine the timing and number of aggregate inpatient bed requests based on broader conditions in both the ED and the hospital. This system-level approach supports more practical implementation by allowing requests to be made at predefined intervals set by ED management.

Although the literature directly addressing our decision problem is limited, a broad operations management literature studies patient flow and operational efficiency in EDs and hospitals; see \cite{wiler2011review}, \cite{saghafian2015operations}, and \cite{dai2021recent} for comprehensive reviews. Closely related studies examine hospital-wide patient flow and ED--ward interactions~\citep{armony2015patient}, proactive diversion decisions based on ED arrival forecasts~\citep{xu2016using}, and provider versus nurse triage using many-server approximations~\citep{kamali2019use}. Another stream studies ED patient streaming strategies based on admission likelihood, case complexity, or resource needs~\citep{saghafian2012patient, saghafian2014complexity, feizi2023a}. Finally, \cite{feizi2023b} examine admission batching for ED patients, showing that batching can improve physician productivity but may increase variability in bed requests and delay downstream inpatient processes.

Other relevant work investigates physician behavior, task prioritization, and staffing decisions in EDs. Studies in this stream analyze providers’ choices between new and in-progress patients using queueing models~\citep{huang2015control}, develop centralized scheduling systems that use real-time patient data and robust-stochastic programming to recommend which patient each provider should see next~\citep{he2019data}, and empirically show that decision makers tend to prioritize patients who are more likely to be discharged when the ED is crowded~\citep{li2021next}. Related work on ED staffing addresses time-varying arrivals, dynamic resource allocation, surge capacity planning, and nonstationary physician productivity~\citep{andersen2019staff, chan2021dynamic, hu2021prediction, zaerpour2022scheduling}.

Complementing the literature on ED patient flow, a related stream studies inpatient-process interventions that improve hospital operations and indirectly reduce ED congestion. These include discharge timing and rounding policies to make inpatient beds available sooner~\citep{shi2016models, chan2017queues} and capacity-management policies for step-down units, ICUs, and hospital wards~\citep{armony2018critical, dong2020structural}. Other studies use dynamic optimization to manage patient overflow and ICU-to-ward transfers~\citep{dai2019inpatient, gonzalez2019proactive}, select care units at admission while balancing readmission risk and bed availability~\citep{zhalechian2020personalized}, and determine when and which patients to discharge using personalized readmission-risk predictions~\citep{shi2021timing}.

Last but not least, there is a growing body of empirical research aimed at improving our understanding of ED and hospital operations while offering useful managerial insights. Recent examples include \cite{song2015diseconomies}, 
\cite{batt2017early}, \cite{kc2017benefits}, \cite{ding2019patient}, \cite{batt2019effects}, \cite{soltani2022does},  and \cite{niewoehner2022physician}.

\section{Problem Formulation and Benchmark Policies}
\label{S4:PFBP}

In this section, we develop a general framework for using real-time information on ED patients and hospital bed status to determine how many inpatient beds to request for potentially soon-to-be admitted ED patients. An effective decision framework would improve ED operational outcomes such as decrease the average LoS in the ED for all patients, while taking into account the hospital's concerns about using their key resources suboptimally, which would manifest itself by the presence of prepared beds that remain unoccupied for extended periods of time.

One of the key assumptions underlying the proposed framework is that requested hospital beds are interchangeable, i.e., each requested bed can be occupied by any ED patient who ended up being admitted to the hospital. Although this may be a restrictive assumption, many hospitals have hospital care teams and beds (such as those in  general medicine wards) that are flexible to accommodate patients with a wide range of ailments. Furthermore, our formulation can be easily extended to the case with multiple types of hospital beds and care teams. For the sake of simplicity, we present the formulation under the assumption that all requested hospital beds are interchangeable, which we later partially relax in our simulation experiments.  

We consider the problem over an infinite planning horizon, where the system is observed periodically at discrete points in time. We call these discrete-time points {\em decision epochs}, and the time between any two consecutive epochs a {\em period}. The choice of the frequency of decision epochs may depend on various factors such as the accessibility of information about patients and hospital beds, and the computational complexity of the solution methods. At each decision epoch $t\geq 1$, the decision maker observes the system state denoted by $\mathbf{s}_t:=(\mathbf{P}_t,\mathbf{B}_t)$, where $\mathbf{P}_t$ represents the state of all patients who are present in the ED with a positive chance of being later admitted to the hospital, and $\mathbf{B}_t$ denotes the state of hospital beds requested by the ED. Specifically, $\mathbf{P}_t$ includes patients waiting for an ED bed, occupying an ED bed without a disposition decision, or boarding. Let $n_t$ be a non-negative integer that denotes the number of all such patients in the ED at time $t$. For each patient $i\in\{1,2,\ldots,n_t\}$ present in the ED at time $t$,  where $n_t\geq 1$,  we assume that the decision maker has an estimate of the distribution of time-to-disposition decision, which is represented by $\boldsymbol{\gamma}_{i,t}=(\gamma_{i,t,1},\gamma_{i,t,2},\ldots)$, with $\gamma_{i,t,k}$ denoting the estimated probability at time $t$ of patient $i$ receiving a disposition decision in period $(t+k-1,t+k]$ for $k\geq 1$. For each patient $i$, the predicted admission probability at time $t$ is represented by $\alpha_{i,t}$, which means that at time $t$, this patient is anticipated to be either admitted with probability $\alpha_{i,t}$ or discharged with probability $1-\alpha_{i,t}$ after their stay in the ED. Information for each patient $i$ at time $t$ is completely characterized by  $\mathbf{P}_{i,t}:=(\boldsymbol{\gamma}_{i,t}, \alpha_{i,t})$. (With this notation,  patient $i$, who is boarding at time $t$ has $\mathbf{P}_{i,t}:=((1,0,\ldots),1)$.) The state of all ED patients at time $t$ is then expressed as
$\mathbf{P}_t=\left(\mathbf{P}_{1,t},\mathbf{P}_{2,t},\ldots,\mathbf{P}_{n_t,t}\right)$.

The vector $\mathbf{B}_t$ captures the state of all ED-requested hospital beds, including those occupied, under preparation, or prepared but not yet assigned to an ED patient. Let $m_t$ be the number of such beds at time $t$. For each requested bed, the remaining preparation time denotes the time until the bed is available and ready for an ED-admitted patient. The distribution of remaining preparation time for  requested hospital bed $j\in\{1,2,\ldots,m_t\}$, with $m_t\geq 1$, is denoted by vector $\boldsymbol{\zeta}_{j,t}=(\zeta_{j,t,1},\zeta_{j,t,2},\ldots)$, where $\zeta_{j,t,k}$, for $k \geq 1$, represents the estimated probability at time $t$ that bed $j$ finishes preparation in period $(t+k-1,t+k]$. (Using this notation, the state of a requested bed that is already prepared but is unoccupied is denoted by $(1,0,\ldots)$.)  The state of all requested hospital beds, $\mathbf{B}_t$, is then represented by $\mathbf{B}_t:=(\boldsymbol{\zeta}_{1,t}, \boldsymbol{\zeta}_{2,t}, \ldots, \boldsymbol \zeta_{m_t,t})$. 
These estimates for hospital beds under preparation and ED patients' admission probabilities and times to disposition decision, all given as part of state $\mathbf{s}_t$, can be updated as more information becomes available. 

Let action $a_t$ be the number of inpatient beds that are requested from the hospital at decision epoch $t$ with action space $\mathcal A:=\{0,1,\ldots, A\}$, where $A$ is a pre-defined positive integer. We assume that the bed requests are sent immediately to the hospital and consequently, $a_t$ hospital beds, either available or unavailable, will be identified for use by patients admitted from the ED. Let $C_t^\pi$ be the total cost of running this bed request system in period $t$ under policy $\pi$, which is a collection of decision rules that maps the state space to the action space for every decision epoch $t\geq 1$. The objective is to minimize the expected total discounted cost over an infinite horizon, i.e., 
$
\inf_{\pi\in\Pi} \mathbb{E} [\sum_{t=1}^{\infty} \rho^t C_t^\pi]
$,
where $\rho \in (0,1)$ is a discount factor and $\Pi$ is the set of all admissible policies.

Specifying an appropriate cost structure is one of the main challenges in formulating this decision problem. At a high level, we envision $C_t^\pi$ to have three main components to capture the basic trade-offs. First, because the primary objective is to reduce ED boarding, the cost should include a penalty that increases with the number of boarding patients. Second, to avoid placing unnecessary burden on the hospital and inefficiently using scarce resources, the cost should penalize requested beds that have completed preparation but remain unoccupied. Third, because large fluctuations in bed requests can delay downstream hospital processes~\citep{feizi2023b}, the cost may include a convex increasing function of the request size $a_t$ to encourage smoother bed-request patterns. In our numerical study, we use the relative weights on these three components as tuning parameters for the proposed heuristic policies and evaluate performance metrics such as long-run average ED LoS under different weight choices. This allows ED and hospital managers to tailor the framework to their operational conditions and priorities.

We conclude this section by introducing two benchmark policies that will later serve as reference points when we evaluate the performance of proposed heuristic policies. 

\noindent\textbf{Current Practice (CP):} According to the current practice in most EDs, a hospital bed is requested for a patient only after their disposition decision is made.  In our formulation, this corresponds to requesting beds at a decision epoch only for patients whose admission decisions have been finalized in the previous period. In other words, we let 
$$
a_t=\max \left(0,\sum_{i=1}^{n_t} \mathbb{I}\left[(\boldsymbol{\gamma}_{i,t}, \alpha_{i,t}) = ((1,0,\ldots),1)\right] -  m_t \right)\ 
$$
at state $\mathbf s_t$, for all $t\geq 1$, where $ \mathbb{I}[\cdot]$ is the indicator function that takes the value of 1 if its argument is true and 0 otherwise. Here, the sum of the indicator functions corresponds to the total number of boarding patients in the ED at time $t$, and hence, subtracting $m_t$ from this value gives the number of hospital beds that are needed for admitted patients when this difference is positive.

\noindent \textbf{Greedy Policy (GP):} This policy aims to match the expected admission demand with the number of requested hospital beds by ignoring patients’ time-to-disposition distributions and hospital bed preparation-time distributions. Specifically, it estimates the expected number of current ED patients who will eventually be admitted, compares it with the number of already requested beds, and requests the rounded shortfall. Thus, in state $\mathbf s_t$ for all $t\geq 1$, we let
$a_t=\max ( 0,\sum_{i=1}^{n_t} \alpha_{i,t} - m_t)$.

\section{MDP Formulation and Resulting Heuristics}
\label{S5:MDP}

Although the generalized framework of Section \ref{S4:PFBP} provides a detailed representation of the system without making overly restrictive assumptions, its complexity and high dimensionality make direct analysis and policy derivation very difficult. Therefore, in Section~\ref{S5:MDP:SS1}, we introduce an alternative MDP formulation that paves the way to three heuristics as discussed in Sections~\ref{S5:MDP:SS2}, \ref{S5:MDP:SS3}, and \ref{S5:MDP:SS4}.

\subsection{MDP Formulation}
\label{S5:MDP:SS1}

Consider an infinite-horizon, discrete-time MDP, where actions are taken at decision epochs $t\geq 1$. Suppose that ED patients can be categorized into $M \geq 2$ classes based on their chief complaints, age, comorbidities, etc., that would inform the decision maker about patients' chances of admission, arrival distributions, and distributions of time to disposition. In this formulation, patients cannot transition from one class to another during their stay at the ED, and hence, all patient characteristics  observed by the decision maker stays unchanged during the process. 

This formulation assumes that there are infinitely many ED and hospital beds.  This means that an ED patient will  be roomed immediately upon arrival, and that the bed preparation process for each hospital bed will begin as soon as it is requested by the ED. (Although these assumptions are not realistic, the simulation experiments in Section~\ref{S7:NS} show that the resulting policies still perform well under finite bed capacities.) The number of arrivals at the ED during a period has a finite mean and is independent of the past arrivals, the state of the ED and the hospital, and past actions. Also, the probability that a patient receives a disposition decision during a period is a constant that depends on the class only and is independent of everything else. We assume that the decision maker can observe the true admission probability $\alpha_m$ for each class $m$ patient that stays unchanged throughout their stay in the ED until the disposition decision. For notational convenience, here we only present the time-homogeneous formulation, although it is possible to formulate the problem with non-homogeneous arrival rates and service rates based on the time of day.

A patient with an admission decision is transferred immediately to the hospital if there is at least one hospital bed that is available for use; otherwise, the admitted patient will board in the ED until a hospital bed becomes available. At each decision epoch, a decision maker observes the state of the ED and the hospital, and requests a number of beds from the hospital. Since we assume an infinite hospital bed capacity, these requested beds will start being prepared immediately. In every period, each requested hospital bed will be ready with a certain probability, independently of everything else. A hospital bed that is ready will be immediately occupied by a patient if there is at least one boarding patient in the ED; otherwise, it will remain idle until another ED patient is admitted.

Under the above assumptions, the state of the system at decision epoch $t$ can be redefined as $\mathbf{s}_t:=(n_{1,t},\ldots,n_{M,t},k_t,r_t)$, where $n_{i,t}$ is a non-negative integer denoting the number of class $i$ patients in the ED who are yet to receive a disposition decision. When positive, $k_t$ denotes the number of boarding patients in the ED at time $t$. If $k_t$ is negative, then $|k_t|$ denotes the number of hospital beds that are available for use by future ED-admits. If $k_t$ is zero, then there is neither any boarding patient nor any hospital bed waiting. Finally, $r_t$ is a non-negative integer denoting the number of requested hospital beds that are under preparation. Let $\mathcal{S}$ be the state space of all such $\mathbf{s}_t$. 

As in the general formulation of Section \ref{S4:PFBP}, we let the action taken by the decision maker at each decision epoch $t$ be denoted by $a_t \in \mathcal{A}$, which is the number of hospital beds requested. Corresponding to the cost structure of Section \ref{S4:PFBP}, we let the immediate cost function be $C(\mathbf{s}_t,a_t):=a_t^2 + C_{b,e}(\mathbf{s}_t)$, where the quadratic term helps penalize the variability in the number of beds requested, $C_{b,e}(\mathbf{s}_t):=c_b \max(0, k_t) + c_e \max(0,-k_t)$ represents the sum of the penalty for having boarding patients and the penalty for having unused hospital beds, $c_b \geq 0$ is the cost of having a single boarding patient per period, and $c_e \geq 0$ is the cost of having one unused prepared hospital bed per period. We regard both $c_b$ and $c_e$ as tuning parameters that determine the relative weights given to the three parts of the per-period cost.

Defining policy $\pi$ as before, the optimization problem can then be stated as
\begin{equation}
\inf_{\pi\in \Pi} \mathbb{E} \Big[\sum_{t=1}^{\infty} \rho^t C(\mathbf{s}_t,a_t) \Big],\label{eq:optimal}
\end{equation}
and the optimality equation for this infinite-horizon discounted MDP is given by
\begin{equation}
\label{eq:op1}
  \begin{aligned} 
V(\mathbf{s})
= \inf_{a \in \mathcal{A}} \Big\{C(\mathbf s, a) + \rho \sum_{\mathbf{ j}\in\mathcal{S}} p(\mathbf j|\mathbf s,a) V(\mathbf j) \Big\},
\end{aligned}  
\end{equation}
where $V(\mathbf s)$ is the optimal expected total discounted cost for initial state $\mathbf s\in\mathcal{S}$ and $p(\mathbf j| \mathbf s,a)$ denotes the probability that the next state will be $\mathbf j\in\mathcal{S}$ given current state $\mathbf s$ and action $a\in \mathcal{A}$. Proposition \ref{theo:1} establishes the existence of a solution to the optimality equation~\eqref{eq:op1} and shows that an optimal stationary policy can be derived based on this solution. (The proofs of all analytical results are deferred to the Online Supplement.)

\begin{proposition}
\label{theo:1}
The optimality equation~\eqref{eq:op1} possesses a unique solution. Furthermore, if there exists a policy $\pi^*$ that selects the action for each state $\mathbf s$ to attain the infimum in \eqref{eq:op1}, then $\pi^*$ is an optimal stationary policy associated with the optimal value function $V(\mathbf s)$.
\end{proposition}

Although the above MDP formulation is a significantly simplified representation of the problem, it still suffers from the curse of dimensionality. Even in the case with only two classes of patients, the state space can be still large for a medium-sized ED with around 40,000 visits per year. The transition probability $p(\cdot|\cdot,\cdot)$ gets even harder to calculate, and the computational burden grows exponentially as the number of patient classes increases. Consequently, solving the MDP using standard methods such as value iteration is computationally intensive if not impossible, and therefore, we next propose three heuristic approaches that can provide practical and effective solutions.

\subsection{MDP-Based Heuristic Policy (MHP)}
\label{S5:MDP:SS2}

To address the difficulty of solving the MDP defined by \eqref{eq:op1}, we next adopt an approach that approximates the cost-to-go function, i.e., the second term inside the infimum in \eqref{eq:op1}. In particular, we introduce an auxiliary finite-horizon, continuous-state MDP with a specially structured immediate cost, for which the optimal solution can be obtained analytically. We then substitute the cost-to-go function for this auxiliary problem into the original MDP, which leads to a heuristic policy we call the MDP-based heuristic policy (MHP). This approach is inspired by \cite{shi2021timing}, who used this method to dynamically determine which patients to discharge from a hospital ward.

We first introduce an auxiliary finite-horizon MDP that spans $G+1$ periods starting at time $t$, i.e., it spans periods $t, t+1, \ldots, t+G$, for fixed $t\geq 0$ and a pre-defined positive integer $G$. When there are $g \in \{0,1,\ldots, G\}$ periods left, the state of this MDP is defined by a scalar $\phi_{t,g} \in \mathbb{R}$, which quantifies the imbalance between the demand (ED patients requiring admission) and the supply (hospital beds that have been requested). A positive/negative value for $\phi_{t,g}$ indicates that there is excess demand/supply, and the intensity of imbalance grows with its absolute value. (We will later in this section provide a precise mapping between the state spaces of this auxiliary MDP and the original MDP of \eqref{eq:op1}.) In this auxiliary MDP, actions are unbounded and can take any real value. We also assume that the state transitions are generated as follows: if action $a_{t,g}$ is taken when there are $g$ periods left, then we have $\phi_{t,g-1} = \phi_{t,g} + X_{t,g} - a_{t,g}$, where $X_{t,g}$ is a random noise that represents the collective randomness in the ED arrival, ED workup, and hospital bed preparation processes. We assume that $\{X_{t,g}\}_{g = 1}^G$ is an exogenous and independent sequence that does not depend on the state and action history, but it can be non-stationary. For this specific finite-horizon MDP, the immediate cost after taking action $a_{t,g}$ in state $\phi_{t,g}$ is given by $\tilde C (\phi_{t,g}, a_{t,g}) := a_{t,g}^2 + C\phi_{t,g}^2 + D\phi_{t,g}$, where $C>0$ and $D$ are some real numbers. This cost structure resembles the generalized cost structure described in our original MDP formulation, with the distinction that it replaces the penalty for having boarding patients and prepared inpatient beds waiting with a quadratic function of $\phi_{t,g}$. Here, $C$ and $D$ quantify the relative weights of costs associated with boarding and inpatient-bed idling with respect to the penalty associated with variability in bed requests (which is normalized to one). 

Let $\tilde V_{t,g} ^{\pi_t} (\phi_{t,g})$ be the cost-to-go function, i.e., the sum of all future costs discounted to the current period under policy $\pi_t$ for this finite-horizon MDP that starts at time $t$ if the system is in state $\phi_{t,g}$ when there are $g$ periods left. Here, $\pi_t$ represents a policy that prescribes what action to take in each state at each decision epoch $t$ through $t+G$ when the decision problem starts at time $t$. We can then state the optimization problem as follows. For $g=0,1,\ldots, G$, we have
$$\tilde V_{t,g}^{\pi_t} (\phi_{t,g}) = \sum_{d=1}^{g} \rho^{g-d} \mathbb E[  \tilde C(\phi_{t,d}, a^{\pi_t}_{t,d})] +  \rho^{g}  \left( C \mathbb{E} [\phi_{t,0}^2] + D \mathbb{E} [\phi_{t,0}]\right),$$
\begin{equation}\tilde V_t^*(\phi_{t,G}) = \inf_{\pi_t} \tilde V_{t,G}^{\pi_t} (\phi_{t,G}), \label{opscs}
\end{equation}
where $ a_{t,d}^{\pi_t}$ denotes the action under policy $\pi_t$ when the decision problem starts at time $t$ and there are $d$ periods remaining. Similar to Proposition 3 of \cite{shi2021timing}, we can prove the following proposition for this finite-horizon problem with a special cost structure.
\begin{proposition}
\label{prop1}
For $t\geq 0$ and $g=1,2,\ldots,G$, the optimal action and the optimal value function for the finite-horizon problem~\eqref{opscs} at state $\phi_{t,g}$ are respectively given by 
\[\begin{aligned}
&\tilde a_{t,g}^* (\phi_{t,g}) = u_{t,g} \phi_{t,g} + w_{t,g},\\
&\tilde V_{t,g}^*(\phi_{t,g})=\eta_{t,g}\phi_{t,g}^2 + \theta_{t,g} \phi_{t,g} + \kappa_{t,g}, 
\end{aligned}\]   
where $u_{t,g}$, $w_{t,g}$, $\eta_{t,g}$, $\theta_{t,g}$, and $\kappa_{t,g}$ are some constants independent of state $\phi_{t,g}$. 
\end{proposition}

Proposition \ref{prop1} shows that the optimal action for this special MDP is linear and the value function is quadratic in state.  We can now use this solution to approximate $a^*(\mathbf s)$, i.e., the optimal action to be taken in state $\mathbf s$, and the value function $V(\cdot)$ for the infinite-horizon discounted MDP with optimality equation \eqref{eq:op1}. For this approximation, we first need to map the state spaces of the two MDPs. Let $\phi_{t, G}$ be the difference between the expected number of patients needing admission and the number of requested hospital beds that can be utilized at time $t$, i.e., set $\phi_{t, G} :=\sum_{i=1} ^ M \alpha_i  n_{i,t} + k_{t} - r_{t}$.  For $t\geq 1$, we can then approximate the optimal action in state $\mathbf s_t$ for the original MDP \eqref{eq:op1} as follows:
\begin{align}
     &a_t^* (\mathbf s_t)       = \arg\min_{a\in \mathcal{A}} \Big\{ a^2\hspace{-1pt} +\hspace{-1pt} \sum_{\mathbf s_{t+1}} p(\mathbf s_{t+1} | \mathbf s_{t}, a) \Big(\rho C_{b,e}(\mathbf{s}_{t+1}) \nonumber\\
     & + \tilde \alpha_{t,G} \phi_{t+1,G}^2 + \tilde \beta_{t,G} \phi_{t+1,G}\Big) \Big\},     \label{eq:3}
\end{align}
where $\tilde \alpha_{t,G}$ and $\tilde \beta_{t,G}$ are two constants that do not depend on any actions and states. We provide the derivation of Equation \eqref{eq:3} in the Online Supplement.

To use the above approximation in developing a heuristic solution, instead of tuning $C$, $D$, and $G$, and calculating $\tilde \alpha_{t,G}$ and $\tilde \beta_{t,G}$ recursively, we ignore the  dependency on time parameters $t$ and $G$, and directly tune for two constants $\tilde \alpha$ and $\tilde \beta$ to replace $\tilde \alpha_{t,G}$ and $\tilde \beta_{t,G}$ in Equation~\eqref{eq:3}. Moreover, since calculating the transition probabilities  $p(\cdot|\cdot, \cdot)$ can be difficult for high dimensional problems, we use a Monte Carlo simulation approach, where we randomly sample future states to evaluate the summation in Equation~\eqref{eq:3}. In implementation, the Monte Carlo simulation can incorporate more realistic features, such as the nonstationary arrival process used in our numerical study in Section~\ref{S7:NS}.

\subsection{Newsvendor Heuristic Policy (NHP)}
\label{S5:MDP:SS3}

With the goal of developing an interpretable heuristic, at any given decision epoch, suppose that the problem is going to be over after a finite horizon of $\psi$ periods, no new patients will arrive during this time, and hospital beds can be request only at the current decision epoch. Obviously, none of these assumptions are true in reality but prior research has shown that such a myopic approach may lead to successful solutions for a variety of sequential decision-making problems. Especially in supply chain literature, approximating a multi-period problem with a single-period newsvendor problem is a commonly used approach that works well~\citep{porteus}. We next follow a similar approach to develop a heuristic policy, which we call the Newsvendor Heuristic Policy (NHP).  

As in Section~\ref{S5:MDP:SS1}, we denote the state observed before making the request decision at decision epoch $t$ as $\mathbf s_t=(n_{1,t},\ldots,n_{M,t}, k_t, r_t)$. For this newsvendor formulation,  we define the random ``demand'' as the number of patients in the ED who will have an admission decision during the next $\psi$ periods, adjusted by subtracting the number of previously requested hospital beds that will become available during the same timeframe. When demand is defined this way, it can take values between $-r_t$ and $n_t$, where positive values indicate that there will be more patients receiving an admission decision than hospital beds that will be ready until the end of the planning horizon, and negative values indicate the opposite. 

To develop an easy-to-implement procedure to model and then estimate demand as a function of state $\mathbf{s}_t$, we assume that the demand is a continuous random variable with cumulative distribution function $F_{\mathbf s_t}(\cdot)$. For analytical tractability, we assume that $F_{\mathbf s_t}(\cdot)$ is differentiable over its domain with probability density function $f_{\mathbf s_t}(\cdot)$. We initially assume that this demand distribution is known for any state $\mathbf s_t$ to derive an expression for the optimal number of hospital beds to request. Later in this section, we discuss how to select a suitable demand distribution for this heuristic solution. 

Consistent with our original formulation, we consider three types of costs for this myopic approach. First, a requesting cost is incurred immediately after a decision, which is a quadratic function of the action. After the action is taken and the demand materializes, we incur either an {\em underage} cost of $\tilde{c}_b$ for each boarding patient or an {\em overage} cost of $\tilde{c}_e$ for each idle hospital bed. Since our goal is to reduce boarding for ED patients and idleness for prepared hospital beds, we assume that $\tilde{c}_b \geq 0$, $\tilde{c}_e \geq 0$, and $\tilde{c}_b + \tilde{c}_e >0$. 

In this newsvendor formulation, the inventory level is defined as the number of hospital beds that are ready to be used by ED patients. Hence, in state $\mathbf s_t=(n_{1,t},\ldots,n_{M,t}, k_t, r_t)$, the inventory level right before the action is taken is $-k_t$, and if the decision is to request $y\geq 0$ beds, then the inventory level right after replenishment will be $y-k_t$, assuming that the newly requested beds will become available instantaneously. Let $\tilde{C}(\mathbf s_t,y)$ denote the total cost incurred when the decision is to request $y\geq 0$ beds when the state is $\mathbf s_t$ for this newsvendor formulation. Then, we can write the expected total cost incurred for any $\mathbf s_t\in\mathcal{S}$ and action $y \geq 0$ as follows:
\begin{align}
     \mathbb {E}\left[\tilde{C}(\mathbf s_t, y)\right]=y^2 &+  \tilde{c}_b \int_{y-k_t}^\infty  (x-y+k_t) f_{\mathbf s_t}(x)dx \nonumber \\
  &+ \tilde{c}_e \int_{-\infty}^{y-k_t}  (y-k_t-x) f_{\mathbf s_t}(x)dx.   \label{eq:5}
\end{align}

\begin{proposition}
\label{prop2}
The optimal value of $y$ over the set of all non-negative real numbers that minimizes \eqref{eq:5} in state $\mathbf s_t$ is given by ${y}^* := \max(\tilde y+k_t, 0)$, where $\tilde y$ satisfies
\begin{equation}
\label{equ_nv_op}
F_{\mathbf s_t}(\tilde{y}) + \frac{2\tilde{y}}{\tilde{c}_b+ \tilde{c}_e} = \frac{\tilde{c}_b -2k_t}{\tilde{c}_b + \tilde{c}_e}.
\end{equation}
\end{proposition}

\begin{corollary}
\label{cor.1}
The optimal value $y^*$ that minimizes \eqref{eq:5} is equal to zero for all states with $F_{\mathbf s_t}(-k_t)\geq \tilde{c}_b/(\tilde{c}_b + \tilde{c}_e).$ Furthermore, $y^*$ never exceeds $\tilde{c}_b/2$.
\end{corollary}

Proposition \ref{prop2} provides the optimal requesting quantity for the newsvendor formulation with expected total cost given by $\eqref{eq:5}$ and known demand distribution $F_{\mathbf s_t}$. Corollary \ref{cor.1} shows that no new beds should be requested if $-k_t$, the inventory level, is sufficiently large. The same corollary also shows that the optimal requesting quantity should not exceed $\tilde{c}_b/2$ regardless of the demand distribution. This shows that the quadratic requesting cost in \eqref{eq:5} works as intended to curb the request amount to control for variation in request sizes. It is intuitive that this upper bound grows with $\tilde{c}_b$ because if the cost of boarding a patient grows then there is reason to request more beds early on.  

\begin{remark}
\textnormal{
For discontinuous $F_{\mathbf s_t}$,   the optimal solution in Proposition \ref{prop2} is the smallest value of $\tilde{y}$ that makes the left-hand side of \eqref{equ_nv_op} greater than or equal to the right-hand side. For a linear requesting cost, i.e., when the first term in Equation~\eqref{eq:5} is replaced with $c_{\ell}y$ for $c_{\ell}\geq 0$,  Equation~\eqref{equ_nv_op} reduces to the standard newsvendor solution 
$F_{\mathbf s_t}(\tilde{y})=(\tilde{c}_b -c_{\ell})/ (\tilde{c}_b + \tilde{c}_e)$, see, e.g., \cite{porteus}.}
\end{remark}

To obtain a heuristic policy using Proposition \ref{prop2}, we need a mechanism to estimate $F_{\mathbf{s}_t}$. Note that any patient who receives an admission decision over the next $\psi$ periods increases the demand by one, whereas any requested hospital bed that becomes ready during this time decreases it by one. Even when admission decisions and bed preparations occur independently of one another and the past, it is not possible to obtain a simple closed-form expression for the demand distribution. We can however obtain its mean and variance in state $\mathbf s_t$ as follows:  
\begin{align}
        \textrm{Mean} &= \sum_{m=1}^M n_{m,t}p_{m} \alpha_m - r_tq, \label{eq:mean}\\
        \textrm{Variance} &= \sum_{m=1}^M n_{m,t}p_{m}\alpha_m (1-\alpha_mp_{m}) + r_t q(1-q),\label{eq:var}    
\end{align}
where $p_{m}$ is the probability that a class-$m$ patient will receive a disposition decision in the next $\psi$ periods and $q$ is the probability that a requested hospital bed will be ready in the next $\psi$ periods. Assuming further that the time to disposition for class $m$ is exponentially distributed with rate $\mu_m$ and that the preparation time for a requested bed follows an exponential distribution with rate $\lambda$, we set $p_m=1-e^{-\mu_{m}\psi}$ and $q=1-e^{-\lambda\psi}$. We then approximate the demand distribution with a convenient theoretical distribution after matching the mean and variance to \eqref{eq:mean} and \eqref{eq:var}, respectively.   

We consider two theoretical distributions to approximate the demand distribution. Specifically,  we provide a closed-form expression for the solution to Equation~\eqref{equ_nv_op} when the demand follows a uniform and a triangular distribution in Propositions \ref{prop3} and \ref{prop4}, respectively.

\begin{proposition}
\label{prop3}
If the demand follows a uniform distribution over $[\delta - \xi,\delta + \xi]$ with $\xi\geq 0$, then $y^*$ that minimizes \eqref{eq:5} in state $\mathbf s_t$ is given by 

\begin{equation*}
y^* = 
\left\{
\begin{array}{l}
0, \mbox{ if } k_t\leq - \delta+\frac{\xi(\tilde{c}_e-\tilde{c}_b)}{\tilde{c}_b+\tilde{c}_e},\\
\frac{\tilde{c}_b}{2}, \textrm{ if}  \ k_t> \frac{\tilde{c}_b}{2} -\delta+\xi,\\
\frac{\tilde{c}_b(\delta+\xi)+\tilde{c}_e(\delta-\xi)+(\tilde{c}_b+\tilde{c}_e)k_t}{\tilde{c}_b+\tilde{c}_e+4\xi},\mbox{ otherwise.}
\end{array}
\right.
\end{equation*}
\end{proposition}

\begin{proposition}
\label{prop4}
If the demand follows a symmetric triangular distribution with mode $\delta$ and range $[\delta - \xi,\delta + \xi]$ with $\xi\geq 0$, then $y^*$ that minimizes \eqref{eq:5} in state $\mathbf s_t$ is given as follows. For $\tilde{c}_e\leq \tilde{c}_b$, 
\begin{equation*}
    y^* = 
\left\{
\begin{array}{l}
0, \textrm{ if } k_t\leq \xi\sqrt{\frac{2\tilde{c}_e}{\tilde{c}_b+\tilde{c}_e}} - \delta-\xi,\\
 \delta + \xi +k_t+ \frac{2 \xi^2-\xi \sqrt{ 4\xi^2 +2(\tilde{c}_e+\tilde{c}_b)(\tilde{c}_e + 2\xi  + \Delta_t)}}{\tilde{c}_b + \tilde{c}_e},\\  \mbox{ if}\ \xi\sqrt{\frac{2\tilde{c}_e}{\tilde{c}_b+\tilde{c}_e}} -\delta-\xi < k_t \leq \frac{\tilde{c}_b-\tilde{c}_e}{4} - \delta,\\
\delta - \xi +k_t -\frac{2 \xi^2-\xi \sqrt{4\xi^2 +2(\tilde{c}_e+\tilde{c}_b)(\tilde{c}_b + 2\xi  - \Delta_t)}}{\tilde{c}_b + \tilde{c}_e},\\ \mbox{ if}\ \frac{\tilde{c}_b-\tilde{c}_e}{4} - \delta < k_t \leq \frac{\tilde{c}_b}{2} - \delta+\xi,\\
\frac{\tilde{c}_b}{2}, \textrm{ if}  \ \frac{\tilde{c}_b}{2} -\delta+\xi < k_t,
\end{array}
\right.
\end{equation*}
where $\Delta_t=2(\delta+k_t)$, and for $\tilde{c}_e> \tilde{c}_b$, 
\begin{equation*}
y^* = 
\left\{
\begin{array}{l}
0, \textrm{ if } k_t\leq -\xi\sqrt{\frac{2\tilde{c}_b}{\tilde{c}_b+\tilde{c}_e}} - \delta+\xi,\\
\frac{\tilde{c}_b}{2}, \textrm{ if}  \ k_t>\frac{\tilde{c}_b}{2} - \delta+\xi,\\
\delta - \xi +k_t\\ \hspace{-2pt}-\frac{2 \xi^2-\xi \sqrt{4\xi^2 +2(\tilde{c}_e+\tilde{c}_b)(\tilde{c}_b + 2\xi  - \Delta_t)}}{\tilde{c}_b + \tilde{c}_e}, \mbox{ otherwise.}
\end{array}
\right.
\end{equation*}
\end{proposition}

Demand in our formulation is driven by the aggregation of multiple random components, including patient admissions and bed availability, and is therefore likely to be better captured by a bell-shaped distribution. A normal approximation would be a natural choice, but its cumulative distribution function does not lead to closed-form expressions for $y^*$ as in Propositions~\ref{prop3} and~\ref{prop4}. We therefore tested only uniform and triangular approximations in our numerical experiments. Because the symmetric triangular approximation with matched mean and variance, given by \eqref{eq:mean} and \eqref{eq:var}, performed better, we report only those results in Section~\ref{S7:NS}.


To implement this heuristic with the symmetric triangular distribution with mode $\delta$ and range $[\delta-\xi,\delta+\xi]$, at decision epoch $t>0$, we first obtain the mean and variance of the demand using \eqref{eq:mean} and \eqref{eq:var} based on the current state $\mathbf{s_t}$, and set them equal to the mean and variance of the triangular distribution, i.e., $\delta$ and $\xi^2/6$, respectively. Then, plugging the resulting $\delta$ and $\xi$ into the expression for $y^*$ in Proposition 4 yields a request quantity, which we round to the nearest integer to determine the action to take at this decision epoch. We then proceed to the next decision epoch, observe the new state, and continue with this process. Note that the user will need to tune three parameters: $\tilde{c}_b$, $\tilde{c}_e$, and $\psi$.  The cost parameter $\tilde{c}_b$ should be set to a value greater than one because otherwise the request quantity will be zero for all states by Corollary \ref{cor.1}. Ultimately, the specific values of $\tilde{c}_b$ and $\tilde{c}_e$ will be determined by how much hospital management values keeping patients versus available hospital beds waiting as well as how important variation in request quantity is. If variation in request quantity is not as important, then $\tilde{c}_b$ and $\tilde{c}_e$ should be much higher than one. The relative value given to boarding and keeping hospital beds empty would determine the ratio of $\tilde{c}_b$ to $\tilde{c}_e$.


\subsection{Deep $Q$-Learning (DQN) Heuristic}
\label{S5:MDP:SS4}

We next use reinforcement learning (RL) to obtain an alternative heuristic solution. Unlike an MDP formulation, RL can learn policies directly through repeated interactions with a realistic simulated environment \citep{sutton2018reinforcement}. This feature is well suited to our setting, where ED and hospital operations data can be used to construct a high-fidelity simulation environment. We implement the RL approach using deep $Q$-learning (DQN), which extends $Q$-learning~\citep{watkins1989learning} by approximating the $Q$-value function with an artificial neural network instead of storing $Q$-values for all state-action pairs. This is important because our problem has a high-dimensional state space. The network takes the current system state as input, outputs estimated $Q$-values for each feasible action, and is updated iteratively to improve these estimates \citep{lapan2018deep}.



To improve learning stability and performance, we use three standard DQN techniques. First, we use a technique called experience replay, in which the agent’s past experiences are stored in a replay buffer and mini-batches are sampled from this buffer during training~\citep{lin1992self}. This technique decreases the correlation among training samples and improves stability. Second, we use a target network with the same structure as the primary network but separately updated weights to stabilize learning by maintaining a steady target value \citep{mnih2015human}.  
Third, we adopt the double DQN model of \cite{van2016deep}, in which the primary network is used to select the action in the next state and the target network is used to evaluate the selected action. This decoupling of action selection and action evaluation reduces overestimation of action values and can improve performance. The target $Q$-value for state $\mathbf s$ and action $a$ can then be written as
\[
Q^* (\mathbf s, a) \hspace{-1pt}=\hspace{-1pt} C(\mathbf s, a)\hspace{-1pt} + \hspace{-1pt} \rho Q(\mathbf s^\prime, \arg\min_{a^\prime\in \mathcal{A}} Q(\mathbf s^\prime, a^\prime| \boldsymbol{\theta})| \boldsymbol{\theta}^\prime),
\]
where $\boldsymbol{\theta}$ and $\boldsymbol{\theta}^\prime$ denote the weights of the primary and target networks, respectively. 
The primary network is trained by minimizing the mean-squared error between its predicted $Q$-value and this target value. The weights of the target network are updated after a constant number of steps, using $\boldsymbol{\theta}^\prime \leftarrow (1-\tau) \boldsymbol{\theta}^\prime  + \tau \boldsymbol{\theta}$, where $\tau \in (0,1)$ is selected by the user. (We set $\tau$ to 0.02 in all numerical experiments discussed in this paper.)

To balance the trade-off between exploration and exploitation in our RL approach, we use a decaying $\epsilon$-greedy policy~ \citep{sutton2018reinforcement}. Specifically, at each period $t$, the agent chooses the greedy action with probability $1-\epsilon_t$ and randomly selects one of the non-greedy actions with probability $\epsilon_t$. Since the learned $Q$-value estimates become more reliable as the agent gains experience, we let $\epsilon_t$ decay exponentially over time. In the simulation experiments in Section~\ref{S7:NS}, we set the decay rate to 0.012 for every 6,000 hours, with an initial exploration probability of one.

To apply this RL approach, we construct a simulator of ED and hospital operations as the agent's environment and train the deep $Q$-learning agent to minimize the total discounted cost as defined in Section~\ref{S5:MDP:SS1}. At each step, the agent observes the state $\mathbf s_t=(n_{1,t},\ldots,n_{M,t}, k_t, r_t)$, takes action $a_t$, and observes the cost incurred $C(\mathbf s_t, a_t)$ as defined in Section~\ref{S5:MDP:SS1} and the next state $\mathbf s_{t+1}$. Since in both reality and our simulation environment presented in Section~\ref{S7:NS}, the arrival process, the service rates, and the preparation rates of hospital beds are non-stationary, we include indicator variables that represent four-hour time slots within a day as  additional inputs to the neural network. After the deep $Q$-learning agent converges, we then obtain a policy by feeding each state $\mathbf s$ into the neural network as input and choosing the optimal action with the minimum $Q$-value.

\section{Simulation Study}
\label{S7:NS}

The high-stakes nature of healthcare, along with the time and resources required, makes evaluating proposed policies through real-world implementation impractical. We therefore conduct a comprehensive simulation study based on patient encounter data collected in calendar year 2019 from a large ED in North Carolina, U.S.A. For this purpose, we use a dual-simulator framework (shown in Figure~\ref{fig:systemrepresentation}) that balances computational tractability with model realism, capturing the key ED and hospital components needed to generate actionable insights for future implementation rather than attempting to construct an exact replica of the system.

\begin{figure}[!ht]
    \centering
    \includegraphics[scale=0.5]{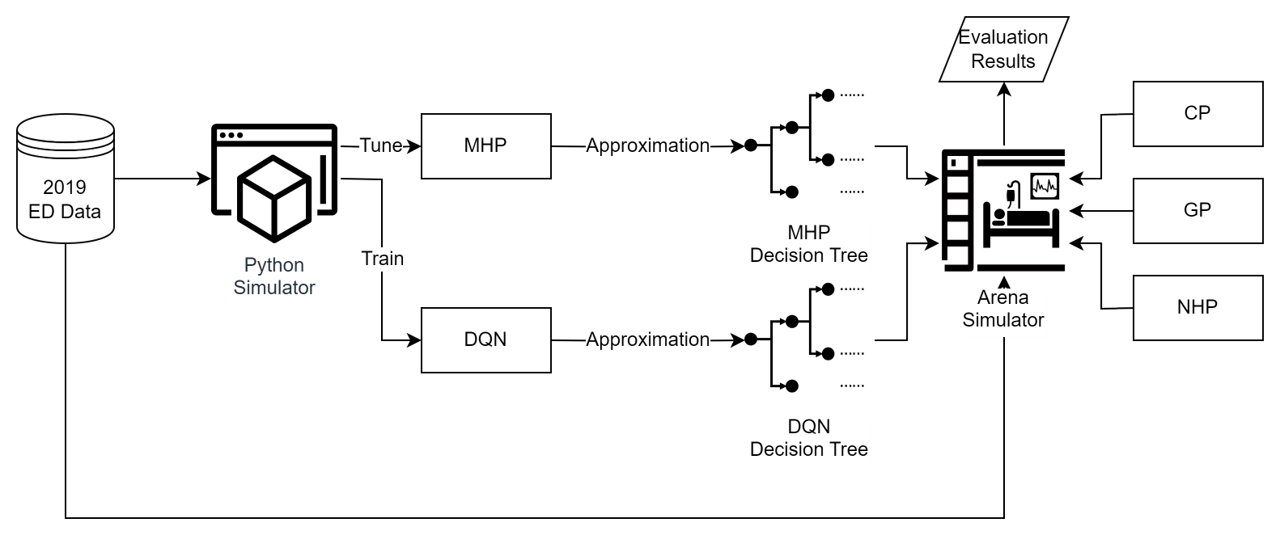}
    \caption{A graphical representation of the dual-simulator approach.}
    \label{fig:systemrepresentation}
\end{figure}

The first component of the dual-simulator is a lower-granularity (fast-execution) Python-based simulator, which is used for training the DQN algorithm and tuning MHP parameters. These tasks require repeated simulation runs and close integration with optimization and learning algorithms, and thus, computational efficiency is important. All modeling details and our input analysis based on ED data for this Phyton simulator are provided in Section~\ref{app:ps:se1} of the Online Supplement. The second component is a higher-granularity Arena \citep{Arena} discrete-event simulator, which we use as the evaluation testbed for comparing the proposed policies with benchmark policies. By capturing key features of ED-hospital operations, including non-homogeneous ED arrivals, multiple patient types, and service time distributions that vary by patient type, time of day, and day of week, it is much more realistic than the Python simulator.

\noindent \textbf{Arena simulator}: This simulator is a modified version of the one used in \cite{chen2022using}, which was based on the workflow of a large academic hospital’s ED, using data from calendar year 2012. We revised the model by redoing the input analysis using all patient encounters in calendar year 2019 for the same ED  and incorporating several operational changes implemented that year. This Arena simulator offers great detail, providing a realistic environment for policy evaluation. Specifically, the model defines twenty patient classes based on all combinations of five triage levels, two disposition categories (admission vs.\ discharge), and two age groups (adult vs.\ pediatric). Each class has a distinct, non-homogeneous Poisson arrival process with rates that vary by day of the week and time of the day. After triage, patients are further classified as having either a high or low probability of admission, based on their triage level and age group. (This classification enables implementation of the proposed policies in the simulation.) The simulated ED includes four care areas, each with dynamic bed capacity that changes throughout the day. After triage, patients wait for a bed in the appropriate care area based on their medical needs. 

Once roomed, each patient's remaining time in the ED is modeled in two stages. The first stage corresponds to the ED workup phase, during which the patient is examined and treated until a disposition decision is made. Workup time distributions vary by patient type and time of day. The ED disposition decision marks the end of the first stage and the beginning of the second. For discharged patients, the second stage corresponds to the discharge process, which includes all steps and procedures necessary to safely leave the medical facility. For admitted patients, it corresponds to the boarding process. At the top of each hour, a decision maker determines how many hospital beds to request, based on the policy used. Following the methodology in \cite{shi2016models}, we assume that the preparation time for a hospital bed depends on both hospital bed availability and the duration of the transfer preparation process. When a patient with an admission decision completes their ED service, they are immediately transferred if a previously requested and fully prepared bed is available. Otherwise, they board in the ED until all earlier boarding patients have been roomed and a new hospital bed becomes available.
To better reflect actual ED and hospital protocols, the Arena simulator includes the following adjustment: as soon as an admission decision is made, if the number of boarding patients waiting for early-requested beds exceeds the number of outstanding early bed requests, a regular bed request is initiated immediately without waiting for the next hourly decision point. (Throughout this paper, ``outstanding early bed requests" refers to hospital beds that have been requested in advance but are either still being prepared or are ready but not yet occupied.) Further details on the probability distributions used, the incorporation of early bed requests under various policies, and the validation of the Arena simulation model can be found in Sections~\ref{app:as:se1} through~\ref{app:as:se4} of the Online Supplement.

\noindent {\bf Approximation by decision trees:} Because the tuned MHP and trained DQN policies depend heavily on the system state and experimental conditions, we approximate them using decision trees, which are easier to implement in Arena and provide transparent decision logic that may facilitate practical acceptance. Specifically, after tuning MHP and training DQN with the Python simulator, we collect state-action pairs under each policy, with 1.2 million observations per heuristic and scenario. We then train cost-complexity-pruned decision trees  \citep{hastie2009} using the system state as the input and the heuristic-recommended action as the response. 
Each resulting tree partitions the state space into leaves, and all states within a leaf are assigned the majority-vote action from the original policy. To balance simplicity and fidelity, we restrict tree sizes to 150--300 nodes.

\noindent {\bf Performance measures:} The primary goal of this study is to reduce boarding times for admitted patients and, by freeing ED beds sooner, improve waiting times for all ED patients. We therefore use the {\em long-run average ED patient length of stay (LoS)} as the primary performance measure and also report the {\em long-run average boarding time for admitted patients}. To assess the operational burden on the hospital, we report the {\em long-run average idle time of hospital beds requested by the ED}, defined as the time from when a requested bed finishes preparation until it is occupied by an ED-admitted patient. Finally, to evaluate the smoothness of the bed-request process, we report the {\em coefficient of variation of total hourly bed requests} including both early and regular requests. Lower values indicate a more uniform request process. We do not report the average cost used in the mathematical formulation because the cost parameters serve only to calibrate the relative importance of operational metrics rather than represent actual penalties in practice.

\noindent {\bf Experimental conditions:} Our simulation study considers three scenarios. The Standard Operations Scenario (SOS) reflects the operating conditions observed in the study ED and hospital during calendar year 2019 (Section~\ref{S7:NS:SS1}). SOS-P uses the same experimental setup as SOS but assumes that only a portion of hospital beds are interchangeable  (Section~\ref{S7:NS:SS2}).  Finally, the Pandemic Response Scenario (PRS) models a six-week surge in ED arrivals, representing a hypothetical pandemic-like event in the study ED and hospital during 2019 (Section~\ref{S7:NS:SS3}).

\subsection{Standard Operations Scenario (SOS)}
\label{S7:NS:SS1}
 
In this simulated scenario, the ED and the hospital operate under the conditions of the actual system in calendar year 2019, as detailed in Sections~\ref{app:as:se1} through \ref{app:as:se4} of the Online Supplement. To obtain confidence intervals (CIs) on the steady-state performance measure for each proposed and benchmark policy, we used the batch means method for steady-state output analysis. In particular, we ran a simulation of length 1,000 years following a warm-up period of 366 days under each policy and experimental setting.

In this simulation study, we compare the proposed policies against the two benchmark policies introduced in Section \ref{S4:PFBP}, namely, the current practice (CP) and the greedy policy (GP). To implement GP in our experiments, we classify patients into two equal-proportioned types upon arrival based on an admission prediction tool adapted from \cite{mehrotra2017}: type I (a low admission probability of 0.0848) or as type II (a high admission probability of 0.577). We also consider a third benchmark policy, which we call the extreme policy (EP). Under this policy, the number of early bed requests placed each hour equals the surplus of current ED patients over any outstanding early hospital bed requests. This is clearly an impractical policy as it essentially assumes that all patients will be admitted, but it helps us explore the theoretical limits of performance of early-bed-request policies.

Table \ref{tab:simulation_results} provides  performances of these three benchmark policies. 
As expected, EP results in the smallest average LoS and boarding times but the highest average hospital bed idle time. On the other hand, CP has the longest LoS and boarding times without any bed idling time, whereas GP performs in between CP and EP, providing a compromise between bed idling and LoS/boarding times. The results on the performance in terms of coefficient of variation are also consistent with intuition: Under CP, the request size for hospital beds is one most of the time, but under GP and EP, a request for beds is given according to the expected number of patients to be admitted and the outstanding bed requests, which leads to more variable request sizes.

\begin{table}[ht]
\centering
\caption{Simulation results for three benchmark policies under SOS. Metrics are reported as mean $\pm$ $95\%$ CI half-width.}
\begin{small}
\begin{tabular}{||c||c|c|c|c||}
\hline
\textbf{Policy} & \textbf{ED LoS (min)} & \textbf{ Boarding Time (min)} & \textbf{CoV of Request Size} & \textbf{Bed Idle Time (min)} \\
\hline\hline
CP & 406.57 $\pm$ 3.27 & 248.26 $\pm$ 2.14 & 0.76 $\pm$ 0.0004 & 0  \\
GP & 349.00 $\pm$ 2.50 & 95.81 $\pm$ 2.47 & 0.89 $\pm$ 0.0050 & 34.56 $\pm$ 0.21 \\
EP & 342.17 $\pm$ 4.64 & 76.68 $\pm$ 4.98 & 1.29 $\pm$ 0.0007 & 356.73 $\pm$ 5.06 \\
\hline
\end{tabular}
\end{small}
\label{tab:simulation_results}
\end{table}

We next compare the performance of the proposed policies against these benchmarks under a variety of tuning parameters representing a spectrum of managerial preferences. In particular, we choose $(c_b, c_e)$ from the condition set $ \mathcal{C}_1:= \{(10,1),\allowbreak (4,1),\allowbreak (1,1), \allowbreak (1,4),\allowbreak (1,10)\}$ to train DQN with the objective of finding an approximate solution to the optimization problem stated by Equation \eqref{eq:optimal}. Similarly, for MHP, we tune $\tilde \alpha$ and $\tilde \beta$ to minimize the total discounted cost corresponding to these $(c_b, c_e)$ combinations. The tuning parameters $(\tilde{c}_b, \tilde{c}_e)$ used in NHP are correspondingly adjusted to ensure that the generated performance curves for average bed idling times versus average length of stay from different policies (see Figure~\ref{fig:R1}) fall within approximately the same range.  After experimenting with different values of $\psi$, which is the other tuning parameter for NHP, we set $\psi=8$, corresponding to looking eight hours into the future when estimating demand. This value was chosen because it yielded the smallest total discounted cost in most scenarios. (All values of tuning parameters in addition to detailed simulation results discussed in this section are provided in  Section \ref{sec.results-tables} of the Online Supplement).

\begin{figure}[!ht]
    \centering
    \includegraphics[width=1\linewidth]{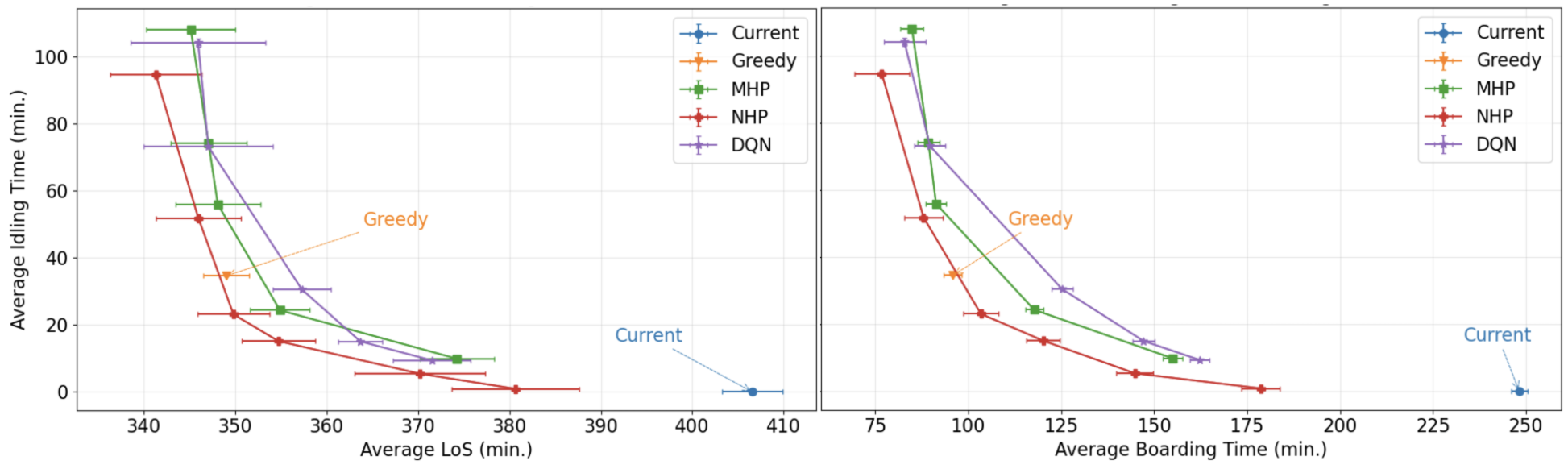}
    \caption{Long-run average ED LoS and boarding time versus long-run average hospital bed idling time under SOS. (CI half-widths on average bed idling times are too small to be visible.)}
    \label{fig:R1}
\end{figure}

Figures~\ref{fig:R1} and \ref{fig:R2} present 95\% CIs for performance measures of interest for all policies under consideration estimated using output data  from the Arena simulator. (We excluded EP from these plots as it resulted in an average bed idling time and a coefficient of variation that are far outside the range of these plots; see Table~\ref{tab:simulation_results}.) Each point on the DQN curve corresponds to one of the five combinations of $(c_b, c_e)\in \mathcal{C}_1$, where $c_b/c_e$ decreases from left to right in each plot. Similarly, each point on the NHP and MHP curves corresponds to similar operating conditions in $\mathcal{C}_1$. As expected, MHP, NHP, and DQN curves depict an ``elbow'' pattern in Figure \ref{fig:R1}, where the long-run average hospital bed idling time increases as the long-run average ED LoS and boarding times decrease. 

The first observation from Figure~\ref{fig:R1} is that, compared with the status quo represented by CP, the proposed framework of aggregate early bed requests significantly lowers the average boarding time from approximately 248 minutes to a range between 77 -- 179 minutes in exchange for a long-run average hospital bed idling time of 95 minutes to a minute. This percentage decrease of 30--70\% in average boarding times leads to significant improvements in the overall ED patient LoS, with reductions ranging from approximately 26 to 65 minutes, i.e., by 6--16\%. The proposed policies can also nearly match EP in terms of the average ED LoS (around 342 minutes), while maintaining a much shorter long-run average hospital bed idling time.  Another interesting observation from Figure~\ref{fig:R1} is that the GP, a simple policy that requests beds in the amount of  expected number of ED admissions, exhibits a reasonable performance, with a reduction of roughly 58 minutes in long-run average ED LoS, and a 152 minute reduction in the average boarding time, while keeping the long-run average hospital bed idling time around 35 minutes. However, as GP does not possess any ``cost'' parameters that can be adjusted, it is not possible to  tune the policy to achieve the level of balance an ED management would desire.

Figure~\ref{fig:R1} also provides insight into the relative performance of the proposed policies. MHP and DQN perform similarly in most cases, while NHP generally provides a better tradeoff between LoS/boarding times and bed idling, although the differences are not always statistically significant. The difference in bed idling times between NHP and the other two is larger/smaller in the range where the average ED LoS is shorter/larger.

\begin{figure}[!ht]
    \centering
    \includegraphics[width=1\linewidth]{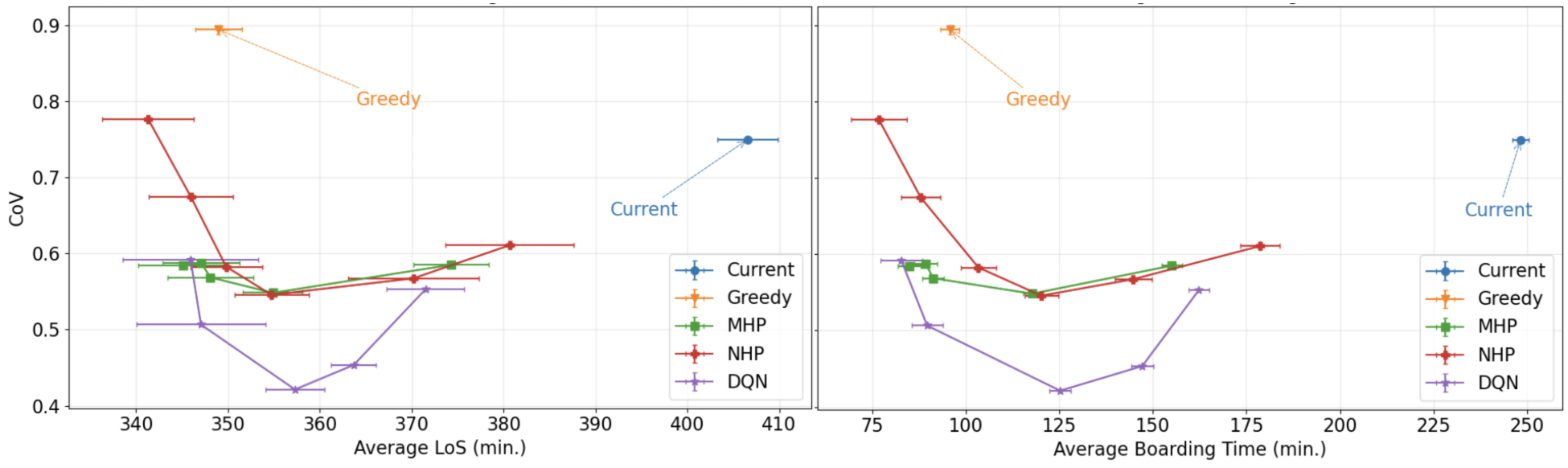}
    \caption{Long-run average ED LoS and boarding time versus long-run average coefficient of variation under SOS. (CI half-widths on average coefficient of variation are too small to be visible.)}
    \label{fig:R2}
\end{figure}

Figure~\ref{fig:R2} illustrates the relationship between the long-run average coefficient of variation of bed requests versus the long-run average ED LoS and boarding time under all policies considered. As anticipated, within a comparable range of average ED LoS, MHP, NHP, and DQN, all of which take into account the variation in the size of early bed requests, exhibit lower coefficients of variation when compared with the benchmark policies. This suggests that our proposed techniques are capable of reducing average ED LoS and boarding time while ensuring a smoother process for early bed requests. Among these methods, DQN outperforms others in terms of achieving lower coefficients of variation for similar performance in ED LoS and boarding times. The superior performance of DQN is likely attributable to its more sophisticated structure that allows it to anticipate and adapt to future demand more effectively. The adaptability to future demand is expected to be more important to achieve a superior performance when the performance measure is nonlinear such as the coefficient of variation. Finally, the U-shaped structure of the CoV plots for each proposed heuristic is likely due to the fact that parameters that correspond to the two ends of the spectrum for each proposed heuristic, are more variable than any of the middle-ground parameter settings. More specifically, each end point on the curves either corresponds to a policy that requests beds conservatively (rarely for patients who are not yet admitted) due to higher cost given to keeping hospitals bed idle or corresponds to an aggressive policy that requests beds early for many patients before their admission decisions due to higher cost given to boarding times. These two extreme settings are likely to lead to more variable request quantities than the middle ground.    

Under SOS, we also test conditions where the variation in request sizes is less of a concern compared to keeping patients boarded or hospital beds idle for extended periods of time.  In particular, we choose the tuning parameters $(c_b, c_e)$ for MHP and DQN  from  $\mathcal{C}_2:=\{(100,10),\allowbreak (40,10),\allowbreak (10,10), \allowbreak (10,40),\allowbreak (10,100)\}$, and adjust $(\tilde{c}_b, \tilde{c}_e)$ used in NHP accordingly for comparable LoS vs.\ bed idling times curves. In the interest of space, we provide these results in Section \ref{sec.results-tables} of the Online Supplement. Based on these results, we conclude that with appropriate parameter tuning for each proposed heuristic, comparable performance in terms of average LoS and bed idling can be achieved, while also attaining a smoother and more stable request process.

\subsection{Standard Operating Scenario with Partial Bed Requests (SOS-P)}\label{S7:NS:SS2}
In this section, we use the same experimental setup of Section \ref{S7:NS:SS1}, except that only $70\%$ of admitted patients are now eligible for interchangeable hospital beds, e.g., those in general medicine wards.  This modification relaxes the assumption, used in the mathematical formulation underlying the proposed policies, that all hospital beds are interchangeable. To simulate this setting, we reduced the number of beds requested by $30\%$ under each heuristic, rounding to the nearest integer. 

\begin{figure}[h]
    \centering
    \includegraphics[width=1\linewidth]{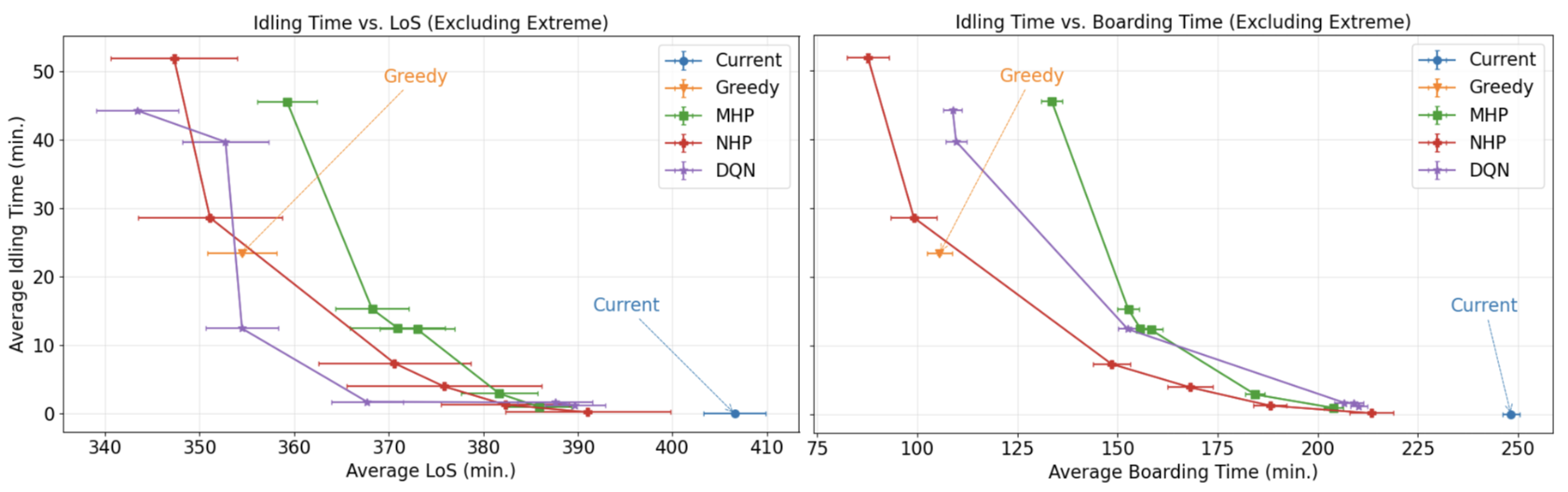}
    \includegraphics[width=1\linewidth]{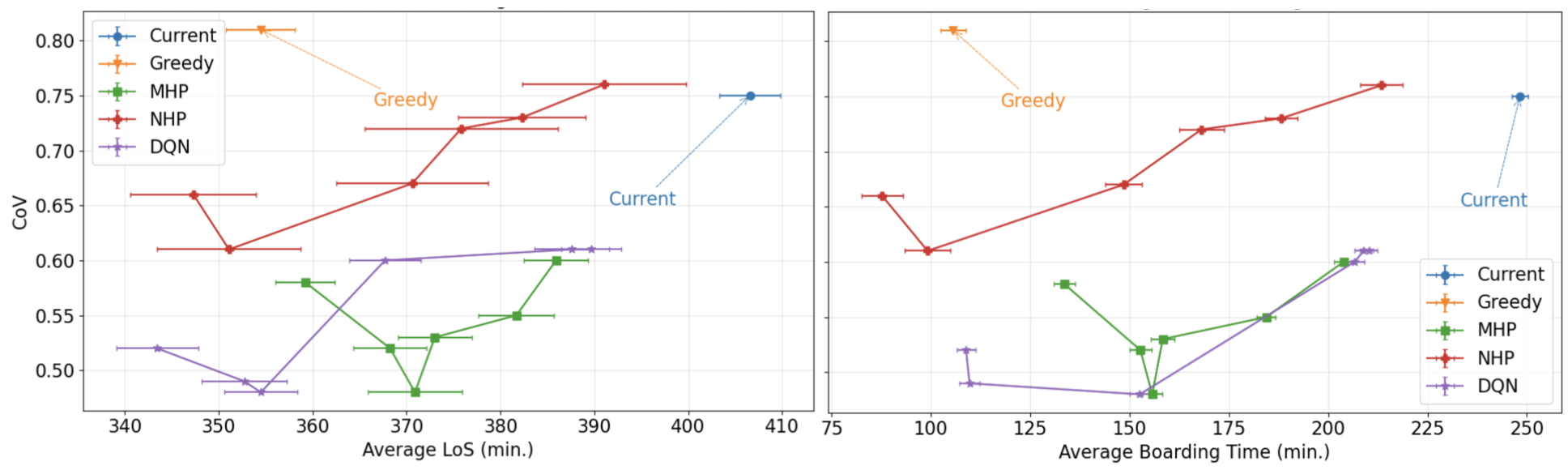}
    \caption{Long-run average ED LoS and boarding time versus long-run average bed idling time and coefficient of variation under SOS-P.}
    \label{fig:70pct}
\end{figure}

Figure~\ref{fig:70pct} presents $95\%$ CIs for performance measures of interest for all policies under consideration estimated using simulation output. (Values of tuning parameters in addition to detailed simulation results for all policies under this experimental setting are provided in Section \ref{sec.results-tables} of the Online Supplement.)  We observe that the proposed framework of early bed requests is still effective in achieving reductions in boarding times and overall LoS with reasonable increases in hospital bed idling times when $70\%$ of admitted patients are eligible for an interchangeable hospital bed. In particular, the average ED patient LoS reduces by approximately 15 to 55 minutes, i.e., by 4--14\%, with an increase of one to 50 minutes in average bed idling time. NHP and DQN can also nearly match EP in terms of the average ED LoS (around 342 minutes), while maintaining a much shorter long-run average hospital bed idling time. GP is still providing a reasonable compromise, with a reduction of roughly 50 minutes in the long-run average ED LoS for a long-run average hospital bed idling time of 24 minutes. Among the proposed heuristics, NHP continues to perform the best in terms of average boarding time but in terms of the overal LoS, the DQN and NHP perform similarly. Figure~\ref{fig:70pct} also shows that DQN still outperforms others in terms of achieving a lower coefficient of variation for similar performances in ED LoS and boarding times.

\subsection{Pandemic Response Scenario (PRS)}
\label{S7:NS:SS3}

While our analyses in Sections \ref{S7:NS:SS1} and \ref{S7:NS:SS2} have demonstrated that the policies we proposed work well in shortening the ED LoS while maintaining a reasonably low average idle time for hospital beds under regular operating conditions, it is equally crucial to evaluate their robustness under more extreme circumstances. Therefore, we next turn our attention to a scenario in which the ED is exposed to higher-than-usual arrival volumes that lead to periods of overcrowding. Such a scenario resembles a winter viral infection surge, like the recent RSV/Flu/Covid-19 surges that occurred every winter since the Covid-19 pandemic. 

The design of this pandemic response scenario (PRS) is primarily based on the FluSurge 2.0 tool developed by the Centers for Disease Control and Prevention (CDC) \cite{flusurge}, which was also used in \cite{chen2022using}. FluSurge 2.0 is a predictive tool designed to assist EDs and hospitals in forecasting resource needs during an influenza pandemic. Following this tool, we simulate a hypothetical flu surge scenario of length 18 weeks and with the following surge parameters: During the first six weeks of this period, the ED and hospital operate under standard conditions as assumed in SOS. Starting in the seventh week, a pandemic begins and it lasts for six weeks. We assume that the majority of infected patients will be classified into either ESI level 2 or 3. Therefore, for a given day $d \in \{1,2,\ldots, 21\}$ in the first three weeks of the outbreak, we assume that the arrival rates of these two classes of patients will be multiplied by a factor of $1.015^{d}$. As the pandemic progresses, in the remaining three weeks of the pandemic, the arrival rates of these two classes of patients will gradually return to the pre-pandemic levels. Specifically, for each day $d \in \{22,23, \ldots, 42\}$ in the last three weeks of the outbreak, we assume that the arrival rates of these two classes of patients will be multiplied by a factor of $1.015^{43-d}$. This approach mimics the common trajectory of a flu pandemic while preserving the inherent weekly pattern in the ED arrival process. Following the outbreak period, the simulation continues for an additional six weeks during which the arrivals to the ED stay under standard operating conditions.

To evaluate the robustness of the developed policies under extreme situations, we continue to employ the policies derived under the conditions of SOS and parameter set $\mathcal{C}_1$ to make inpatient bed requests. Nevertheless, we hypothesize that the hospital bed capacity will increase to 400, up from the baseline of 337 in SOS, five days after the pandemic starts and then return to the normal level five days after the pandemic ends. This assumption is rooted in real-life hospital operations, which often involve creating temporary care areas and canceling elective surgeries to create surge capacity during pandemic conditions. The five-day delay is implemented to mimic the lag in healthcare system response in such scenarios.  The ED capacity is assumed to remain constant during the 18 weeks of pandemic conditions. For this specific analysis, we use the replication-deletion method with 1,000 replications and a warm-up period of 366 days.

\begin{figure}
    \centering
    \includegraphics[width=1\linewidth]{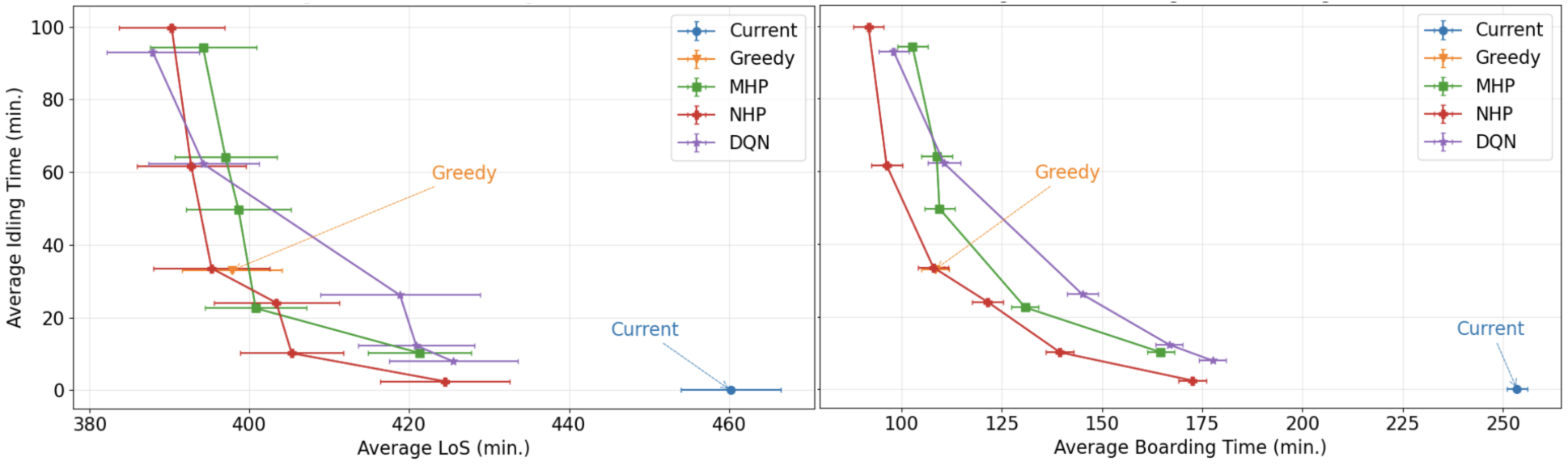}
    \caption{Average ED LoS and boarding time versus average hospital bed idling time for the pandemic scenario.}
    \label{fig:R4}
\end{figure}

Because this analysis focuses on policy responsiveness during a brief period of extreme conditions, request-size variability is no longer a priority. We therefore evaluate the policies using average boarding time, average ED LoS, and average hospital bed idling time, as shown in Figure~\ref{fig:R4}. The conclusions are broadly consistent with those under standard operating conditions. All proposed policies substantially outperform CP in reducing average ED LoS and boarding time, suggesting that they remain effective during brief periods of extreme conditions. GP remains competitive, while the adjustable policies (MHP, NHP, and DQN) exhibit the expected elbow-shaped frontier. NHP still provides the most efficient trade-off between average bed idling and boarding times, although its advantage is less clear for average ED LoS.

\section{Conclusion}
\label{S8:C}

In this paper, we propose a novel framework that uses admission probability and time-to-disposition predictions for ED patients to periodically generate aggregate early inpatient bed requests. Although our approach is not the first to employ admission probability predictions to reduce boarding times, {\em it is novel in that it takes a collective view of all patients in the ED and the state of the hospital, and uses this information to make periodic ED-level decisions to request inpatient beds}. As we demonstrate through a simulation study based on data from our partner hospital, this approach can reduce average boarding times for admitted patients by $30-70\%$ and average length of stay (LoS) for all ED patients by $6-15\%$, while imposing only a modest burden on the hospital in the form of low average idle time for prepared inpatient beds.  This study is also among the first to explicitly  consider the smoothness of the bed-requesting process by favoring policies that reduce  fluctuations in the number of requested  beds, which may help reduce delays in downstream hospital processes.

Based on this decision framework, we propose three policies that use admission probability and time-to-disposition predictions for ED patients. Among these, a newsvendor-type heuristic (NHP) generates the most efficient frontier in balancing ED operational performance (measured by shorter average LoS and boarding times) with hospital operational performance (measured  by lower average idle times for prepared inpatient beds). NHP is also expected to be more appealing to healthcare managers because of its interpretability and computational simplicity. Our two other heuristics, based on Markov decision processes and reinforcement learning (RL), require more complex tuning and are less easily interpretable. Nevertheless, if variability in request size is an  important managerial concern, our simulation study shows that the RL-based heuristic can be tuned to achieve a substantially smoother bed-requesting process than the other heuristics, while achieving similar average inpatient bed idle times and ED lengths of stay. These findings also contribute to the growing literature on RL in healthcare operations management \citep{Wu2025} and suggest that RL methods may be useful when adaptability is an important objective. 


Although we evaluated the performance of the proposed policies in a realistic discrete-event simulator developed from actual ED data, further research is needed through controlled studies in practice. While implementation would require nontrivial organizational and operational changes, frameworks such as ours may be increasingly feasible given the widespread adoption of predictive modeling tools and other analytical tools for policy evaluation and operational planning. 

Several extensions of the proposed framework merit future study, including relaxing our assumption that all inpatient beds requested by the ED are interchangeable. Our simulation results suggest that the framework can still have a sizable impact when applied to a large single admitting unit, such as general medicine, where beds are interchangeable. A natural extension is to generalize the framework to multiple such units, e.g.,  separate general medicine units for pediatric, adult, and geriatric patients. Another promising direction is to develop hybrid approaches that use aggregate bed requests for large general admitting units while relying on individualized early bed requests, as in \cite{chen2022using}, for highly specialized units. Finally, in our simulation study, we used a simple admission prediction tool by \cite{mehrotra2017} to categorize patients into two admission levels. It would be interesting to extend this numerical evaluation by using more sophisticated prediction tools with a more granular categorization of patients and higher prediction accuracy.

\bibliography{abert_bib}

\newpage
\section*{Online Supplement}
\renewcommand{\thesection}{OS.\arabic{section}}
\setcounter{section}{0}
\setcounter{page}{1}
\input{SSRN_appendix}

\FloatBarrier

\end{document}

%% file: SSRN_appendix.tex
\section{Proofs of Theoretical Results}

\noindent {\bf Proof of Proposition~\ref{theo:1}:}
This proof uses theorems and propositions from \cite{puterman2014markov}. According to Theorem 6.10.4 of \cite{puterman2014markov}, since the state space is countable for our problem, if Assumptions 6.10.1 and 6.10.2 of \cite{puterman2014markov} hold, then there exists a unique solution to the optimality equation, and if there exists a policy $\pi^*$ that chooses the action that minimizes the right-hand side of optimality equation~\eqref{eq:op1}, then $\pi^*$ is the optimal policy that is associated with optimal value function $V(\mathbf s)$.

To prove that Assumptions 6.10.1 and 6.10.2 hold, suppose the current state is denoted by $\mathbf s=(n_{1}, \ldots,\allowbreak  n_{M},\allowbreak k, r)$, and the next state is $\mathbf j = (n_{1}^\prime, \allowbreak \ldots,n_{M}^\prime, \allowbreak  k^\prime,\allowbreak  r^\prime)$. To show that Assumption 6.10.1 holds, we next find a positive real-valued function $w(\mathbf s)$ on $\mathcal{S}$ that satisfies $\inf_{\mathbf s \in \mathcal{S}} w(\mathbf s)>0$ so that $\sup_{a \in \mathcal{A}} |C(\mathbf s,a)|  \leq  w(\mathbf s)$. 
Define $c_{max} = \max(c_b, c_e)$ and $w(\mathbf s) = A^2 +  c_{max} \left( \sum_{m=1}^M n_{m} + |k| + |r| \right)$ for all $\mathbf s \in \mathcal{S}$. It can be then easily seen that Assumption 6.10.1 holds because
    $$\sup_{a \in \mathcal{A}} |C(\mathbf s,a)| = \sup_{a\in \mathcal{A}} \left\{ a^2 + c_b \max(0,k) + c_e \max(0,-k)  \right\}\leq  A^2 + c_{max} |k| \leq w(\mathbf s).$$
To show that Assumption 6.10.2 also holds, we use part (a) of Proposition 6.10.5 in \cite{puterman2014markov}. This proposition states that if there exists a constant $L>0$, for which 
$$\sum_{\mathbf j \in S} p(\mathbf j| \mathbf s, a) w(\mathbf j) \leq w(\mathbf s) +L, \mbox{ for } a \in \mathcal{A},\mathbf s \in \mathcal{S},$$
then Assumption 6.10.2 holds. We now let $N$ be the number of new arrivals during a period. By the Markovian assumption, $N$ is independent of everything else in the system. Define $p(\mathbf j, N|\mathbf s, a)$ as the probability that there are $N$ new arrivals and the next state is $\mathbf j$, given the current state is $\mathbf s$ and the action taken is $a$. We then have

\begin{align*}
    \sum_{ \mathbf j \in \mathcal{S}} p(\mathbf j|s,a) w(\mathbf j)  &\leq \sum_{\mathbf j \in \mathcal{S}, N\geq 0} p(\mathbf j, N | \mathbf s, a) \left[ A^2 +  c_{max} \left( \sum_{m=1}^M n_{m} + N + |k| + |r| + a \right)\right]\\
    &\leq \sum_{\mathbf j \in \mathcal{S}, N\geq 0} p(\mathbf j, N | \mathbf s, a) \left[ A^2 +  c_{max} \left( \sum_{m=1}^M n_{m} + |k| + |r|\right)  \right]\\
    &\qquad +  c_{max} \sum_{\mathbf j \in \mathcal{S}, N\geq 0} p(\mathbf j, N | \mathbf s, a) N +c_{max} A \qquad \quad (\textrm{because}\ a\leq A)\\
    &=c_{max} A + w(\mathbf s) + c_{max}\hspace{-4pt}\sum_{\mathbf j \in \mathcal{S}, N\geq 0}  N p(\mathbf j, N | \mathbf s, a)\\
    &=c_{max} A + w(\mathbf s) + c_{max}\mathbb{E}[N]\\
    &= w(\mathbf s) + L, 
\end{align*}
where $L=c_{max} (A + \mathbb{E} [N])$ is a finite, positive constant.

\noindent\textbf{Proof of Proposition~\ref{prop1}:}
The proof follows an induction argument, starting with the case when there is only one period left and then extending to the general case when there are $q\geq 1$ periods left. For notational convenience,  we let $S_{t,g}=\mathbb{E}[X_{t,g}]$ and $S_{t,g}^2=\mathbb{E}[X_{t,g}^2]$. We also 
define, for $t=0,1,\ldots$ and $g=1,2,\ldots,G$,
\[\begin{aligned}
&u_{t,g} = \frac{\eta_{t,g-1}}{\eta_{t,g-1} + \rho^{-1}},\\
   &w_{t,g} = \frac{2\eta_{t,g-1} \mathbb E[X_{t,g}] + \theta_{t,g-1}}{2(\eta_{t,g-1} + \rho^{-1})},\\
   &\eta_{t,g}= u_{t,g} + C,
   \theta_{t,g} = 2w_{t,g} + D,\\
   &\kappa_{t,g}=\rho \kappa_{t,g-1} +\left(\eta_{t,g-1} \mathbb E[X_{t,g}^2] + \theta_{t,g-1} \mathbb E[X_{t,g}]\right. \\ &\qquad \quad \left.+ \rho\eta^2_{t,g-1}\mathrm{Var}(X_{t,g}) - \rho \theta_{t,g-1}^2/4 \right) /\left(\eta_{t,g-1} + \rho^{-1}\right),\\
   &\eta_{t,0}=C, \theta_{t,0}=D, \mbox{ and } \kappa_{t,0}=0.\end{aligned}\]

When there is only one period left, the optimal value function $\tilde V_{t,1}^*(\phi_{t,1})$ is given by 
\[
\begin{aligned}
       &\inf_{a\in\mathbb{R}} \left\{a^2 +C\phi_{t,1}^2 + D\phi_{t,1} + \rho \mathbb{E}\left[C\phi_{t,0}^2 + D\phi_{t,0} \right] \right\}\\
        &= \inf_{a\in\mathbb{R}} \left\{ a^2 + C\phi_{t,1}^2 + D \phi_{t,1}+   \rho D (\phi_{t,1} + S_{t,1} -  a) \right. \\  
         & \qquad \quad \left.+ \rho C(\phi_{t,1}^2 + 2\phi_{t,1} S_{t,1}+  S_{t,1}^2-2 a (\phi_{t,1} + S_{t,1})+ a^2) \right\}\\
        &= \inf_{a\in\mathbb{R}} \left\{(\rho C+1)a^2 - \rho (2C\phi_{t,1} + 2C S_{t,1} + D) a  \right.\\
          & \qquad \quad + (\rho+1) C \phi_{t,1}^2 + 2\rho C \phi_{t,1} S_{t,1} + \rho C S_{t,1}^2\\ &\qquad \quad +\left. (\rho+1) D \phi_{t,1} + \rho D S_{t,1} \right \}.
\end{aligned}
\]
Since the function inside the above infimum is a quadratic function of $a\in\mathbb{R}$, the optimal action $\tilde a_{t,1}^*$ is given by
$$
        \tilde a_{t,1}^*=\frac{ \rho (2C\phi_{t,1} + 2 C S_{t,1} +D)} {2(\rho C+1)}  = \frac{C}{C+\rho^{-1}}\phi_{t,1} + \frac{2CS_{t,1} + D}{2(C+\rho^{-1})} =u_{t,1} \phi_{t,1} + w_{t,1}.
$$
Plugging the optimal action $\tilde a_{t,1}^*$ back into the value function, we have
\[
\begin{aligned}
    &\tilde V_{t,1}^*(\phi_{t,1})\hspace{-1pt} =\hspace{-1pt}[(\rho C+1)u_{t,1}^2\hspace{-1pt} -\hspace{-1pt}2 \rho C u_{t,1} +(\rho+1) C]\phi_{t,1}^2  \\
    &\quad \qquad \qquad  +[2(\rho C+1)u_{t,1} w_{t,1} - \rho (2C S_{t,1} + D) u_{t,1}\\ 
    &\quad \qquad \qquad  - 2 \rho C w_{t,1} + 2 \rho C S_{t,1} + (\rho+1)D]\phi_{t,1} + \rho D S_{t,1}\\
    &\quad \qquad \qquad + (\rho C+1)w_{t,1}^2 - \rho (2C S_{t,1} + D) w_{t,1} + \rho C S_{t,1}^2  \\
    &= \left( \frac{\eta_{t,0}}{\eta_{t,0} + \rho^{-1}} + C \right) \phi_{t,1}^2 + \left( \frac{2\eta_{t,0} S_{t,1} + \theta_{t,0}}{\eta_{t,0} + \rho^{-1}} + D \right) \phi_{t,1} \\
    &\quad \qquad \qquad +\frac{4\eta_{t,0} S_{t,1}^2 + 4\theta_{t,0} S_{t,1} + 4\eta_{t,0}^2\rho \mathrm{Var}(X_{t,1}) - \rho \theta_{t,0}^2 } {4(\eta_{t,0} + \rho^{-1})}\\ & \quad \qquad \qquad + \rho \kappa_{t,0}= \eta_{t,1} \phi_{t,1}^2 +\theta_{t,1} \phi_{t,1} + \kappa_{t,1}.
\end{aligned}  
\]
This proves the proposition when there is one period left. Suppose the result holds for the case when there are $g-1$ periods left. Then, the value function when there are $g$ periods left is given by
$$\tilde V_{t,g}^*(\phi_{t,g}) = \inf_{a\in\mathbb{R}} \left\{ a^2 +C\phi_{t,g}^2 + D\phi_{t,g} + \rho \mathbb{E} \left[ \tilde V_{t,g-1}^* (\phi_{t,g}+X_{t,g} - a) \right] \right\}$$    
\[
\begin{aligned}
    &=  \inf_{a\in\mathbb{R}}  \left\{ a^2  + \rho \eta_{t,g-1} \mathbb{E}\left[ (\phi_{t,g} + X_{t,g} - a)^2 \right] + \rho \kappa_{t,g-1}\right.\\
    &\qquad \left. + \rho \theta_{t,g-1} \mathbb{E}\left[ \phi_{t,g} + X_{t,g} - a \right]   + C\phi_{t,g}^2 + D\phi_{t,g}\right\} \\
    &=\hspace{-1pt} \inf_{a\in\mathbb{R}}\hspace{-1pt}\{ C\phi_{t,g}^2 \hspace{-1pt} + \hspace{-1pt}D\phi_{t,g}\hspace{-1pt} + \hspace{-1pt}\rho \eta_{t,g-1}\phi_{t,g}^2 +2\rho \eta_{t,g-1} \phi_{t,g} S_{t,g}\\
     &  \qquad +  (\rho \eta_{t,g-1} + 1)a^2 \hspace{-1pt}-\hspace{-1pt} \rho [2\eta_{t,g-1}(\phi_{t,g}\hspace{-1pt}+\hspace{-1pt}S_{t,g})+\theta_{t,g-1}]a\\
     &\qquad  + \rho \eta_{t,g-1} S_{t,g}^2 + \rho \theta_{t,g-1} \phi_{t,g} + \rho \theta_{t,g-1} S_{t,g} + \rho \kappa_{t,g-1} \}.
\end{aligned}  \]  
Then, the optimal action $\tilde a_{t,g}^*$ is given by
$$\tilde a_{t,g}^* = \frac{ \rho( 2\eta_{t,g-1} (\phi_{t,g} + S_{t,g}) +\theta _{t,g-1})}{2( \rho \eta_{t,g-1} + 1 )}= u_{t,g} \phi_{t,g} + w_{t,g}.$$
Plugging $\tilde a_{t,g}^*$ back into the value function yields $\tilde V^*_{t,g} (s_{t,g}) = \eta_{t,g} \phi_{t,g}^2 + \theta_{t,g} \phi_{t,g} + \kappa_{t,g}.$

\noindent\textbf{Proof of Equation~\eqref{eq:3}:}
For $t\geq 1$, we can approximate the optimal action to take in state $\mathbf s_t$ for the original MDP with optimality equation \eqref{eq:op1} as follows:
$$a_t^* (\mathbf s_t)  =  \arg\min_{a \in \mathcal{A}} \Big\{a^2 + C_{b,e}  (\mathbf{s}_t)       + \rho \sum_{\mathbf s_{t+1}} p(\mathbf s_{t+1} | \mathbf s_{t}, a) \Big(  \Big( \tilde a_{t+1,G}^*(\phi_{t+1,G})\Big)^2$$ 

$$+ C_{b,e}(\mathbf{s}_{t+1}) + \rho \mathbb{E}_{X_{t+1,G}} \left[\tilde V_{t+1,G-1}^* (\phi_{t+1,G-1}) \right] \Big)\Big\},$$  
which gives Equation \eqref{eq:3} by  reorganizing the right-hand side of the above equation by plugging in the optimal action for the special finite-horizon MDP given in Proposition~\ref{prop1}, removing $C_{b,e} (\mathbf s_t)$ as it does not depend on $a_t$, and using $\phi_{t,g-1} = \phi_{t,g} + X_{t,g} - a_{t,g}$. Then, we use the result that $\tilde V_{t+1,G-1}^*(\cdot)$ is a quadratic function and inner terms are a linear transformation of $\phi_{t+1,G}$. 

\noindent\textbf{Proof of Proposition~\ref{prop2}:}
The first-order and second-order partial derivatives of $\mathbb {E}\left[\tilde{C}(s_t, y)\right]$ with respect to $y$ are
$$
\frac{\partial \mathbb E [\tilde{C}(s_t,y)]}{\partial y} = 2y - \tilde{c}_b \int_{y-k_t}^\infty  f_{\mathbf s_t}(x) dx + \tilde{c}_e  \int_{-\infty}^{y-k_t}  f_{{\mathbf s_t}}(x) dx,
$$
$$
\frac{\partial^2 \mathbb E [\tilde{C}(s_t,y)]}{\partial y^2} = 2 + \tilde{c}_b f_{\mathbf s_t}(y-k_t) + \tilde{c}_e f_{\mathbf s_t}(y-k_t).
$$
Since $\tilde{c}_b\geq 0, \tilde{c}_e \geq 0$, and $f_{\mathbf s_t}(y) \geq 0$, we have $\frac{\partial^2 \mathbb E [\tilde{C}(s_t,y)]}{\partial y^2}  > 0$. To minimize the total expected cost, we set the first derivative to zero:
$$2y + \tilde{c}_b \int_{-\infty}^{y-k_t}  f_{\mathbf s_t} (x) dx + \tilde{c}_e \int_{-\infty}^{y-k_t}  f_{\mathbf s_t}(x) dx - \tilde{c}_b\int_{-\infty}^{\infty}  f_{\mathbf s_t}(x) dx =0$$
    
$$\Rightarrow F_{\mathbf s_t}( \tilde{y}) + \frac{2\tilde y}{\tilde{c}_b + \tilde{c}_e} = \frac{\tilde{c}_b - 2k_t} {\tilde{c}_b + \tilde{c}_e}, \mbox{ where }\tilde y=y-k_t.$$ Then, $y^*=\max(\tilde y+ k_t,0)$ minimizes $\mathbb E[\tilde{C}(s_t,y)]$ because $y \geq 0$ and $\mathbb E[\tilde{C}(s_t,y)]$ is convex in $y$. 

\noindent\textbf{Proof of Corollary~\ref{cor.1}:} To find states where optimal order quantity is zero, we need to identify states for which $\tilde{y}\leq -k_t$. By the monotonicity of $F_{\mathbf{s}_t}(\cdot)$, we have
$$
\tilde{y}\leq -k_t \Leftrightarrow F_{\mathbf{s}_t}(\tilde{y})+\frac{2\tilde{y}}{\tilde{c}_b+\tilde{c}_e}\leq F_{\mathbf{s}_t}(-k_t)-\frac{2k_t}{\tilde{c}_b+\tilde{c}_e}.
$$
Then, because $\tilde{y}$ satisfies \eqref{eq:5}, we have 
$$
F_{\mathbf{s}_t}(\tilde{y})+\frac{2\tilde{y}}{\tilde{c}_b+\tilde{c}_e}\leq F_{\mathbf{s}_t}(-k_t)-\frac{2k_t}{\tilde{c}_b+\tilde{c}_e}  \Leftrightarrow  \frac{\tilde{c}_b-2k_t}{\tilde{c}_b+\tilde{c}_e}\leq F_{\mathbf{s}_t}(-k_t)-\frac{2k_t}{\tilde{c}_b+\tilde{c}_e} \Leftrightarrow F_{\mathbf{s}_t}(-k_t)\geq \frac{\tilde{c}_b}{\tilde{c}_b+\tilde{c}_e}.
$$

We prove the upper bound on the optimal quantity by contradiction. If $y^*=\max\{\tilde{y}+k_t,0\}>\tilde{c}_b/2$, then  $$F_{\mathbf{s}_t}(\tilde{y})+\frac{2\tilde{y}}{\tilde{c}_b+\tilde{c}_e}>\frac{\tilde{c}_b-2k_t}{\tilde{c}_b+\tilde{c}_e}.$$ This means that $\tilde{y}$ does not satisfy \eqref{equ_nv_op}, and hence cannot be optimal. 

\noindent\textbf{Proof of Proposition~\ref{prop3}:}
The cumulative distribution function for a uniform distribution over $[\delta - \xi,\delta + \xi]$  is given by
\begin{equation*}
    F_{\mathbf s_t} (y) = \left\{
    \begin{array}{ll}
    0, &\ \mbox{if}\ y\leq \delta - \xi,\\
    \frac{y-\delta+\xi}{2\xi}, &\ \mbox{if}\ \delta-\xi < y \leq \delta+\xi,\\
    1, &\ \mbox{otherwise}.
    \end{array}
    \right.
\end{equation*}
We first use Corollary \ref{cor.1} to obtain that $y^*=0$ if and only if $k_t\leq -\delta+\xi(\tilde{c}_e-\tilde{c}_b)/(\tilde{c}_e+\tilde{c}_b)$ for this distribution.
For all other values of $k_t$, we then find the solution $\tilde{y}$ to Equation~\eqref{equ_nv_op} and let $y^*=\tilde{y}+k_t$.

When $-\delta+\xi(\tilde{c}_e-\tilde{c}_b)/(\tilde{c}_e+\tilde{c}_b)< k_t \leq (\tilde{c}_b/2)-(\delta-\xi)$, $\tilde y$ must be in $(\delta-\xi, \delta+\xi]$ for Equation~\eqref{equ_nv_op} to have a solution. In this case, solving the corresponding equation
\begin{equation*}
\frac{y-\delta+\xi}{2\xi} + \frac{2y}{\tilde{c}_b+ \tilde{c}_e} = \frac{\tilde{c}_b - 2k_t}{\tilde{c}_b + \tilde{c}_e}\label{eq:proof-case2-unif}
\end{equation*}
for $y$, we obtain 
\begin{equation*}
\tilde y=  \frac{\tilde{c}_b(\delta+\xi) + \tilde{c}_e(\delta-\xi)  - 4\xi k_t}{\tilde{c}_b + \tilde{c}_e+4\xi}.
\end{equation*}
When $k_t > (\tilde{c}_b/2)-(\delta-\xi)$, $\tilde y = \frac{\tilde{c}_b}{2} - k_t$ is the only solution to Equation~\eqref{equ_nv_op}.

\noindent\textbf{Proof of Proposition~\ref{prop4}:}
For a  triangular distribution with mode $\delta$ and range $[\delta - \xi,\delta + \xi]$, we have
\begin{equation*}
    F_{\mathbf s_t} (y) = \left\{
    \begin{array}{ll}
    0, &\ \mbox{if}\ y\leq \delta - \xi,\\
    \frac{(y-\delta+\xi)^2}{2\xi^2}, &\ \mbox{if}\ \delta-\xi < y \leq \delta,\\
    1- \frac{(\delta + \xi -y)^2}{2\xi^2}, &\ \mbox{if}\ \delta < y \leq \delta + \xi,\\
    1, &\ \mbox{otherwise}.
    \end{array}
    \right.
\end{equation*}
\noindent\textbf{Case I ($\tilde{c}_e\leq \tilde{c}_b$):} When $k_t\leq -\delta-\xi+\xi\sqrt{2\tilde{c}_e}/\sqrt{\tilde{c}_e+\tilde{c}_b}$, we have $k_t\leq -\delta$ because $\tilde{c}_e\leq \tilde{c}_b$. Then, we have $F_{\mathbf s_t}(-k_t)=1-(\delta+\xi+k_t)^2/(2\xi^2)$, which is greater than or equal to $\tilde{c}_b/(\tilde{c}_b+\tilde{c}_e)$, and hence by Corollary \ref{cor.1}, $y^*=0$. For all other values of $k_t$, we then find the solution $\tilde{y}$ to Equation~\eqref{equ_nv_op} and let $y^*=\tilde{y}+k_t$.

When $-\delta-\xi+\xi\sqrt{2\tilde{c}_e}/\sqrt{\tilde{c}_e+\tilde{c}_b}<k_t \leq (\tilde{c}_b - \tilde{c}_e)/4 -\delta$, $\tilde y$ must be in $(\delta, \delta+\xi]$  for Equation~\eqref{equ_nv_op} to have a solution. In this case, solving the corresponding equation
\begin{equation*}
        1- \frac{(\delta + \xi -y)^2}{2\xi^2} + \frac{2y}{\tilde{c}_b + \tilde{c}_e} = \frac{\tilde{c}_b -2k_t}{\tilde{c}_b + \tilde{c}_e}\label{eq:proof-case3}      
\end{equation*}
for $y$, we obtain
$$\tilde y= \delta + \xi + \frac{2 \xi^2-\xi \sqrt{4\xi^2 +2(\tilde{c}_e+\tilde{c}_b) (\tilde{c}_e+2\xi + 2\Delta_t )}}{\tilde{c}_b + \tilde{c}_e},$$
where the inside of the square root is positive for this range of $k_t$. Similarly, when $(\tilde{c}_b - \tilde{c}_e)/4 -\delta<k_t \leq \tilde{c}_b/2 -\delta+\xi$, $\tilde y$ must be in $(\delta-\xi, \delta]$ for Equation~\eqref{equ_nv_op} to have a solution. In this case, solving the corresponding equation
\begin{equation*}
\frac{(y-\delta+\xi)^2}{2\xi^2} + \frac{2y}{\tilde{c}_b+ \tilde{c}_e} = \frac{\tilde{c}_b - 2k_t}{\tilde{c}_b + \tilde{c}_e}\label{eq:proof-case2}
\end{equation*}
for $y$, we obtain 
$$\tilde y= \delta - \xi - \frac{2 \xi^2-\xi \sqrt{4\xi^2 +2(\tilde{c}_e+\tilde{c}_b) (\tilde{c}_b+2\xi  - 2\Delta_t)}}{\tilde{c}_b + \tilde{c}_e},$$
where the inside of the square root is positive because 
$k_t < \tilde{c}_b/2 - \delta+\xi$. When $k_t > \tilde{c}_b/2-\delta+\xi$, $\tilde y = \tilde{c}_b/2 - k_t$ is the only solution to Equation~\eqref{equ_nv_op}.

\noindent\textbf{Case II ($\tilde{c}_e> \tilde{c}_b$):} When $k_t\leq -\delta $, we have $F_{\mathbf s_t}(-k_t)=1-(\delta+\xi+k_t)^2/(2\xi^2)$, which is greater than $\tilde{c}_b/(\tilde{c}_b+\tilde{c}_e)$ because $k_t\leq-\delta <-\delta-\xi+\xi\sqrt{2\tilde{c}_e}/\sqrt{\tilde{c}_e+\tilde{c}_b}$ as $\tilde{c}_e> \tilde{c}_b$. When $-\delta<k_t\leq -\delta+\xi-\xi\sqrt{2\tilde{c}_b}/\sqrt{\tilde{c}_e+\tilde{c}_b}$, we have $F_{\mathbf s_t}(-k_t)=(-\delta+\xi-k_t)^2/(2\xi^2)$, which is greater than or equal to $\tilde{c}_b/(\tilde{c}_b+\tilde{c}_e)$. Hence, by Corollary \ref{cor.1}, $y^*=0$ when $k_t\leq -\delta+\xi-\xi\sqrt{2\tilde{c}_b}/\sqrt{\tilde{c}_e+\tilde{c}_b}$. For all other values of $k_t$, we then find the solution $\tilde{y}$ to Equation~\eqref{equ_nv_op} and let $y^*=\tilde{y}+k_t$. When $-\delta+\xi-\xi\sqrt{2\tilde{c}_b}/\sqrt{\tilde{c}_e+\tilde{c}_b}<k_t \leq \tilde{c}_b/2 -\delta+\xi$, $\tilde y$ must be in $(\delta-\xi, \delta]$ for Equation~\eqref{equ_nv_op} to have a solution. In this case, solving the corresponding equation, we obtain 
$$\tilde y= \delta - \xi - \frac{2 \xi^2-\xi \sqrt{4\xi^2 +2(\tilde{c}_e+\tilde{c}_b) (\tilde{c}_b+2\xi  - 2\Delta_t)}}{\tilde{c}_b + \tilde{c}_e}.
$$
When $k_t > \tilde{c}_b/2-\delta+\xi$, $\tilde y = \tilde{c}_b/2 - k_t$ is the only solution to Equation~\eqref{equ_nv_op}.

\section{Simulation Models}

We use two simulators developed from 2019 data from a large academic hospital in North Carolina, U.S. The Python simulator is used to train and tune the DQN and MHP algorithms (Section~\ref{app:ps:se1}). The Arena simulator, adapted from \cite{chen2022using}, is updated to reflect 2019 operational practices and input distributions estimated from the new data (Section~\ref{app:as:se1}). We validate the Arena simulator by comparing its outputs under standard 2019 inpatient bed-requesting practices with observed ED data (Section~\ref{app:as:se2}). Section~\ref{app:as:se3} describes how the proposed policies are integrated into the simulator, and Section~\ref{app:as:se4} summarizes the input analysis for the Arena simulator.

\subsection{Python Simulator}
\label{app:ps:se1}

The Python simulator provides a low-granularity, fast-executing environment for training and tuning the DQN and MHP heuristics. It abstracts away less critical operational details while preserving the key dynamics of ED arrivals, service, admission decisions, boarding, and hospital bed requests. The simulator incorporates an admission prediction tool (APT), adapted from \cite{mehrotra2017}, which classifies patients into two equal-sized groups: low admission probability (Type-I) and high admission probability (Type-II). Type-I and Type-II patients arrive according to nonhomogeneous Poisson processes with hourly varying rates. ED capacity is assumed to be infinite, and ED service times, representing combined triage, waiting, and workup durations, are exponentially distributed with rates that depend on patient type and arrival hour. These rates are calibrated using historical data from calendar year 2019.

After completing ED service, a patient is either discharged or admitted according to a Bernoulli disposition decision. The estimated admission probabilities are 0.0848 for Type-I and 0.5770 for Type-II patients. Discharged patients leave the system immediately, whereas admitted patients board until a previously requested hospital bed becomes available and all earlier boarders have been roomed. Every hour, the decision maker observes the number of patients in service and boarding, as well as the number of requested hospital beds, and determines how many additional beds to request. The simulator assumes infinite hospital bed capacity, with exponentially distributed bed preparation times whose rates depend on the request hour and are calibrated using historical boarding-time data. Table~\ref{app:combined} reports the hourly arrival rates, mean ED service times, and mean hospital bed preparation times used in the Python simulator.

\begin{table}[ht!]
\begin{center}
\fontsize{9}{8}\selectfont
\caption{Expected number of hourly arrivals and ED service times for Type-I and Type-II patients; and hospital bed preparation times used in the Python simulator.}
\label{app:combined}
\begin{tabular}{||c||c|c|c|c|c||} \hline
Hour of Day & \begin{tabular}[c]{@{}c@{}}Type-I\\ Arrival Rate\\ (patients per hour)\end{tabular} & \begin{tabular}[c]{@{}c@{}}Type-II\\ Arrival Rate\\ (patients per hour)\end{tabular} & \begin{tabular}[c]{@{}c@{}}Type-I\\ Service Time\\ (hours)\end{tabular} & \begin{tabular}[c]{@{}c@{}}Type-II\\ Service Time\\ (hours)\end{tabular} & \begin{tabular}[c]{@{}c@{}}Average Bed\\ Preparation Time\\ (hours)\end{tabular} \\ \hline\hline
0           & 2.40 & 2.09 & 3.57 & 4.65 & 5.22 \\
1           & 1.90 & 1.76 & 3.55 & 4.54 & 5.40 \\
2           & 1.56 & 1.42 & 3.25 & 4.65 & 5.66 \\
3           & 1.25 & 1.25 & 3.30 & 4.27 & 5.92 \\
4           & 1.25 & 1.19 & 3.04 & 3.85 & 5.45 \\
5           & 1.18 & 1.12 & 3.39 & 4.09 & 6.14 \\
6           & 1.44 & 1.17 & 3.49 & 4.50 & 6.96 \\
7           & 1.96 & 1.50 & 3.53 & 4.41 & 6.35 \\
8           & 3.11 & 2.35 & 3.45 & 4.52 & 5.46 \\
9           & 4.33 & 3.71 & 3.64 & 4.44 & 4.99 \\
10          & 4.87 & 4.76 & 3.87 & 4.63 & 5.06 \\
11          & 5.13 & 5.17 & 3.91 & 4.86 & 4.89 \\
12          & 4.63 & 5.73 & 3.99 & 4.97 & 4.65 \\
13          & 4.54 & 5.53 & 4.09 & 5.15 & 4.44 \\
14          & 4.06 & 5.50 & 4.19 & 5.08 & 4.47 \\
15          & 4.18 & 5.76 & 4.21 & 5.39 & 4.38 \\
16          & 4.44 & 5.34 & 4.28 & 5.40 & 4.38 \\
17          & 4.45 & 5.05 & 4.07 & 5.13 & 4.73 \\
18          & 4.67 & 4.85 & 4.07 & 5.07 & 5.09 \\
19          & 4.82 & 4.54 & 4.11 & 5.05 & 5.19 \\
20          & 5.31 & 3.99 & 3.91 & 4.81 & 5.00 \\
21          & 4.41 & 3.61 & 3.86 & 5.07 & 5.33 \\
22          & 3.79 & 2.99 & 3.89 & 4.88 & 5.36 \\
23          & 3.13 & 2.44 & 3.68 & 4.52 & 5.30 \\ \hline
\end{tabular}
\end{center}
\end{table}

\subsection{Arena Simulator}
\label{app:as:se1}

We use an Arena-based simulator, which is developed using the discrete-event simulation software Arena version 16.1 \cite{Arena}, as the testbed for evaluating the proposed  policies. Through a comprehensive exploratory analysis, we identified twenty distinct patient categories necessary for a realistic representation of the actual system. These categories are defined by combinations of five triage levels, two age groups (pediatric vs. adult), and two disposition outcomes (admit vs. discharge). Each patient category is characterized by unique arrival rates, time-to-disposition distributions, and admission probabilities, which are used to classify patients as either low-admit or high-admit types. While the simulator keeps track of and distinguishes all twenty categories, in alignment with real-life scenarios, the decision maker is only informed of a patient's disposition after their ED workup is complete. 

Upon arrival at the ED, a patient enters triage, where a nurse assesses their condition and assigns them to one of five triage categories based on acuity and resource needs. If an appropriate ED bed is available, the patient is roomed immediately. Otherwise, they wait until all higher-priority patients are accommodated and a suitable bed becomes available. Once roomed, the patient undergoes two stages of care. The first stage, referred to as the \textit{ED workup}, involves diagnostic tests, treatments, and clinical evaluation leading to a disposition decision (discharge or admission). This decision marks the start of the second stage: \textit{boarding} (for admissions) or \textit{discharge boarding} (for discharges). For discharged patients, discharge boarding is the period after medical clearance during which the patient remains in the ED while administrative and logistical tasks are completed. Only after discharge boarding ends does the patient leave and the ED bed become available for a waiting patient. For admitted patients, the second stage is boarding, during which they remain in the ED until a hospital bed is available. This stage also includes identifying care teams as well as preparing the hospital bed for the transfer of the patient, which we call the transfer preparation process (TPP). If a hospital bed is available at the time of the request, the boarding time equals the TPP duration. Otherwise, it includes both the wait time for a bed and the TPP duration. Once the hospital bed is available for use by the admitted patient, the patient is transferred, and their ED bed is freed. The patient’s hospital stay begins upon transfer. Since our primary interest lies in the stay duration in the ED rather than in-hospital processes, we model the entire hospital stay as a single stage of service. At the end of the hospital stay, the hospital bed is released and becomes available for the next ED-admitted patient.

\noindent \textbf{Arrival Processes:} Following the approach of  \cite{chen2022using}, we model the patient arrival processes using non-homogeneous Poisson processes. To ensure a robust fit, we adopt the methodology described in \cite{brown2005statistical}, dividing patients into 672 subgroups based on hour of the day, day of the week, and season, so that arrival rates within each subgroup are approximately constant. To address rounding errors in data recording, we apply a small uniform random noise to each data point. We then conduct a Kolmogorov-Smirnov (K-S) test to evaluate the null hypothesis that the arrival process within each subgroup is a Poisson process with a constant rate. Encouragingly, the null hypothesis is rejected in only 52 of the 672 subgroups, supporting the use of non-homogeneous Poisson processes for modeling arrival patterns. Our exploratory analysis also reveals that arrival patterns for different patient categories exhibit similar daily trends across certain weekdays. After extensive experimentation, we group arrival rates by 24-hour time slots, three weekday categories (Monday, Tuesday -– Friday, and Saturday -–Sunday), and 20 patient categories, resulting in a parsimonious yet flexible representation. Detailed arrival rate estimates are provided in Tables~\ref{app:ase:t0} through \ref{app:ase:pediatric-weekend-arrivals}.

\noindent \textbf{Triage, waiting area, and bed assignment:} Upon arrival at the ED, patients join a first-come-first-served triage queue staffed by  two triage nurses.   Triage duration (in minutes) follows a triangular distribution with a minimum of 5, maximum of 15, and a mode of 10 based on observed data. The triage process assigns each patient an Emergency Severity Index (ESI),  which in turn determines their priority in queue. (ESI is the most commonly used ED triage system in US and has five levels; level 1 being the most emergent and level 5 being the least urgent.) Patients are then assigned to an available bed; if none is available, they wait until all higher-priority patients have been accommodated.

The study ED has four care areas. Team A (14 beds, including 2 trauma beds for the highest-acuity patients) and Team B (16 beds) operate 24/7 as the primary adult treatment areas. Team D functions as a fast track for adult patients, operating from 7:00 am to 2:00 am, with capacity increasing during the day: 10 beds at 7 am, 15 beds at 11 am, and 20 beds at 3 pm. Team Ped serves pediatric patients, operating 24/7 with 9 beds. These capacities, provided by the ED manager, serve as reference values; actual ED capacity may vary through flexible measures such as opening additional areas or using hallway beds. As detailed in Section \ref{app:as:se2}, we adjust these bed capacities in the model for calibration. A patient’s priority level and suitable care areas are determined jointly by pediatric status and ESI levels. In our dataset, most patients fall into ESI level 3. To better represent real-world variation within this large group, we randomly split these patients into two subcategories of equal size: acute ESI 3 and non-acute ESI 3. Table~\ref{app:as1} summarizes these classifications.

\begin{table}[!ht]
\centering
\caption{Suitable ED Care Areas with respect to Age Group, Triage Level, and Priority Order.}
\label{app:as1}
\begin{tabular}{||c||c|c|c||}
\hline
Age Group & ESI           & Priority Order& \begin{tabular}[c]{@{}c@{}}Suitable Care Areas\\ (In order of preference as capacity allows)\end{tabular} \\ \hline \hline
Adult     & 1             & 1        & Team A Reserved, Team A, Team B                                                                      \\
Adult     & 2             & 2        & Team A, Team B                                                                                       \\
Adult     & 3 (Acute)     & 3        & Team A, Team B                                                                                       \\
Adult     & 3 (Non-acute) & 4        & Team A, Team B, Team D                                                                               \\
Adult     & 4 - 5         & 4        & 
Team D when operating; otherwise, Team A, Team B          \\
Pediatric & 1             & 1        & Team Ped, Team A, Team B                                                                             \\
Pediatric & 2             & 2        & Team Ped, Team A, Team B                                                                             \\
Pediatric & 3             & 3        & Team Ped, Team A, Team B                                                                             \\
Pediatric & 4 - 5         & 4        & Team Ped                                                                                             \\ \hline
\end{tabular}
\end{table}

\noindent \textbf{First-Stage ED Service:} Once assigned to an ED bed, the patient’s workup begins. In reality, this phase involves a complex sequence of evaluations, diagnostic tests, treatments, and monitoring, all constrained by resources and dependent on interactions with providers. Its duration is affected by staffing levels and resource availability. However, our data lack the granularity needed to estimate these factors reliably, and modeling them in full detail would be computationally costly. Since our focus is not on reducing workup time, we model the ED workup as a single activity. To reflect variation in resources and staffing throughout the day, this duration depends on both the patient’s age group and the time of day at which workup begins. Table~\ref{app:ase:workup-duration} presents the corresponding probability distributions.

\noindent \textbf{Second-Stage ED Service}: For discharged patients, the second stage is the discharge boarding process, during which the patient remains in the ED bed while staff complete all logistical tasks. Since our primary focus is on admitted patients, we model discharge boarding as a single activity. This duration’s distribution and parameters are estimated from data, varying by patient ESI level and the disposition time. Table~\ref{app:ase:second-stage-discharge} provides the corresponding estimates. For an admitted patient, the disposition decision initiates the boarding process, during which they remain in the ED bed until transfer to the hospital. Following standard procedures at most US EDs, a hospital bed request is sent immediately after the admission decision. Beds are assigned on a first-come-first-served basis. If a suitable hospital bed is available with no pending earlier requests, it is assigned immediately and the TPP begins. Otherwise, the request queues until prior requests are fulfilled and a bed becomes available, at which point TPP starts. We model the boarding process based on this logic and the approach outlined in \cite{shi2016models}, with boarding duration determined jointly by hospital bed availability and TPP length. Unlike \cite{shi2016models}, which separates delays into pre- and post-allocation periods, we treat TPP as a single delay due to limited data and modeling efficiency. This single delay effectively combines pre- and post-allocation times, with hospital bed capacity already incorporated into the simulation. Similar to \cite{shi2016models} and \cite{chen2022using}, TPP follows a truncated lognormal distribution with a mean of 3.3 hours, coefficient of variation of 0.6, and maximum of 12 hours.

\noindent \textbf{Hospital Stay:} Since our data  covers only ED encounters, we lack direct estimates for hospital length of stay (hLoS). Furthermore, it is not straightforward to set a capacity to hospital beds used by ED patients, as in reality, the hospital bed capacity may be adjusted in response to varying needs and there could be other sources of patients. To address these challenges, we follow the approach of \cite{chen2022using}, employing the two-time-scale approach from \cite{shi2016models} to represent hLoS in days (hLoS-day) and hours (hLoS-hour). This framework models hLoS by number of nights stayed and hours spent on the discharge day. Due to insufficient data for direct estimation, we adapt parameter estimates from \cite{shi2014patient} with minor modifications. For hLoS-day, we primarily use parameters from Table 11 of \cite{shi2014patient}, distinguishing between patients admitted before noon (ED-AM) and after noon (ED-PM) as in \cite{shi2016models}. According to our partner hospital’s operational practices, we consider that an ED-AM patient is a same-day patient (discharged on the same day as admission) with a probability of 0.1129, meaning that this patient will have hLoS-day of zero. For other ED-AM
patients, we assume their hLoS-day is equal to $\min(21, \lceil X \rceil)$, where $X$ follows a lognormal distribution with a
mean of 3.86 days and a standard deviation of 3.93 days. ED-PM patients, who cannot be same-day patients,
have an hLoS-day of $\min(21, \lceil Y \rceil)$, where $Y$ follows a lognormal distribution with mean of 4.51 days and
standard deviation of 4.2 days. For hLoS-hour, we modify and use the 
period-1 empirical distribution from Table 1 of \cite{shi2014patient}. Specifically, non–same-day patients follow this distribution directly. Same-day discharges, assumed to occur only in the afternoon, use empirical probabilities conditional on afternoon discharge hours. Table~\ref{app:as2} summarizes these empirical probabilities from Table 1 of \cite{shi2014patient}. 

To model hospital capacity available to ED-admitted patients, we assume a fixed number of beds are exclusively reserved for this group. We represent this reserved capacity as a multi-server, first-come-first-served queue with homogeneous servers dedicated solely to ED-admitted patients. Since our data do not allow reliable estimation of the reserved bed count, we treat this number as a tuning parameter, selecting the value that best aligns the simulated second-stage durations with observed data as part of the validation step of this simulation study.

\begin{table}[!ht]
\centering
\caption{Empirical Probability Distribution Used for hLoS-hour.}
\label{app:as2}
\begin{tabular}{||c|c|c||c|c|c||}
\hline
\multirow{2}{*}{hLoS-hour} & \textbf{Same-Day} & \textbf{Non Same-Day} & \multirow{2}{*}{hLoS-hour} & \textbf{Same-Day} & \textbf{Non Same-Day} \\
 & \textbf{Patient} & \textbf{Patient} & & \textbf{Patient} & \textbf{Patient} \\ \hline\hline
0--1   & 0.0000 & 0.0015 & 12--13 & 0.1119 & 0.0977 \\
1--2   & 0.0000 & 0.0015 & 13--14 & 0.2221 & 0.1939 \\
2--3   & 0.0000 & 0.0011 & 14--15 & 0.2948 & 0.2574 \\
3--4   & 0.0000 & 0.0009 & 15--16 & 0.1209 & 0.1056 \\
4--5   & 0.0000 & 0.0010 & 16--17 & 0.0696 & 0.0608 \\
5--6   & 0.0000 & 0.0011 & 17--18 & 0.0511 & 0.0446 \\
6--7   & 0.0000 & 0.0015 & 18--19 & 0.0421 & 0.0368 \\
7--8   & 0.0000 & 0.0007 & 19--20 & 0.0371 & 0.0324 \\
8--9   & 0.0000 & 0.0016 & 20--21 & 0.0292 & 0.0255 \\
9--10  & 0.0000 & 0.0132 & 21--22 & 0.0121 & 0.0106 \\
10--11 & 0.0000 & 0.0369 & 22--23 & 0.0054 & 0.0047 \\
11--12 & 0.0000 & 0.0655 & 23--24 & 0.0037 & 0.0032 \\
\hline
\end{tabular}
\end{table}

\subsection{Validation of the Arena Simulator} 
\label{app:as:se2}

Because the Arena simulator serves as the test bed for our framework, it must provide a realistic and reliable basis for evaluating potential implementation. Since the 2019 hospital data used for validation reflect current practice, in which beds are requested only after an admission decision, we validate the simulator by comparing its outputs without aggregate early bed requests to observed 2019 ED data. We use ED LoS as the primary validation metric, but also include first- and second-stage durations as secondary validation metrics.

The capacities of the four ED care areas and the inpatient unit are used as tuning parameters. The ED manager provided baseline values for the care areas (Section \ref{app:as:se1}), though in practice these can expand temporarily using hallways or stretchers. Hospital bed capacity, lacking a precise estimate, is adjusted as well to align simulated second-stage durations with observed data. After testing multiple combinations, the following configuration best matched reality: Team A capacity increased from 12 + 2 (reserved) to 16 + 2 (reserved), Team B from 16 to 19, Team D from 10 to 13 between 7 am -– 11 am, and Team Peds from 9 to 10; hospital bed capacity fixed at 337.

\begin{figure}[!ht]
    \centering
    \includegraphics[scale=0.40]{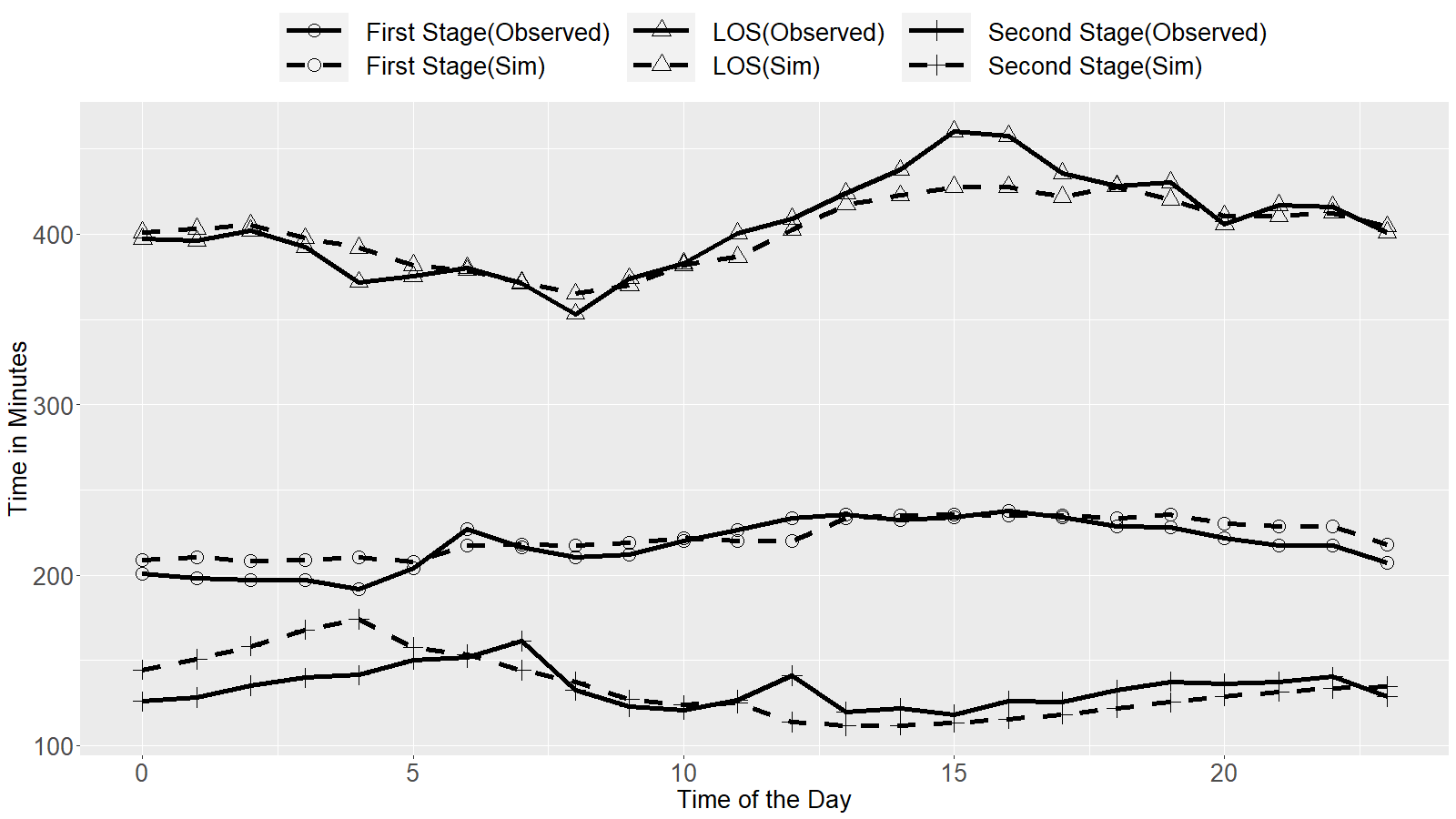}
    \caption{Observed and Simulated ED LoS, First-Stage, and Second-Stage Durations.}
    \label{fig:app:as:cm}
\end{figure}

Figure~\ref{fig:app:as:cm} compares average ED LoS, first-stage duration, and second-stage duration on an hourly basis using both observed and simulated data. The simulation runs for 1,000 years with a one-year warm-up. On the $x$-axis, each value denotes a one-hour slot (e.g., 0 is midnight to 1:00 am, and 1 represents 1:00 am to 2:00 am). ED LoS averages are computed over  patients arriving in that time slot; first-stage averages over patients roomed in that slot; and second-stage averages over patients receiving a disposition decision in that slot.
These plots show that the simulated metrics closely track those from observed data. Also, a Welch two-sample $t$-test fails to reject the null hypothesis that mean simulated ED LoS equals the observed mean at a significance level of 0.05. We therefore conclude that the Arena simulator reasonably captures reality, supporting its use for evaluating policies that are designed to improve average ED LoS.
 
\subsection{Implementing the Framework and Resulting Policies in ARENA}
\label{app:as:se3}

We here discuss how the proposed early bed request framework and proposed heuristics are implemented within the Arena simulator, which is built and validated under the current practice. We first assume that at the time a patient’s disposition decision is made, if the number of boarding patients awaiting early requested hospital beds equals the number of outstanding early bed requests, an individual hospital bed request is placed immediately for that patient instead of waiting to include it in the aggregate early bed requests at the start of the next hour. This individual bed request is specifically tied to the patient, distinguishing it from the aggregate early bed requests. Hence, a bed request is made for every patient at or before their disposition decision.

The Arena simulator models 20 patient categories, while the Python simulator includes only two types: Type-I and Type-II. To use the policies trained and tuned using the Python simulator, we first establish a mapping from the 20 patient categories in the Arena simulator to the two types in the Python simulator. We assume that the classification of a patient as Type-I in the Arena simulator follows a Bernoulli distribution, with probabilities depending on the patient’s age group (pediatric or adult) and ESI level, but independent of the final disposition decision. Probability estimates based on the 2019 ED data and the admission prediction tool used in this study are as follows: For pediatric patients, probabilities for Type-I patients are 0.026, 0.116, 0.765, 1.000, and 0.996 for ESI levels 1 through 5, respectively. For adult patients, the corresponding probabilities are 0.025, 0.060, 0.467, 0.983, and 0.921.
A patient’s classification as Type-I or Type-II is revealed upon completion of triage. When the patient’s ED workup ends, their admission status becomes known, and the counts of Type-I and Type-II patients are decreased accordingly. Every hour, the decision maker observes the number of Type-I and Type-II patients currently waiting or undergoing ED workup. Additionally, the number of patients boarding and awaiting early requested beds as well as the number of outstanding early bed requests are monitored. This information, combined with the active policy, guides the decision maker in determining the appropriate number of early bed requests to place with the hospital according to the heuristic policy used.

Early bed requests have the same priority as regular bed requests and are processed similarly. Specifically, once hospital beds have been assigned to all the bed requests preceding this request, if another hospital bed is available, the bed is assigned to this request and the TPP starts immediately. If no hospital bed is available at that time, the TPP can only start when another hospital bed becomes available. Upon completion of the TPP, the early requested hospital bed will be instantly occupied if there is a boarding patient waiting for an early requested hospital bed, otherwise, it joins a first-come-first-served queue and waits to be occupied.

A potential risk of early bed requests is that, when hospital capacity is extremely constrained, early bed requests could further overwhelm the bed-request queue and increase boarding times for patients who are ultimately admitted; see Section 5.4 in \cite{chen2022using}. To mitigate this risk, for every policy tested in our simulations under SOS, we pause early bed requests when the number of outstanding early bed requests is 55 or higher; regular bed requests, however, are still placed whenever needed. This cutoff of 55 was chosen as it  corresponds to the average daily number of hospital admissions in 2019. We raise the cap for the outstanding early bed requests from 55 to 65 under PRS.

\subsection{Input Distributions for the Arena Simulator}
\label{app:as:se4}

In this section, we provide a summary of the probability distributions and parameter estimates that are implemented in the Arena simulator. 
Tables~\ref{app:ase:t0} through \ref{app:ase:pediatric-weekend-arrivals} provide the arrival rates estimated from the 2019 arrival data and used for the non-homogeneous Poisson processes in our simulator. These rates depend on patients' age groups, ESI levels, and final disposition decisions. As noted earlier, since the arrival patterns are similar on certain weekdays, we estimate the arrival rates separately for Monday, Tuesday through Friday, and the weekend (Saturday and Sunday). Each one-hour slot is represented by an integer from 0 to 23, where 0 corresponds to 12:00am -- 12:59am and 23 to 11:00pm-- 11:59pm.

\begin{table}[!ht]
\renewcommand{\arraystretch}{1}
\setlength{\tabcolsep}{3pt}
\centering
\caption{Hourly Arrival Rates on Weekdays for Adult Admit Patients.}
\label{app:ase:t0}
{\fontsize{9}{8}\selectfont
\begin{tabular}{c|ccccc|ccccc}
\hline
& \multicolumn{5}{c|}{Monday} & \multicolumn{5}{c}{Tuesday--Friday} \\
Hour of Day & ESI 1 & ESI 2 & ESI 3 & ESI 4 & ESI 5 & ESI 1 & ESI 2 & ESI 3 & ESI 4 & ESI 5 \\ \hline
0  & 0.076923 & 0.538462 & 0.615385 & 0        & 0        & 0.033493 & 0.344498 & 0.578947 & 0.047847 & 0        \\
1  & 0.096154 & 0.403846 & 0.480769 & 0.038462 & 0        & 0.047847 & 0.315790 & 0.377990 & 0.028708 & 0.004785 \\
2  & 0.038462 & 0.326923 & 0.557692 & 0.038462 & 0        & 0.023923 & 0.210526 & 0.531101 & 0.019139 & 0        \\
3  & 0.019231 & 0.307692 & 0.346154 & 0.038462 & 0        & 0.033493 & 0.320574 & 0.406699 & 0.019139 & 0        \\
4  & 0.019231 & 0.346154 & 0.442308 & 0.019231 & 0        & 0.028708 & 0.200957 & 0.387560 & 0.004785 & 0        \\
5  & 0.057692 & 0.230769 & 0.442308 & 0.019231 & 0        & 0.023923 & 0.210526 & 0.344498 & 0.014354 & 0        \\
6  & 0        & 0.192308 & 0.480769 & 0        & 0        & 0.033493 & 0.215311 & 0.444976 & 0.009569 & 0.004785 \\
7  & 0.038462 & 0.442308 & 0.596154 & 0.019231 & 0        & 0.014354 & 0.291866 & 0.526316 & 0.023923 & 0        \\
8  & 0.076923 & 0.519231 & 0.730769 & 0.057692 & 0        & 0.052632 & 0.421053 & 0.923445 & 0.028708 & 0.004785 \\
9  & 0.115385 & 0.807692 & 1.807692 & 0.019231 & 0.019231 & 0.105263 & 0.779904 & 1.416268 & 0.019139 & 0        \\
10 & 0.038462 & 1.442308 & 1.615385 & 0.038462 & 0        & 0.095694 & 1.033493 & 1.885168 & 0.057416 & 0        \\
11 & 0.134615 & 1.442308 & 2.211539 & 0.019231 & 0        & 0.100478 & 1.248804 & 2.100479 & 0.033493 & 0        \\
12 & 0.134615 & 1.903846 & 2.134615 & 0.019231 & 0        & 0.119617 & 1.253589 & 2.009569 & 0.043062 & 0        \\
13 & 0.115385 & 1.538462 & 2.173077 & 0.096154 & 0.019231 & 0.133971 & 1.392345 & 2.181818 & 0.043062 & 0        \\
14 & 0.096154 & 1.519231 & 2.057692 & 0.019231 & 0        & 0.086124 & 1.377990 & 2.009569 & 0.009569 & 0        \\
15 & 0.211538 & 1.980769 & 1.846154 & 0.096154 & 0        & 0.100478 & 1.397129 & 1.995215 & 0.019139 & 0        \\
16 & 0.153846 & 1.750000 & 1.596154 & 0.057692 & 0        & 0.129187 & 1.459330 & 1.789474 & 0.023923 & 0        \\
17 & 0.153846 & 1.576923 & 1.692308 & 0.019231 & 0        & 0.110048 & 1.258373 & 1.660287 & 0.033493 & 0        \\
18 & 0.192308 & 1.442308 & 1.538462 & 0.038462 & 0        & 0.119617 & 1.770335 & 1.540670 & 0.033493 & 0        \\
19 & 0.057692 & 1.096154 & 1.769231 & 0.076923 & 0        & 0.110048 & 0.990431 & 1.521531 & 0.081340 & 0        \\
20 & 0.076923 & 1.269231 & 1.423077 & 0.038462 & 0        & 0.119617 & 0.928230 & 1.191388 & 0.062201 & 0        \\
21 & 0.134615 & 0.923077 & 1.019231 & 0.057692 & 0        & 0.076555 & 0.842105 & 1.100479 & 0.047847 & 0        \\
22 & 0.192308 & 0.865385 & 0.846154 & 0.057692 & 0        & 0.052632 & 0.751196 & 0.880383 & 0.028708 & 0        \\
23 & 0.038462 & 0.634615 & 0.788462 & 0.019231 & 0        & 0.062201 & 0.497608 & 0.755981 & 0.033493 & 0  \\ \hline     
\end{tabular}
}
\end{table}

\begin{table}[!ht]
\renewcommand{\arraystretch}{1}
\setlength{\tabcolsep}{3pt}
\centering
\caption{Hourly Arrival Rates on Weekdays for Adult Discharge Patients.}
\label{app:ase:adult-discharge-arrivals}
{\fontsize{9}{8}\selectfont
\begin{tabular}{c|ccccc|ccccc}
\hline
& \multicolumn{5}{c|}{Monday} & \multicolumn{5}{c}{Tuesday--Friday} \\
Hour of Day & ESI 1 & ESI 2 & ESI 3 & ESI 4 & ESI 5 & ESI 1 & ESI 2 & ESI 3 & ESI 4 & ESI 5 \\
\hline
0  & 0        & 0.153846 & 1.711539 & 0.384615 & 0.057692 & 0.009569 & 0.272727 & 1.330144 & 0.511962 & 0.057416 \\
1  & 0.076923 & 0.173077 & 1.134615 & 0.346154 & 0.019231 & 0.014354 & 0.239234 & 1.100479 & 0.488038 & 0.043062 \\
2  & 0        & 0.230769 & 1        & 0.269231 & 0.019231 & 0.023923 & 0.133971 & 0.808612 & 0.392345 & 0.052632 \\
3  & 0.019231 & 0.173077 & 0.75     & 0.230769 & 0.057692 & 0.019139 & 0.08134  & 0.736842 & 0.291866 & 0.038278 \\
4  & 0        & 0.192308 & 1        & 0.307692 & 0.019231 & 0        & 0.129187 & 0.717703 & 0.244019 & 0.038278 \\
5  & 0        & 0.057692 & 0.846154 & 0.269231 & 0        & 0.014354 & 0.07177  & 0.784689 & 0.296651 & 0.038278 \\
6  & 0.019231 & 0.057692 & 1.211539 & 0.230769 & 0.096154 & 0        & 0.076555 & 0.942584 & 0.368421 & 0.066986 \\
7  & 0.019231 & 0.076923 & 1.346154 & 0.557692 & 0        & 0        & 0.124402 & 1.440191 & 0.488038 & 0.028708 \\
8  & 0.019231 & 0.230769 & 2.75     & 0.865385 & 0.038462 & 0.014354 & 0.267943 & 2.177034 & 0.746412 & 0.057416 \\
9  & 0.057692 & 0.576923 & 3.461539 & 1.019231 & 0.057692 & 0.028708 & 0.392345 & 3.440191 & 0.889952 & 0.057416 \\
10 & 0.019231 & 0.673077 & 5.019231 & 1.288462 & 0.076923 & 0.038278 & 0.588517 & 3.837321 & 0.880383 & 0.071703 \\
11 & 0        & 0.673077 & 4.384615 & 1.423077 & 0.057692 & 0.043062 & 0.593301 & 3.861244 & 0.904306 & 0.066986 \\
12 & 0.038462 & 0.884615 & 3.75     & 1.153846 & 0.115385 & 0.062201 & 0.641148 & 3.641148 & 1.009569 & 0.047847 \\
13 & 0.038462 & 0.480769 & 4.057692 & 1.076923 & 0.134615 & 0.009569 & 0.061244 & 3.354067 & 0.803828 & 0.057416 \\
14 & 0.019231 & 0.75     & 3.096154 & 1.230769 & 0.057692 & 0.062201 & 0.660287 & 3.019139 & 0.751196 & 0.038278 \\
15 & 0.076923 & 0.826923 & 3.288462 & 0.75     & 0.038462 & 0.033493 & 0.688995 & 3.186603 & 0.84689  & 0.043062 \\
16 & 0.038462 & 0.980769 & 3.115385 & 1.076923 & 0.057692 & 0.052632 & 0.655502 & 3.282297 & 0.799043 & 0.062201 \\
17 & 0.019231 & 0.538462 & 3.269231 & 1.038462 & 0.076923 & 0.028708 & 0.660287 & 2.913876 & 0.77512  & 0.052632 \\
18 & 0.096154 & 0.865385 & 2.653846 & 0.961539 & 0.096154 & 0.033493 & 0.636364 & 3.019139 & 0.899522 & 0.033493 \\
19 & 0.096154 & 0.615385 & 3.596154 & 0.807692 & 0.057692 & 0.023923 & 0.607656 & 2.947368 & 0.856459 & 0.062201 \\
20 & 0.057692 & 0.423077 & 2.711539 & 1.038462 & 0.096154 & 0.066986 & 0.564593 & 2.84689  & 1.047847 & 0.08134  \\
21 & 0.038462 & 0.461538 & 2.788462 & 0.788462 & 0.038462 & 0.043062 & 0.4689   & 2.397129 & 0.91866  & 0.07177  \\
22 & 0.019231 & 0.442308 & 2.423077 & 0.711539 & 0.096154 & 0.033493 & 0.37799  & 2.129187 & 0.813397 & 0.076555 \\
23 & 0.019231 & 0.25     & 1.538462 & 0.769231 & 0.076923 & 0.023292 & 0.296651 & 1.669857 & 0.588517 & 0.052632 \\
\hline
\end{tabular}
}
\end{table}

\begin{table}[!ht]
\renewcommand{\arraystretch}{1}
\setlength{\tabcolsep}{3pt}
\centering
\caption{Hourly Arrival Rates on Saturdays and Sundays for Adult Patients.}
\label{app:ase:adult-weekend-arrivals}
{\fontsize{9}{8}\selectfont
\begin{tabular}{c|ccccc|ccccc}
\hline
 & \multicolumn{5}{c|}{Admit Patients} & \multicolumn{5}{c}{Discharge Patients} \\
Hour of Day & ESI 1 & ESI 2 & ESI 3 & ESI 4 & ESI 5 & ESI 1 & ESI 2 & ESI 3 & ESI 4 & ESI 5 \\
\hline
0  & 0.153846 & 0.442308 & 0.692308 & 0.028846 & 0 & 0.028846 & 0.317308 & 1.548077 & 0.701923 & 0.038462 \\
1  & 0.038462 & 0.471154 & 0.538462 & 0.009615 & 0 & 0.048077 & 0.336539 & 1.326923 & 0.644231 & 0.040877 \\
2  & 0.057692 & 0.365385 & 0.461539 & 0.028846 & 0 & 0.048077 & 0.288462 & 1.134615 & 0.528846 & 0.028846 \\
3  & 0.086538 & 0.365385 & 0.451923 & 0.019231 & 0 & 0.038462 & 0.173077 & 0.778846 & 0.461539 & 0.067308 \\
4  & 0.076923 & 0.278846 & 0.509615 & 0.009615 & 0 & 0.009615 & 0.153846 & 0.961539 & 0.365385 & 0.048077 \\
5  & 0.076923 & 0.288462 & 0.413462 & 0.009615 & 0 & 0.028846 & 0.115385 & 1.134615 & 0.278846 & 0.019231 \\
6  & 0.028846 & 0.269231 & 0.336585 & 0        & 0 & 0.009615 & 0.182692 & 1.038462 & 0.442308 & 0.028846 \\
7  & 0.028846 & 0.250000 & 0.490385 & 0.019231 & 0 & 0.019231 & 0.163462 & 1.509615 & 0.451923 & 0.038462 \\
8  & 0.057692 & 0.403846 & 0.951923 & 0.019231 & 0 & 0.009615 & 0.201923 & 1.875000 & 0.586539 & 0.038462 \\
9  & 0.048077 & 0.596154 & 1.115385 & 0.019231 & 0 & 0        & 0.298077 & 2.759615 & 0.663462 & 0.067308 \\
10 & 0.038462 & 0.692308 & 1.298077 & 0.009615 & 0 & 0.019231 & 0.336539 & 3.394231 & 0.846154 & 0.038462 \\
11 & 0.076923 & 0.605769 & 1.355769 & 0.028846 & 0 & 0.038462 & 0.490385 & 3.394231 & 0.942308 & 0.105769 \\
12 & 0.076923 & 0.846154 & 1.355769 & 0.028846 & 0 & 0.009615 & 0.461539 & 3.538462 & 1.067308 & 0.038462 \\
13 & 0.038462 & 0.836539 & 1.509615 & 0        & 0 & 0.038462 & 0.461539 & 3.355769 & 1.019231 & 0.038462 \\
14 & 0.048077 & 0.798077 & 1.490385 & 0.019231 & 0 & 0.009615 & 0.413462 & 3.326923 & 0.923077 & 0.038462 \\
15 & 0.105769 & 0.990385 & 1.548077 & 0.048077 & 0 & 0.067308 & 0.625000 & 2.817308 & 0.980769 & 0.067308 \\
16 & 0.105769 & 0.750000 & 1.365385 & 0.048077 & 0 & 0.009615 & 0.432692 & 2.980769 & 0.913462 & 0.076923 \\
17 & 0.067308 & 0.701923 & 1.355769 & 0.028846 & 0 & 0.096154 & 0.423077 & 2.980769 & 0.807692 & 0.038462 \\
18 & 0.105769 & 0.701923 & 1.201923 & 0.048077 & 0 & 0.009615 & 0.500000 & 2.596154 & 0.884615 & 0.067308 \\
19 & 0.125000 & 0.682692 & 1.096154 & 0.057692 & 0 & 0.048077 & 0.500000 & 2.596154 & 1.211539 & 0.086538 \\
20 & 0.115385 & 0.625000 & 1.182692 & 0.076923 & 0 & 0.028846 & 0.442308 & 2.923077 & 1.221154 & 0.163462 \\
21 & 0.057692 & 0.615385 & 0.875000 & 0.057692 & 0 & 0.028846 & 0.384615 & 2.471154 & 1.057692 & 0.067308 \\
22 & 0.125000 & 0.519231 & 0.884615 & 0.038462 & 0 & 0.048077 & 0.269231 & 1.903846 & 1.096154 & 0.067308 \\
23 & 0.038462 & 0.509615 & 0.951923 & 0.019231 & 0 & 0.038462 & 0.250000 & 1.836539 & 0.817308 & 0.057692 \\
\hline
\end{tabular}
}
\end{table}

\begin{table}[!ht]
\renewcommand{\arraystretch}{1}
\setlength{\tabcolsep}{3pt}
\centering
\caption{Hourly Arrival Rates on Weekdays for Pediatric Admit Patients.}
\label{app:ase:pediatric-admit-arrivals}
{\fontsize{9}{8}\selectfont
\begin{tabular}{c|ccccc|ccccc}
\hline
& \multicolumn{5}{c|}{Monday} & \multicolumn{5}{c}{Tuesday--Friday} \\
Hour of Day & ESI 1 & ESI 2 & ESI 3 & ESI 4 & ESI 5 & ESI 1 & ESI 2 & ESI 3 & ESI 4 & ESI 5 \\
\hline
0  & 0.019231 & 0.057692 & 0.153846 & 0.019231 & 0 & 0.009569 & 0.052632 & 0.129187 & 0.004785 & 0 \\
1  & 0        & 0.096154 & 0.115385 & 0.019231 & 0 & 0.009569 & 0.076555 & 0.138756 & 0.004785 & 0 \\
2  & 0        & 0.019231 & 0.115385 & 0        & 0 & 0        & 0.066986 & 0.081340 & 0.014354 & 0 \\
3  & 0        & 0.057692 & 0.076923 & 0        & 0 & 0        & 0.028708 & 0.052632 & 0        & 0 \\
4  & 0        & 0.038462 & 0.076923 & 0        & 0 & 0        & 0.033493 & 0.043062 & 0.004785 & 0 \\
5  & 0        & 0.038462 & 0        & 0        & 0 & 0.004785 & 0.028708 & 0.038278 & 0.004785 & 0 \\
6  & 0.019231 & 0.057692 & 0.019231 & 0        & 0 & 0        & 0.019139 & 0.023923 & 0.004785 & 0 \\
7  & 0        & 0.019231 & 0.038462 & 0        & 0 & 0.004785 & 0.043062 & 0.023923 & 0        & 0 \\
8  & 0        & 0.019231 & 0.057692 & 0        & 0 & 0.004785 & 0.057416 & 0.066986 & 0        & 0 \\
9  & 0        & 0.096154 & 0.134615 & 0.019231 & 0 & 0.004785 & 0.095694 & 0.090909 & 0.004785 & 0 \\
10 & 0        & 0.211538 & 0.173077 & 0.019231 & 0 & 0.004785 & 0.143541 & 0.148325 & 0.019139 & 0 \\
11 & 0        & 0.115385 & 0.134615 & 0.038462 & 0 & 0.009569 & 0.129187 & 0.167464 & 0.033493 & 0 \\
12 & 0.057692 & 0.230769 & 0.230769 & 0.019231 & 0 & 0.004785 & 0.186603 & 0.205742 & 0.004785 & 0 \\
13 & 0        & 0.365385 & 0.230769 & 0        & 0 & 0        & 0.181818 & 0.234450 & 0.009569 & 0 \\
14 & 0.019231 & 0.211538 & 0.153846 & 0.038462 & 0 & 0.004785 & 0.215311 & 0.186603 & 0        & 0 \\
15 & 0.019231 & 0.192308 & 0.230769 & 0.076923 & 0 & 0.023923 & 0.200957 & 0.186603 & 0.019139 & 0 \\
16 & 0        & 0.326923 & 0.346154 & 0        & 0 & 0.009569 & 0.296651 & 0.229665 & 0.028708 & 0 \\
17 & 0.019231 & 0.230769 & 0.288462 & 0.019231 & 0 & 0.014354 & 0.191388 & 0.291866 & 0.009569 & 0 \\
18 & 0.019231 & 0.346154 & 0.153846 & 0.019231 & 0 & 0.028708 & 0.282297 & 0.162679 & 0.033493 & 0 \\
19 & 0.019231 & 0.230769 & 0.269231 & 0.019231 & 0 & 0.043062 & 0.229665 & 0.215311 & 0.014354 & 0 \\
20 & 0.019231 & 0.211538 & 0.211538 & 0.019231 & 0 & 0.019139 & 0.162679 & 0.224880 & 0.028708 & 0 \\
21 & 0        & 0.173077 & 0.096154 & 0.019231 & 0 & 0.028708 & 0.110048 & 0.234450 & 0.028708 & 0 \\
22 & 0.019231 & 0.134615 & 0.153846 & 0        & 0 & 0.019139 & 0.148325 & 0.148325 & 0.019139 & 0 \\
23 & 0        & 0.096154 & 0.134615 & 0        & 0 & 0        & 0.071770 & 0.133971 & 0        & 0 \\
\hline
\end{tabular}
}
\end{table}

\begin{table}[!ht]
\renewcommand{\arraystretch}{1}
\setlength{\tabcolsep}{3pt}
\centering
\caption{Hourly Arrival Rates on Weekdays for Pediatric Discharge Patients.}
\label{app:ase:pediatric-discharge-arrivals}
{\fontsize{9}{8}\selectfont
\begin{tabular}{c|ccccc|ccccc}
\hline
& \multicolumn{5}{c|}{Monday} & \multicolumn{5}{c}{Tuesday--Friday} \\
Hour of Day & ESI 1 & ESI 2 & ESI 3 & ESI 4 & ESI 5 & ESI 1 & ESI 2 & ESI 3 & ESI 4 & ESI 5 \\
\hline
0  & 0        & 0.057692 & 0.423077 & 0.326923 & 0.057692 & 0        & 0.086124 & 0.349282 & 0.291866 & 0.009569 \\
1  & 0        & 0.019231 & 0.403846 & 0.230769 & 0        & 0        & 0.043062 & 0.167464 & 0.253589 & 0.014354 \\
2  & 0        & 0.038462 & 0.269231 & 0.115385 & 0        & 0        & 0.019139 & 0.167464 & 0.129187 & 0.014354 \\
3  & 0        & 0.019231 & 0.173077 & 0.173077 & 0.019231 & 0        & 0        & 0.167464 & 0.114833 & 0.009569 \\
4  & 0        & 0.038462 & 0.250000 & 0.134615 & 0        & 0        & 0.019139 & 0.143541 & 0.105263 & 0        \\
5  & 0        & 0        & 0.057692 & 0.115385 & 0        & 0        & 0.028708 & 0.114833 & 0.105263 & 0        \\
6  & 0.019231 & 0.019231 & 0.153846 & 0.096154 & 0        & 0.004785 & 0.014354 & 0.119617 & 0.167464 & 0        \\
7  & 0        & 0.038462 & 0.153846 & 0.134615 & 0        & 0        & 0.019139 & 0.191388 & 0.157895 & 0.033493 \\
8  & 0        & 0.096154 & 0.461538 & 0.211538 & 0.019231 & 0.004785 & 0.057416 & 0.354067 & 0.277512 & 0.028708 \\
9  & 0        & 0.192308 & 0.596154 & 0.442308 & 0.096154 & 0        & 0.095694 & 0.483254 & 0.354067 & 0.052632 \\
10 & 0        & 0.153846 & 0.634615 & 0.307692 & 0        & 0        & 0.119617 & 0.617225 & 0.358852 & 0.043062 \\
11 & 0        & 0.365385 & 0.826923 & 0.634615 & 0.038462 & 0.004785 & 0.138756 & 0.669857 & 0.387560 & 0.028708 \\
12 & 0        & 0.173077 & 0.942308 & 0.365385 & 0.038462 & 0        & 0.186603 & 0.669857 & 0.325359 & 0.038278 \\
13 & 0        & 0.307692 & 0.730769 & 0.384615 & 0.038462 & 0        & 0.181818 & 0.665072 & 0.272727 & 0.047847 \\
14 & 0        & 0.230769 & 0.769231 & 0.307692 & 0.076923 & 0        & 0.191388 & 0.641148 & 0.334928 & 0.038278 \\
15 & 0        & 0.192308 & 0.807692 & 0.403846 & 0.038462 & 0.004785 & 0.272727 & 0.751196 & 0.258373 & 0.019139 \\
16 & 0        & 0.192308 & 0.634615 & 0.423077 & 0.057692 & 0.004785 & 0.200957 & 0.712919 & 0.344498 & 0.028708 \\
17 & 0        & 0.403846 & 0.980769 & 0.500000 & 0.096154 & 0        & 0.334928 & 0.732057 & 0.478469 & 0.047847 \\
18 & 0        & 0.269231 & 1.096154 & 0.711538 & 0.153846 & 0        & 0.315789 & 0.918660 & 0.550239 & 0.057416 \\
19 & 0        & 0.192308 & 0.884615 & 0.634615 & 0        & 0.004785 & 0.239234 & 0.832536 & 0.693780 & 0.047847 \\
20 & 0        & 0.153846 & 1.076923 & 1.057692 & 0.019231 & 0.004785 & 0.215311 & 0.861244 & 0.684211 & 0.028708 \\
21 & 0        & 0.307692 & 0.942308 & 0.846154 & 0.038462 & 0.004785 & 0.177033 & 0.755981 & 0.665072 & 0.038278 \\
22 & 0        & 0.173077 & 0.653846 & 0.519231 & 0.038462 & 0.004785 & 0.095694 & 0.645933 & 0.497608 & 0.028708 \\
23 & 0        & 0.173077 & 0.500000 & 0.442308 & 0        & 0.004785 & 0.100478 & 0.516746 & 0.468900 & 0.043062 \\
\hline
\end{tabular}
}
\end{table}

\begin{table}[!ht]
\renewcommand{\arraystretch}{1}
\setlength{\tabcolsep}{3pt}
\centering
\caption{Hourly Arrival Rates on Saturdays and Sundays for Pediatric Patients.}
\label{app:ase:pediatric-weekend-arrivals}
{\fontsize{9}{8}\selectfont
\begin{tabular}{c|ccccc|ccccc}
\hline
 & \multicolumn{5}{c|}{Admit Patients} & \multicolumn{5}{c}{Discharge Patients} \\
Hour of Day & ESI 1 & ESI 2 & ESI 3 & ESI 4 & ESI 5 & ESI 1 & ESI 2 & ESI 3 & ESI 4 & ESI 5 \\
\hline
0  & 0.009615 & 0.096154 & 0.038462 & 0.009615 & 0        & 0        & 0.115385 & 0.528846 & 0.326923 & 0        \\
1  & 0.019231 & 0.009615 & 0.067308 & 0.009615 & 0        & 0.009615 & 0.067308 & 0.250000 & 0.259615 & 0.019231 \\
2  & 0        & 0.067308 & 0.057692 & 0        & 0        & 0        & 0.038462 & 0.211539 & 0.211538 & 0        \\
3  & 0.009615 & 0.028846 & 0.048077 & 0        & 0.009615 & 0        & 0.019231 & 0.173077 & 0.115385 & 0.009615 \\
4  & 0.009615 & 0.028846 & 0.086538 & 0.019231 & 0        & 0.009615 & 0.019231 & 0.182692 & 0.067308 & 0        \\
5  & 0        & 0.019204 & 0.009615 & 0.009615 & 0        & 0        & 0.038462 & 0.134615 & 0.086538 & 0.019231 \\
6  & 0        & 0.028842 & 0.028846 & 0        & 0        & 0        & 0.057692 & 0.115385 & 0.173077 & 0        \\
7  & 0.009615 & 0.067308 & 0.038462 & 0.019231 & 0        & 0        & 0.057692 & 0.221154 & 0.115385 & 0        \\
8  & 0.009615 & 0.028846 & 0.067308 & 0.009615 & 0        & 0        & 0.057692 & 0.403846 & 0.153846 & 0.009615 \\
9  & 0.009615 & 0.057692 & 0.096154 & 0.028846 & 0        & 0        & 0.057692 & 0.403846 & 0.442308 & 0.028846 \\
10 & 0        & 0.076923 & 0.163462 & 0.009615 & 0        & 0        & 0.076923 & 0.615385 & 0.259615 & 0.028846 \\
11 & 0        & 0.086538 & 0.163462 & 0.019231 & 0        & 0.009615 & 0.086538 & 0.769231 & 0.490385 & 0.028846 \\
12 & 0.019231 & 0.153846 & 0.115385 & 0.038462 & 0        & 0.019231 & 0.125000 & 0.759615 & 0.442308 & 0.057692 \\
13 & 0.019231 & 0.221154 & 0.086538 & 0        & 0        & 0.019231 & 0.221154 & 0.625000 & 0.365385 & 0.038462 \\
14 & 0.009615 & 0.125000 & 0.182692 & 0.019231 & 0        & 0        & 0.173077 & 0.730769 & 0.500000 & 0        \\
15 & 0.009615 & 0.105769 & 0.182692 & 0.009615 & 0        & 0.019231 & 0.182692 & 0.682692 & 0.519231 & 0.057692 \\
16 & 0.019231 & 0.115385 & 0.221154 & 0        & 0        & 0        & 0.173077 & 0.750000 & 0.519231 & 0.057692 \\
17 & 0        & 0.173077 & 0.173077 & 0.009615 & 0        & 0        & 0.182692 & 0.961539 & 0.461538 & 0.057692 \\
18 & 0        & 0.173077 & 0.163462 & 0.009615 & 0        & 0        & 0.221154 & 1.000000 & 0.557692 & 0.038462 \\
19 & 0.019231 & 0.153846 & 0.173077 & 0.009615 & 0        & 0.009615 & 0.115385 & 0.711539 & 0.798077 & 0.067308 \\
20 & 0.038462 & 0.182692 & 0.211538 & 0.009615 & 0        & 0.009615 & 0.240385 & 0.951923 & 0.701923 & 0.048077 \\
21 & 0.019231 & 0.163462 & 0.221154 & 0        & 0.009615 & 0        & 0.201923 & 0.750000 & 0.634615 & 0.019231 \\
22 & 0.019231 & 0.115385 & 0.115385 & 0.028846 & 0        & 0        & 0.096154 & 0.548077 & 0.615385 & 0.038462 \\
23 & 0.019231 & 0.067308 & 0.192308 & 0.009615 & 0        & 0        & 0.134615 & 0.519231 & 0.576923 & 0.028846 \\
\hline
\end{tabular}
}
\end{table}

For probability distributions of first- and second-stage durations, we again report the $p$-value from the Kolmogorov-Smirnov (K-S) test, where a large $p$-value indicates a good fit. For convenience, we define the following notation to refer to distributions that provided a good fit in our input analysis of the 2019 ED data. \textbf{EXP($\beta$)}: Exponential distribution with mean $\beta$. \textbf{GAMMA($\beta,\alpha$)}: Gamma distribution with shape parameter $\alpha$ and scale parameter $\beta$. \textbf{MIXEDLOGN($p,\mu_1,\sigma_1,1-p,\mu_2,\sigma_2$)}: A mixed distribution that is lognormal with location parameter $\mu_1$ and scale parameter $\sigma_1$ with probability $p$ and lognormal with location parameter $\mu_2$ and scale parameter $\sigma_2$ with probability $1-p$.  \textbf{WEIB($\beta,\alpha$)}: Weibull distribution with shape parameter $\alpha$ and scale parameter $\beta$.

For the first-stage (ED workup) duration, we assume it depends both on the patient's ESI level and age group. To account for variations in resource availability and operational conditions, we also allow the distribution to depend on the start time of the first stage (i.e., rooming time).  For similar time periods or categories with limited data, we fit a single distribution; otherwise, we fit separate distributions when data is sufficient and patterns differ. Table~\ref{app:ase:workup-duration} summarizes the distributions used. For the second stage, admitted patients’ durations are jointly determined by hospital bed availability and TPP, as discussed earlier. For discharged patients, we approximate stage durations using fitted probability distributions, assuming they depend on age group and ESI level. For similar time periods or categories with limited data, we fit a single distribution; when sufficient data is available and patterns differ, we fit separate distributions by time of day and patient category. Table~\ref{app:ase:second-stage-discharge} provides the distributions used.

\begin{table}[!ht]
\centering
\caption{Distribution of First-Stage (ED Workup) Durations.}
\label{app:ase:workup-duration}
{
\begin{tabular}{cccc}
\hline
\multicolumn{4}{c}{\textbf{Adult Patients}} \\
\hline
\begin{tabular}[c]{@{}c@{}}ESI \\ Level\end{tabular} & Time of Day & Probability Distribution & $p$-value \\
\hline
1 & 12:00 am - 11:59 pm & 10+GAMM(157,1.02) & 0.14 \\\hline
2 & 12:00 am - 5:59 am  & MIXEDLOGN(0.878,5.023,0.893,0.122,6.631,0.254) & $>$0.15 \\
2 & 6:00 am - 0:59 pm   & MIXEDLOGN(0.354,5.113,0.986,0.646,5.341,0.653) & $>$0.15 \\
2 & 1:00 pm - 6:59 pm   & MIXEDLOGN(0.014,3.032,0.328,0.986,5.439,0.873) & $<$0.01 \\
2 & 7:00 pm - 11:59 pm  & MIXEDLOGN(0.048,4.222,0.999,0.952,5.482,0.940) & $<$0.01 \\\hline
3 & 6:00 am - 10:59 pm  & MIXEDLOGN(0.675,5.467,0.489,0.325,4.865,0.854) & 0.07 \\
3 & 11:00 pm - 5:59 am  & MIXEDLOGN(0.094,3.794,0.666,0.906,5.234,0.636) & $>$0.15 \\\hline
4 & 2:00 am - 5:59 am   & MIXEDLOGN(0.036,2.933,0.295,0.964,4.734,0.809) & $>$0.15 \\
4 & 6:00 am - 4:59 pm   & MIXEDLOGN(0.100,3.149,0.423,0.900,4.860,0.701) & $>$0.15 \\
4 & 5:00 pm - 1:59 am   & MIXEDLOGN(0.812,4.577,0.920,0.188,5.136,0.422) & $>$0.15 \\\hline
5 & 8:00 am - 10:59 pm  & 10+EXP(96.2)                                   & $>$0.15 \\
5 & 11:00 pm - 7:59 am  & 10+WEIB(133,0.783)                             & $>$0.15 \\
\hline
\multicolumn{4}{c}{\textbf{Pediatric Patients}} \\
\hline
\begin{tabular}[c]{@{}c@{}}ESI \\ Level\end{tabular} & Time of Day & Probability Distribution & $p$-value \\
\hline
1 & 12:00 am - 11:59 pm & 10+WEIB(104,0.972)                             & $>$0.15 \\\hline
2 & 3:00 am - 12:59 pm  & MIXEDLOGN(0.555,4.875,0.933,0.445,5.318,0.536) & $>$0.15 \\
2 & 1:00 pm - 2:59 am   & MIXEDLOGN(0.312,5.180,0.384,0.688,4.994,0.964) & $>$0.15 \\\hline
3 & 5:00 am - 10:59 pm  & MIXEDLOGN(0.514,5.268,0.487,0.486,4.731,0.776) & $>$0.15 \\
3 & 11:00 pm - 4:59 am  & MIXEDLOGN(0.639,4.922,0.568,0.361,4.750,0.926) & $>$0.15 \\\hline
4 & 12:00 am - 11:59 pm & MIXEDLOGN(0.896,4.482,0.734,0.104,4.828,0.325) & $>$0.15 \\\hline
5 & 12:00 am - 11:59 pm & 10+GAMM(41.7,1.62)                             & $>$0.15 \\
\hline
\end{tabular}
}
\end{table}

\begin{table}[!ht]
\centering
\caption{Distribution of Second-Stage Durations.}
\label{app:ase:second-stage-discharge}
{
\begin{tabular}{cccc}
\hline
\multicolumn{4}{c}{\textbf{Adult Discharge Patients}} \\
\hline
\begin{tabular}[c]{@{}c@{}}ESI \\ Level\end{tabular} & Time of Day & Probability Distribution & $p$-value \\
\hline
1--2 & 3:00 am - 8:59 am   & MIXEDLOGN(0.139,5.976,0.494,0.861,3.420,0.855) & $>$0.15 \\
1--2 & 9:00 am - 2:59 am   & MIXEDLOGN(0.614,3.199,0.608,0.386,4.182,1.439) & $>$0.15 \\\hline
3    & 5:00 am - 10:59 am  & MIXEDLOGN(0.359,3.881,1.405,0.641,3.061,0.622) & $>$0.15 \\
3    & 11:00 am - 4:59 am  & MIXEDLOGN(0.759,3.129,0.651,0.241,3.778,1.565) & $>$0.15 \\\hline
4--5 & 12:00 am - 11:59 am & MIXEDLOGN(0.692,3.020,0.660,0.308,3.500,1.593) & $>$0.15 \\
4--5 & 12:00 pm - 11:59 pm & MIXEDLOGN(0.751,2.985,0.663,0.249,3.048,1.559) & $>$0.15 \\
\hline
\multicolumn{4}{c}{\textbf{Pediatric Discharge Patients}} \\
\hline
\begin{tabular}[c]{@{}c@{}}ESI \\ Level\end{tabular} & Time of Day & Probability Distribution & $p$-value \\
\hline
1--2 & 12:00 am - 10:59 am & MIXEDLOGN(0.757,3.038,0.622,0.243,3.619,1.596) & $>$0.15 \\
1--2 & 11:00 am - 11:59 pm & MIXEDLOGN(0.033,5.618,0.668,0.967,3.076,0.790) & $>$0.15 \\\hline
3    & 12:00 am - 11:59 pm & MIXEDLOGN(0.150,2.898,1.690,0.850,2.956,0.706) & $>$0.15 \\\hline
4--5 & 12:00 am - 11:59 pm & MIXEDLOGN(0.086,1.924,1.490,0.914,2.865,0.705) & $>$0.15 \\
\hline
\end{tabular}
}
\end{table}

\begin{figure}[!ht]
    \centering
    \includegraphics[scale=0.6]{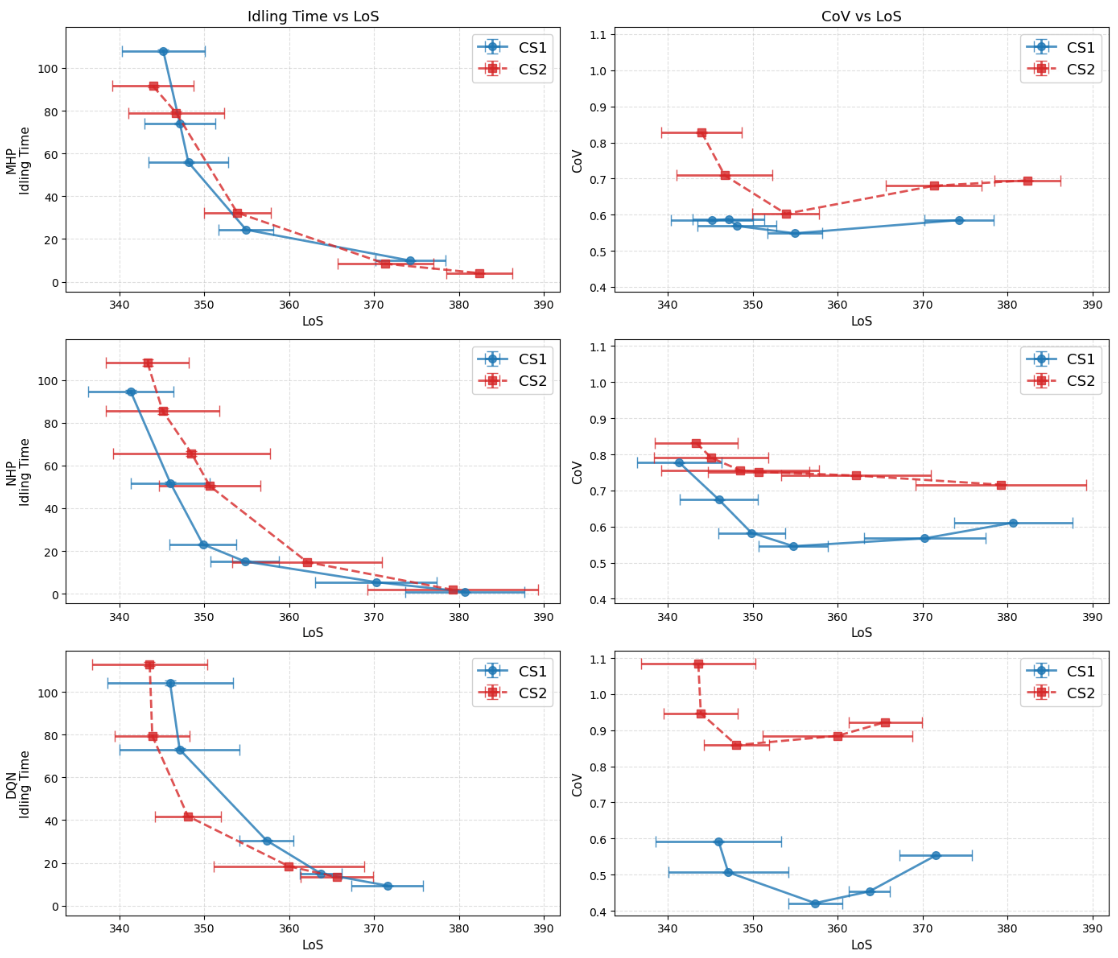}
    \caption{Performance comparison of MHP, NHP, and DQN under $\mathcal{C}_1$ (CS1) versus $\mathcal{C}_2$  (CS2) for SOS.}
    \label{fig:R3}
\end{figure}

\subsection{Simulation Results}\label{sec.results-tables}

We first present the results under SOS and $\mathcal{C}_2$ discussed in Section \ref{S7:NS:SS1}.  Figure~\ref{fig:R3} presents a comparison of performances of MHP, NHP, and DQN under $\mathcal{C}_1$ and $\mathcal{C}_2$ within SOS. (CP, GP, and EP are not affected by the cost parameters, and hence, are omitted.)    Under both $\mathcal{C}_1$ and $\mathcal{C}_2$, the proposed policies yield substantial reductions in average ED LoS, and the trade-off frontier between bed idling time and LoS appears largely unaffected by the shift from $\mathcal{C}_1$ to $\mathcal{C}_2$. However, $\mathcal{C}_2$ leads to markedly higher coefficients of variation across all three policies. 

 We next provide more detailed information about the simulation results presented in Figures \ref{fig:R1} through \ref{fig:R4} as well as plots under additional experimental conditions. In the following tables, we provide numerical values of all confidence intervals that appear in the referenced plots. We also provide the tuning parameters corresponding to each heuristic policy. All heuristics depicted in Table~\ref{tab:sosp} are tuned using cost functions as in Table~\ref{tab:sos-cs1}, except for the first row for the MHP heuristic. For this case, we used a modified cost function $\tilde{C}(\phi_{t, g}, a_{t, g}) = a_{t, g}^2 + C \phi_{t, g}^2 \mathbb{I}_{\{\phi_{t, g} \geq 0\}} + D \phi_{t, g}^2 \mathbb{I}_{\{\phi_{t, g} < 0\}}$, where different weights are given to different directions of imbalance. (A positive/negative value of $\phi_{t, g}$ corresponds to the presence of excess boarding patients/idle beds.) This adjustment was made as the original cost function $\tilde{C}(\phi_{t, g}, a_{t, g}) = a_{t, g}^2 + C \phi_{t, g}^2 + D \phi_{t, g}$, did not extend to cover the whole range of average idling times observed by other heuristics. The rest of the analysis proceeded as in the other cases.

\begin{table}[H]
\centering
{\fontsize{9}{8}\selectfont
\caption{95\% Confidence Intervals on Performance Metrics under SOS and $\mathcal{C}_1$.}
\label{tab:sos-cs1}
\begin{tabular}{||l|c||c|c|c|c|c||}
\hline
\textbf{Policy} & \textbf{Tuning Parameters} & \textbf{LoS} & \textbf{LoS (Admit)} & \textbf{Boarding Time} & \textbf{Idling Time} & \textbf{CoV} \\ \hline\hline
Current & --& 406.57 $\pm$ 3.27 & 549.79 $\pm$ 4.89 & 248.26 $\pm$ 2.14 & 0.00 $\pm$ 0.00 & 0.76 $\pm$ 0.00 \\
Extreme & --& 342.17 $\pm$ 4.64 & 369.52 $\pm$ 7.18 & 76.68 $\pm$ 4.98 & 356.73 $\pm$ 5.06 & 1.29 $\pm$ 0.00 \\
Greedy & --& 349.00 $\pm$ 2.50 & 389.05 $\pm$ 4.24 & 95.81 $\pm$ 2.47 & 34.56 $\pm$ 0.21 & 0.89 $\pm$ 0.01 \\ \hline
MHP & $\tilde \alpha=1.4, \tilde \beta=18$ & 345.15 $\pm$ 4.86 & 380.10 $\pm$ 5.82 & 84.85 $\pm$ 3.06 & 108.06 $\pm$ 0.71 & 0.58 $\pm$ 0.00 \\
MHP &$\alpha=1.6, \beta=12$ & 347.11 $\pm$ 4.19 & 382.62 $\pm$ 5.56 & 89.29 $\pm$ 2.97 & 74.01 $\pm$ 0.47 & 0.59 $\pm$ 0.00 \\
MHP & $\tilde \alpha=0.5, \tilde \beta=5$ & 348.11 $\pm$ 4.65 & 382.72 $\pm$ 5.47 & 91.34 $\pm$ 2.81 & 55.90 $\pm$ 0.35 & 0.57 $\pm$ 0.00 \\
MHP & $\tilde \alpha=0.5, \tilde \beta=2$ & 354.91 $\pm$ 3.25 & 410.10 $\pm$ 4.28 & 117.74 $\pm$ 2.44 & 24.35 $\pm$ 0.14 & 0.55 $\pm$ 0.00 \\
MHP & $\tilde \alpha=0.5, \tilde \beta=0$ & 374.27 $\pm$ 4.10 & 452.95 $\pm$ 5.09 & 155.02 $\pm$ 2.72 & 9.88 $\pm$ 0.10 & 0.59 $\pm$ 0.00 \\ \hline
NHP & $\psi=8,\tilde{c}_b=256,\tilde{c}_e=1$ & 341.33 $\pm$ 5.00 & 368.95 $\pm$ 9.83 & 76.76 $\pm$ 7.40 & 94.52 $\pm$ 0.74 & 0.78 $\pm$ 0.00 \\
NHP & $\psi=8,\tilde{c}_b=32,\tilde{c}_e=1$ & 345.99 $\pm$ 4.62 & 381.04 $\pm$ 8.76 & 87.99 $\pm$ 5.25 & 51.81 $\pm$ 0.73 & 0.67 $\pm$ 0.00 \\
NHP & $\psi=8,\tilde{c}_b=12,\tilde{c}_e=1$ & 349.83 $\pm$ 3.93 & 395.52 $\pm$ 8.04 & 103.40 $\pm$ 4.75 & 23.12 $\pm$ 0.29 & 0.58 $\pm$ 0.00 \\
NHP & $\psi=8,\tilde{c}_b=10,\tilde{c}_e=2$ & 354.76 $\pm$ 4.05 & 411.70 $\pm$ 7.32 & 120.19 $\pm$ 4.48 & 15.13 $\pm$ 0.21 & 0.55 $\pm$ 0.00 \\
NHP & $\psi=8,\tilde{c}_b=10,\tilde{c}_e=16$ & 370.21 $\pm$ 7.15 & 442.10 $\pm$ 10.21 & 144.81 $\pm$ 4.99 & 5.43 $\pm$ 0.08 & 0.57 $\pm$ 0.00 \\
NHP & $\psi=8,\tilde{c}_b=10,\tilde{c}_e=128$ & 380.66 $\pm$ 6.98 & 475.84 $\pm$ 10.28 & 178.68 $\pm$ 5.21 & 0.93 $\pm$ 0.02 & 0.61 $\pm$ 0.00 \\ \hline
DQN & $c_b=10, c_e=1$ & 345.96 $\pm$ 7.40 & 377.79 $\pm$ 11.20 & 82.84 $\pm$ 5.60 & 104.06 $\pm$ 1.14 & 0.59 $\pm$ 0.00 \\
DQN & $c_b=4, c_e=1$ & 347.09 $\pm$ 7.04 & 382.16 $\pm$ 4.62 & 89.72 $\pm$ 4.15 & 73.05 $\pm$ 0.50 & 0.51 $\pm$ 0.00 \\
DQN & $c_b=1, c_e=1$ & 357.31 $\pm$ 3.17 & 417.54 $\pm$ 4.67 & 125.24 $\pm$ 2.82 & 30.52 $\pm$ 0.22 & 0.42 $\pm$ 0.00 \\
DQN & $c_b=1, c_e=4$ & 363.70 $\pm$ 2.43 & 439.00 $\pm$ 4.90 & 147.17 $\pm$ 2.96 & 15.05 $\pm$ 0.09 & 0.45 $\pm$ 0.00 \\
DQN & $c_b=1, c_e=10$ & 371.52 $\pm$ 4.24 & 456.54 $\pm$ 4.84 & 162.29 $\pm$ 2.70 & 9.40 $\pm$ 0.04 & 0.55 $\pm$ 0.00 \\ \hline
\end{tabular}
}
\end{table}

\begin{table}[H]
\centering
\fontsize{9}{8}\selectfont
\caption{95\% Confidence Intervals on Performance Metrics under SOS and $\mathcal{C}_2$.}
\begin{tabular}{||l|c||c|c|c|c|c||}
\hline
\textbf{Policy} & \textbf{Tuning Parameters} & \textbf{LoS} & \textbf{LoS (Admit)} & \textbf{Boarding Time} & \textbf{Idling Time} & \textbf{CoV} \\ \hline\hline
MHP & $\tilde \alpha=1.4, \tilde \beta=18$ & 343.95 $\pm$ 4.77 & 372.71 $\pm$ 6.05 & 79.36 $\pm$ 3.22 & 91.52 $\pm$ 0.59 & 0.83 $\pm$ 0.00 \\
MHP & $\tilde \alpha=1.6, \tilde \beta=12$ & 346.67 $\pm$ 5.64 & 377.69 $\pm$ 6.80 & 83.18 $\pm$ 3.79 & 79.05 $\pm$ 0.62 & 0.71 $\pm$ 0.00 \\
MHP & $\tilde \alpha=0.5, \tilde \beta=5$ & 353.87 $\pm$ 3.91 & 401.08 $\pm$ 5.09 & 106.93 $\pm$ 3.51 & 32.32 $\pm$ 0.34 & 0.60 $\pm$ 0.00 \\
MHP & $\tilde \alpha=0.5, \tilde \beta=2$ & 371.31 $\pm$ 5.64 & 445.97 $\pm$ 6.10 & 148.77 $\pm$ 3.25 & 8.44 $\pm$ 0.06 & 0.68 $\pm$ 0.00 \\
MHP & $\tilde \alpha=0.5, \tilde \beta=0$ & 382.37 $\pm$ 3.89 & 476.15 $\pm$ 5.02 & 177.66 $\pm$ 2.56 & 4.06 $\pm$ 0.04 & 0.69 $\pm$ 0.00 \\ \hline
NHP & $\psi=8,\tilde{c}_b=2560,\tilde{c}_e=10$ & 343.30 $\pm$ 4.86 & 370.71 $\pm$ 11.81 & 77.10 $\pm$ 6.83 & 108.04 $\pm$ 1.54 & 0.83 $\pm$ 0.00 \\
NHP & $\psi=8,\tilde{c}_b=320,\tilde{c}_e=10$ & 345.09 $\pm$ 6.71 & 377.45 $\pm$ 11.58 & 80.64 $\pm$ 7.38 & 85.42 $\pm$ 1.39 & 0.79 $\pm$ 0.00 \\
NHP & $\psi=8,\tilde{c}_b=120,\tilde{c}_e=10$ & 348.49 $\pm$ 9.27 & 377.50 $\pm$ 11.54 & 85.10 $\pm$ 7.92 & 65.60 $\pm$ 1.14 & 0.76 $\pm$ 0.00 \\
NHP & $\psi=8,\tilde{c}_b=100,\tilde{c}_e=20$ & 350.65 $\pm$ 5.96 & 382.37 $\pm$ 10.71 & 84.39 $\pm$ 5.73 & 50.34 $\pm$ 0.67 & 0.75 $\pm$ 0.00 \\
NHP & $\psi=8,\tilde{c}_b=100,\tilde{c}_e=160$ & 362.13 $\pm$ 8.82 & 412.52 $\pm$ 10.25 & 113.47 $\pm$ 6.28 & 14.76 $\pm$ 0.30 & 0.74 $\pm$ 0.00 \\
NHP & $\psi=8,\tilde{c}_b=100,\tilde{c}_e=1280$ & 379.22 $\pm$ 10.03 & 463.50 $\pm$ 9.93 & 163.93 $\pm$ 5.33 & 1.84 $\pm$ 0.03 & 0.72 $\pm$ 0.00 \\ \hline
DQN & $c_b=100, c_e=10$ & 343.56 $\pm$ 6.74 & 371.44 $\pm$ 7.46 & 77.89 $\pm$ 4.02 & 112.97 $\pm$ 0.86 & 1.08 $\pm$ 0.00 \\
DQN & $c_b=40, c_e=10$ & 343.86 $\pm$ 4.38 & 374.16 $\pm$ 4.82 & 81.16 $\pm$ 2.70 & 79.44 $\pm$ 0.47 & 0.95 $\pm$ 0.00 \\
DQN & $c_b=10, c_e=10$ & 348.06 $\pm$ 3.88 & 389.70 $\pm$ 4.34 & 97.52 $\pm$ 2.38 & 41.84 $\pm$ 0.20 & 0.86 $\pm$ 0.00 \\
DQN & $c_b=10, c_e=40$ & 359.96 $\pm$ 8.85 & 415.23 $\pm$ 10.57 & 119.28 $\pm$ 5.36 & 18.40 $\pm$ 0.18 & 0.88 $\pm$ 0.00 \\
DQN & $c_b=10, c_e=100$ & 365.58 $\pm$ 4.28 & 434.27 $\pm$ 5.28 & 138.94 $\pm$ 4.69 & 13.39 $\pm$ 0.07 & 0.92 $\pm$ 0.00 \\ \hline
\end{tabular}
\end{table}

\begin{table}[H]
\centering
\fontsize{9}{8}\selectfont
\caption{95\% Confidence Intervals on Performance Metrics under SOS-P and $\mathcal{C}_1$.}
\label{tab:sosp}
\begin{tabular}{||l|c||c|c|c|c|c||}
\hline
\textbf{Policy} & \textbf{Tuning Parameters} & \textbf{LoS} & \textbf{LoS (Admit)} & \textbf{Boarding Time} & \textbf{Idling Time} & \textbf{CoV} \\ \hline\hline
Current & --& 406.57 $\pm$ 3.27 & 549.79 $\pm$ 4.89 & 248.26 $\pm$ 2.14 & 0.00 $\pm$ 0.00 & 0.76 $\pm$ 0.00 \\
Extreme & -- & $341.63 \pm 4.32$ & $366.61 \pm 6.25$ & $73.58 \pm 3.25$ & $338.55	\pm 2.48$ & $1.47 \pm 0.00$\\
Greedy &-- & 354.47 $\pm$ 3.65 & 400.42 $\pm$ 5.36 & 105.55 $\pm$ 3.15 & 23.44 $\pm$ 0.19 & 0.89 $\pm$ 0.01 \\ \hline
MHP & $\tilde \alpha=0.75, \tilde \beta=18$ & 359.24 $\pm$ 3.16 & 425.12 $\pm$ 4.66 & 133.56 $\pm$ 2.71 & 45.62 $\pm$ 0.42 & 0.58 $\pm$ 0.00 \\
MHP & $\tilde \alpha=1.4, \tilde \beta=18$ & 368.26 $\pm$ 3.88 & 446.70 $\pm$ 5.25 & 152.73 $\pm$ 2.77 & 15.29 $\pm$ 0.14 & 0.52 $\pm$ 0.00 \\
MHP & $\tilde \alpha=1.6, \tilde \beta=12$ & 370.91 $\pm$ 5.03 & 450.67 $\pm$ 5.32 & 155.70 $\pm$ 2.43 & 12.57 $\pm$ 0.09 & 0.48 $\pm$ 0.00 \\
MHP & $\tilde \alpha=0.5, \tilde \beta=5$ & 373.02 $\pm$ 3.95 & 454.37 $\pm$ 5.24 & 158.35 $\pm$ 2.98 & 12.35 $\pm$ 0.12 & 0.53 $\pm$ 0.00 \\
MHP & $\tilde \alpha=0.5, \tilde \beta=2$ & 381.70 $\pm$ 4.01 & 480.89 $\pm$ 4.52 & 184.38 $\pm$ 2.41 & 2.95 $\pm$ 0.03 & 0.55 $\pm$ 0.00 \\
MHP & $\tilde \alpha=0.5, \tilde \beta=0$ & 385.93 $\pm$ 3.41 & 498.74 $\pm$ 4.35 & 203.87 $\pm$ 2.42 & 1.01 $\pm$ 0.01 & 0.60 $\pm$ 0.00 \\\hline
NHP & $\psi=8,\tilde{c}_b=96,\tilde{c}_e=1$ & 347.31 $\pm$ 6.71 & 381.64 $\pm$ 9.76 & 87.68 $\pm$ 5.23 & 51.91 $\pm$ 0.64 & 0.66 $\pm$ 0.00 \\
NHP & $\psi=8,\tilde{c}_b=32,\tilde{c}_e=1$ & 351.11 $\pm$ 7.63 & 393.32 $\pm$ 10.30 & 99.16 $\pm$ 5.77 & 28.65 $\pm$ 0.40 & 0.61 $\pm$ 0.00 \\
NHP & $\psi=8,\tilde{c}_b=12,\tilde{c}_e=1$ & 370.62 $\pm$ 8.05 & 445.51 $\pm$ 9.07 & 148.51 $\pm$ 4.58 & 7.38 $\pm$ 0.12 & 0.67 $\pm$ 0.00 \\
NHP & $\psi=8,\tilde{c}_b=10,\tilde{c}_e=2$ & 375.85 $\pm$ 10.29 & 464.16 $\pm$ 11.86 & 168.10 $\pm$ 5.64 & 3.99 $\pm$ 0.07 & 0.72 $\pm$ 0.00 \\
NHP & $\psi=8,\tilde{c}_b=10,\tilde{c}_e=16$ & 382.32 $\pm$ 6.78 & 484.13 $\pm$ 8.49 & 188.07 $\pm$ 4.11 & 1.33 $\pm$ 0.03 & 0.73 $\pm$ 0.00 \\
NHP & $\psi=8,\tilde{c}_b=10,\tilde{c}_e=128$ & 391.05 $\pm$ 8.74 & 510.33 $\pm$ 9.79 & 213.43 $\pm$ 5.39 & 0.23 $\pm$ 0.01 & 0.60 $\pm$ 0.00 \\ \hline
DQN & $c_b=20, c_e=1$ & 343.45 $\pm$ 4.36 & 401.63 $\pm$ 5.14 & 108.84 $\pm$ 2.38 & 44.29 $\pm$ 0.28 & 0.43 $\pm$ 0.00 \\
DQN & $c_b=10, c_e=1$ & 352.75 $\pm$ 4.52 & 403.69 $\pm$ 5.50 & 109.72 $\pm$ 2.63 & 39.73 $\pm$ 0.24 & 0.49 $\pm$ 0.00 \\
DQN & $c_b=4, c_e=1$ & 354.50 $\pm$ 3.87 & 446.16 $\pm$ 4.59 & 152.55 $\pm$ 2.44 & 12.46 $\pm$ 0.10 & 0.48 $\pm$ 0.00 \\
DQN & $c_b=1, c_e=1$ & 367.72 $\pm$ 3.79 & 502.28 $\pm$ 4.28 & 206.63 $\pm$ 2.48 & 1.75 $\pm$ 0.13 & 0.61 $\pm$ 0.00 \\
DQN & $c_b=1, c_e=4$ & 387.61 $\pm$ 3.95 & 505.39 $\pm$ 4.54 & 208.93 $\pm$ 2.37 & 1.65 $\pm$ 0.02 & 0.61 $\pm$ 0.00 \\
DQN & $c_b=1, c_e=10$ & 389.68 $\pm$ 3.19 & 513.47 $\pm$ 5.03 & 210.28 $\pm$ 2.12 & 1.25 $\pm$ 0.02 & 0.61 $\pm$ 0.00 \\\hline
\end{tabular}
\end{table}

\begin{table}[H]
\centering
\caption{95\% Confidence Intervals on Performance Metrics under PRS.}
\begin{tabular}{||l|c||c|c|c|c||}
\hline
\textbf{Policy}& \textbf{Tuning Parameters} & \textbf{LoS} & \textbf{LoS (Admit)} & \textbf{Boarding Time} & \textbf{Idling Time} \\ \hline\hline
Current & -- & 460.20 $\pm$ 6.25 & 596.00 $\pm$ 6.50 & 253.50 $\pm$ 2.49 & 0.00 $\pm$ 0.00 \\
Extreme & --& 388.70 $\pm$ 6.92 & 416.59 $\pm$ 8.02 & 88.24 $\pm$ 3.92 & 341.01 $\pm$ 4.11 \\
Greedy & --& 397.80 $\pm$ 6.25 & 438.37 $\pm$ 7.13 & 108.34 $\pm$ 3.44 & 32.95 $\pm$ 0.35 \\ \hline
MHP & $\tilde \alpha=1.4, \tilde \beta=18$ & 394.20 $\pm$ 6.65 & 435.04 $\pm$ 7.75 & 102.75 $\pm$ 3.81 & 94.24 $\pm$ 0.94 \\
MHP & $\tilde \alpha=1.6, \tilde \beta=12$ & 397.00 $\pm$ 6.40 & 435.86 $\pm$ 7.55 & 108.78 $\pm$ 3.82 & 64.04 $\pm$ 0.62 \\
MHP & $\tilde \alpha=0.5, \tilde \beta=5$ & 398.63 $\pm$ 6.61 & 438.56 $\pm$ 7.63 & 109.48 $\pm$ 3.74 & 49.68 $\pm$ 0.50 \\
MHP & $\tilde \alpha=0.5, \tilde \beta=2$ & 400.78 $\pm$ 6.32 & 457.79 $\pm$ 7.21 & 130.85 $\pm$ 3.38 & 22.54 $\pm$ 0.22 \\
MHP & $\tilde \alpha=0.5, \tilde \beta=0$ & 421.27 $\pm$ 6.44 & 499.00 $\pm$ 7.29 & 164.59 $\pm$ 3.39 & 10.23 $\pm$ 0.16 \\ \hline
NHP & $\psi=8,\tilde{c}_b=256,\tilde{c}_e=1$ & 390.27 $\pm$ 6.58 & 421.79 $\pm$ 7.76 & 91.78 $\pm$ 3.85 & 99.61 $\pm$ 1.05 \\
NHP & $\psi=8,\tilde{c}_b=32,\tilde{c}_e=1$ & 392.72 $\pm$ 6.81 & 425.87 $\pm$ 7.97 & 96.31 $\pm$ 3.92 & 61.59 $\pm$ 0.67 \\
NHP & $\psi=8,\tilde{c}_b=12,\tilde{c}_e=1$ & 395.26 $\pm$ 7.27 & 436.30 $\pm$ 8.24 & 107.96 $\pm$ 3.80 & 33.52 $\pm$ 0.35 \\
NHP & $\psi=8,\tilde{c}_b=10,\tilde{c}_e=2$ & 403.37 $\pm$ 7.85 & 452.53 $\pm$ 8.71 & 121.48 $\pm$ 3.92 & 24.05 $\pm$ 0.26 \\
NHP & $\psi=8,\tilde{c}_b=10,\tilde{c}_e=16$ & 405.28 $\pm$ 6.45 & 467.88 $\pm$ 7.42 & 139.53 $\pm$ 3.48 & 10.17 $\pm$ 0.11 \\
NHP & $\psi=8,\tilde{c}_b=10,\tilde{c}_e=128$ & 424.47 $\pm$ 8.08 & 507.29 $\pm$ 8.51 & 172.55 $\pm$ 3.46 & 2.43 $\pm$ 0.04 \\ \hline
DQN & $c_b=10, c_e=1$ & 387.92 $\pm$ 5.81 & 423.29 $\pm$ 7.05 & 98.05 $\pm$ 3.74 & 92.93 $\pm$ 0.83 \\
DQN & $c_b=4, c_e=1$  & 394.25 $\pm$ 6.91 & 437.21 $\pm$ 8.14 & 110.61 $\pm$ 4.06 & 62.21 $\pm$ 0.56 \\
DQN & $c_b=1, c_e=1$ & 418.85 $\pm$ 10.00 & 482.44 $\pm$ 10.42 & 145.12 $\pm$ 3.84 & 26.12 $\pm$ 0.31 \\
DQN & $c_b=1, c_e=4$ & 420.88 $\pm$ 7.27 & 500.01 $\pm$ 7.85 & 166.78 $\pm$ 3.38 & 12.27 $\pm$ 0.13 \\
DQN & $c_b=1, c_e=10$ & 425.50 $\pm$ 8.04 & 512.07 $\pm$ 8.35 & 177.62 $\pm$ 3.37 & 7.89 $\pm$ 0.10 \\ \hline
\end{tabular}
\end{table}

%% file: abert_bib.bib
@article{american2011definition,
  title={Definition of boarded patient},
  author={{American College of Emergency Physicians (ACEP)}},
  journal={Annals of Emergency Medicine},
  volume={57},
  number={5},
  pages={548},
  year={2011}
}

@article{viccellio2013patients,
  title={Patients overwhelmingly prefer inpatient boarding to emergency department boarding},
  author={Viccellio, Peter and Zito, Joseph A. and Sayage, Valerie and Chohan, Jasmine and Garra, Gregory and Santora, Carolyn and Singer, Adam J.},
  journal={The Journal of Emergency Medicine},
  volume={45},
  number={6},
  pages={942--946},
  year={2013},
  publisher={Elsevier}
}

@article{singer2011association,
  title={The association between length of emergency department boarding and mortality},
  author={Singer, Adam J. and Thode Jr., Henry C. and Viccellio, Peter and Pines, Jesse M.},
  journal={Academic Emergency Medicine},
  volume={18},
  number={12},
  pages={1324--1329},
  year={2011},
  publisher={Wiley Online Library}
}

@article{chalfin2007impact,
  title={Impact of delayed transfer of critically ill patients from the emergency department to the intensive care unit},
  author={Chalfin, Donald B. and Trzeciak, Stephen and Likourezos, Antonios and Baumann, Brigitte M. and Dellinger, R. Phillip and Delay-Ed Study Group},
  journal={Critical Care Medicine},
  volume={35},
  number={6},
  pages={1477--1483},
  year={2007},
  publisher={LWW}
}

@article{white2013boarding,
  title={Boarding inpatients in the emergency department increases discharged patient length of stay},
  author={White, Benjamin A. and Biddinger, Paul D. and Chang, Yuchiao and Grabowski, Beth and Carignan, Sarah and Brown, David F. M.},
  journal={The Journal of Emergency Medicine},
  volume={44},
  number={1},
  pages={230--235},
  year={2013},
  publisher={Elsevier}
}

@article{bair2010impact,
  title={The impact of inpatient boarding on ED efficiency: a discrete-event simulation study},
  author={Bair, Aaron E. and Song, Wheyming T. and Chen, Yi-Chun and Morris, Beth A.},
  journal={Journal of Medical Systems},
  volume={34},
  number={5},
  pages={919--929},
  year={2010},
  publisher={Springer}
}

@article{chen2022using,
  title={Using Hospital Admission Predictions at Triage for Improving Patient Length of Stay in Emergency Departments},
  author={Chen, Wanyi and Argon, Nilay T. and Bohrmann, Tommy and Linthicum, Benjamin and Lopiano, Kenneth and Mehrotra, Abhishek and Travers, Debbie and Ziya, Serhan},
  journal={Operations Research},
  volume={71},
  number={5},
  pages={1733--1755},
  year={2023},
  publisher={INFORMS},
  doi={10.1287/opre.2022.2405}
}

@article{zhang2017prediction,
  title={Prediction of emergency department hospital admission based on natural language processing and neural networks},
  author={Zhang, Xingyu and Kim, Joyce and Patzer, Rachel E. and Pitts, Stephen R. and Patzer, Aaron and Schrager, Justin D.},
  journal={Methods of Information in Medicine},
  volume={56},
  number={05},
  pages={377--389},
  year={2017},
  publisher={Schattauer GmbH}
}

@article{roquette2020prediction,
  title={Prediction of admission in pediatric emergency department with deep neural networks and triage textual data},
  author={Roquette, Bruno P. and Nagano, Hitoshi and Marujo, Ernesto C. and Maiorano, Alexandre C.},
  journal={Neural Networks},
  volume={126},
  pages={170--177},
  year={2020},
  publisher={Elsevier}
}

@inproceedings{azari2015imbalanced,
  title={Imbalanced learning to predict long stay emergency department patients},
  author={Azari, Ali and Janeja, Vandana P and Levin, Scott},
  booktitle={{2015 IEEE International Conference on Bioinformatics and Biomedicine (BIBM)}},
  pages={807--814},
  year={2015},
  organization={IEEE}
}

@article{rahman2020using,
  title={Using data mining to predict emergency department length of stay greater than 4 hours: Derivation and single-site validation of a decision tree algorithm},
  author={Rahman, Md Anisur and Honan, Bridget and Glanville, Thomas and Hough, Peter and Walker, Katie},
  journal={Emergency Medicine Australasia},
  volume={32},
  number={3},
  pages={416--421},
  year={2020},
  publisher={Wiley Online Library}
}

@article{armony2015patient,
  title={On patient flow in hospitals: A data-based queueing-science perspective},
  author={Armony, Mor and Israelit, Shlomo and Mandelbaum, Avishai and Marmor, Yariv N. and Tseytlin, Yulia and Yom-Tov, Galit B.},
  journal={Stochastic Systems},
  volume={5},
  number={1},
  pages={146--194},
  year={2015},
  publisher={INFORMS}
}

@article{qiu2015cost,
  title={A cost-sensitive inpatient bed reservation approach to reduce emergency department boarding times},
  author={Qiu, Shanshan and Chinnam, Ratna Babu and Murat, Alper and Batarse, Bassam and Neemuchwala, Hakimuddin and Jordan, Will},
  journal={Health Care Management Science},
  volume={18},
  number={1},
  pages={67--85},
  year={2015},
  publisher={Springer}
}

@article{shi2016models,
  title={Models and insights for hospital inpatient operations: Time-dependent {ED} boarding time},
  author={Shi, Pengyi and Chou, Mabel C. and Dai, Jim G. and Ding, Ding and Sim, Joe},
  journal={Management Science},
  volume={62},
  number={1},
  pages={1--28},
  year={2016},
  publisher={INFORMS}
}

@book{sutton2018reinforcement,
	title={Reinforcement learning: An introduction},
	author={Sutton, Richard S. and Barto, Andrew G.},
	year={2018},
	publisher={MIT press}
}

@article{hoot2008systematic,
	title={Systematic review of emergency department crowding: causes, effects, and solutions},
	author={Hoot, Nathan R. and Aronsky, Dominik},
	journal={Annals of Emergency Medicine},
	volume={52},
	number={2},
	pages={126--136},
	year={2008},
	publisher={Elsevier}
}

@article{shi2021timing,
  title={Timing it right: Balancing inpatient congestion vs. readmission risk at discharge},
  author={Shi, Pengyi and Helm, Jonathan E. and Deglise-Hawkinson, Jivan and Pan, Julian},
  journal={Operations Research},
  volume={69},
  number={6},
  pages={1842--1865},
  year={2021},
  publisher={INFORMS}
}

@book{watkins1989learning,
  title={Learning from delayed rewards},
  author={Watkins, Christopher J. C. H.},
  year={1989},
  note={PhD Thesis, King's College Oxford},
  publisher={King's College, Cambridge United Kingdom}
}

@inproceedings{van2016deep,
  title={Deep reinforcement learning with double {Q}-learning},
  author={Van Hasselt, Hado and Guez, Arthur and Silver, David},
  booktitle = {Proceedings of the Thirtieth AAAI Conference on Artificial Intelligence},
 pages = {2094–2100},
  year={2016}
}

@techreport{gilboy2012emergency,
  title={{Emergency {S}everity {I}ndex ({ESI}): A Triage Tool for Emergency Department Care, Version 4. Implementation Handbook 2012 Edition}},
  author={Gilboy, Nicki and Tanabe, Paula and Travers, Debbie and Rosenau, Alexander M.},
  institution={Agency for Healthcare Research and Quality},
  year={2012}
}

@article{mnih2015human,
author = {Mnih, Volodymyr and Kavukcuoglu, Koray and Silver, David and Rusu, Andrei and Veness, Joel and Bellemare, Marc and Graves, Alex and Riedmiller, Martin and Fidjeland, Andreas and Ostrovski, Georg and Petersen, Stig and Beattie, Charles and Sadik, Amir and Antonoglou, Ioannis and King, Helen and Kumaran, Dharshan and Wierstra, Daan and Legg, Shane and Hassabis, Demis},
year = {2015},
month = {02},
pages = {529-533},
title = {Human-level control through deep reinforcement learning},
volume = {518},
number={7540},
journal = {Nature}
}

@article{peck2012predicting,
  title={Predicting emergency department inpatient admissions to improve same-day patient flow},
  author={Peck, Jordan S. and Benneyan, James C. and Nightingale, Deborah J. and Gaehde, Stephan A.},
  journal={Academic Emergency Medicine},
  volume={19},
  number={9},
  pages={E1045--E1054},
  year={2012},
  publisher={Wiley Online Library}
}

@article{somanchi2022predict,
  title={To Predict or Not to Predict: The Case of the Emergency Department},
  author={Somanchi, Sriram and Adjerid, Idris and Gross, Ralph},
  journal={Production and Operations Management},
  volume={31},
  number={2},
  pages={799--818},
  year={2022},
  publisher={Wiley Online Library}
}

@article{sun2011predicting,
  title={Predicting hospital admissions at emergency department triage using routine administrative data},
  author={Sun, Yan and Heng, Bee Hoon and Tay, Seow Yian and Seow, Eillyne},
  journal={Academic Emergency Medicine},
  volume={18},
  number={8},
  pages={844--850},
  year={2011},
  publisher={Wiley Online Library}
}

@article{lee2020prediction,
  title={Prediction of emergency department patient disposition decision for proactive resource allocation for admission},
  author={Lee, Seung-Yup and Chinnam, Ratna Babu and Dalkiran, Evrim and Krupp, Seth and Nauss, Michael},
  journal={Health Care Management Science},
  volume={23},
  pages={339--359},
  year={2020},
  publisher={Springer}
}

@article{lee2021proactive,
  title={Proactive coordination of inpatient bed management to reduce emergency department patient boarding},
  author={Lee, Seung-Yup and Chinnam, Ratna Babu and Dalkiran, Evrim and Krupp, Seth and Nauss, Michael},
  journal={International Journal of Production Economics},
  volume={231},
  pages={107842},
  year={2021},
  publisher={Elsevier}
}

@article{peck2014characterizing,
  title={Characterizing the value of predictive analytics in facilitating hospital patient flow},
  author={Peck, Jordan S. and Benneyan, James C. and Nightingale, Deborah J. and Gaehde, Stephan A.},
  journal={IIE Transactions on Healthcare Systems Engineering},
  volume={4},
  number={3},
  pages={135--143},
  year={2014},
  publisher={Taylor \& Francis}
}

@book{puterman2014markov,
  title={Markov Decision Processes: Discrete Stochastic Dynamic Programming},
  author={Puterman, Martin L},
  year={2014},
  publisher={John Wiley \& Sons}
}

@article{wiler2011review,
  title={Review of modeling approaches for emergency department patient flow and crowding research},
  author={Wiler, Jennifer L. and Griffey, Richard T. and Olsen, Tava},
  journal={Academic Emergency Medicine},
  volume={18},
  number={12},
  pages={1371--1379},
  year={2011},
  publisher={Wiley Online Library}
}

@article{saghafian2015operations,
  title={Operations research/management contributions to emergency department patient flow optimization: Review and research prospects},
  author={Saghafian, Soroush and Austin, Garrett and Traub, Stephen J.},
  journal={IIE Transactions on Healthcare Systems Engineering},
  volume={5},
  number={2},
  pages={101--123},
  year={2015},
  publisher={Taylor \& Francis}
}

@article{dai2021recent,
  title={Recent modeling and analytical advances in hospital inpatient flow management},
  author={Dai, Jim G. and Shi, Pengyi},
  journal={Production and Operations Management},
  volume={30},
  number={6},
  pages={1838--1862},
  year={2021},
  publisher={Wiley Online Library}
}

@article{kamali2019use,
  title={When to use provider triage in emergency departments},
  author={Kamali, Michael F. and Tezcan, Tolga and Yildiz, Ozlem},
  journal={Management Science},
  volume={65},
  number={3},
  pages={1003--1019},
  year={2019},
  publisher={INFORMS}
}

@article{saghafian2012patient,
  title={Patient streaming as a mechanism for improving responsiveness in emergency departments},
  author={Saghafian, Soroush and Hopp, Wallace J. and Van Oyen, Mark P. and Desmond, Jeffrey S. and Kronick, Steven L.},
  journal={Operations Research},
  volume={60},
  number={5},
  pages={1080--1097},
  year={2012},
  publisher={INFORMS}
}

@article{saghafian2014complexity,
  title={Complexity-augmented triage: A tool for improving patient safety and operational efficiency},
  author={Saghafian, Soroush and Hopp, Wallace J. and Van Oyen, Mark P. and Desmond, Jeffrey S. and Kronick, Steven L.},
  journal={Manufacturing \& Service Operations Management},
  volume={16},
  number={3},
  pages={329--345},
  year={2014},
  publisher={INFORMS}
}

@article{feizi2023a,
  title={Vertical Patient Streaming in Emergency Departments},
  author={Feizi, Arshya and Orfanoudaki, Agni and Saghafian, Soroush and Hudgson, Nicole},
  journal={SSRN},
  year={2023}
}

@article{he2019data,
  title={Data-driven patient scheduling in emergency departments: A hybrid robust-stochastic approach},
  author={He, Shuangchi and Sim, Melvyn and Zhang, Meilin},
  journal={Management Science},
  volume={65},
  number={9},
  pages={4123--4140},
  year={2019},
  publisher={INFORMS}
}

@article{huang2015control,
  title={Control of patient flow in emergency departments, or multiclass queues with deadlines and feedback},
  author={Huang, Junfei and Carmeli, Boaz and Mandelbaum, Avishai},
  journal={Operations Research},
  volume={63},
  number={4},
  pages={892--908},
  year={2015},
  publisher={INFORMS}
}

@article{li2021next,
  title={Who is next: Patient prioritization under emergency department blocking},
  author={Li, Wenhao and Sun, Zhankun and Hong, Jeff},
  journal={Operations Research},
  volume={71},
  number={3},
  pages={821--842},
  year={2023},
  publisher={INFORMS}
}

@article{niewoehner2022physician,
  title={Physician discretion and patient pick-up: How familiarity encourages multitasking in the emergency department},
  author={Niewoehner III, Robert J. and Diwas, KC and Staats, Bradley},
  journal={Operations Research},
  year={2023},
  volume={71},
  number={3},
  pages={958--978},
  publisher={INFORMS}
}

@article{andersen2019staff,
  title={Staff optimization for time-dependent acute patient flow},
  author={Andersen, Anders Reenberg and Nielsen, Bo Friis and Reinhardt, Line Blander and Stidsen, Thomas Riis},
  journal={European Journal of Operational Research},
  volume={272},
  number={1},
  pages={94--105},
  year={2019},
  publisher={Elsevier}
}

@article{zaerpour2022scheduling,
  title={Scheduling of Physicians with Time-Varying Productivity Levels in Emergency Departments},
  author={Zaerpour, Farzad and Bijvank, Marco and Ouyang, Huiyin and Sun, Zhankun},
  journal={Production and Operations Management},
  volume={31},
  number={2},
  pages={645--667},
  year={2022},
  publisher={Wiley Online Library}
}

@article{xu2016using,
  title={Using future information to reduce waiting times in the emergency department via diversion},
  author={Xu, Kuang and Chan, Carri W.},
  journal={Manufacturing \& Service Operations Management},
  volume={18},
  number={3},
  pages={314--331},
  year={2016},
  publisher={INFORMS}
}

@article{chan2021dynamic,
  title={Dynamic server assignment in multiclass queues with shifts, with applications to nurse staffing in emergency departments},
  author={Chan, Carri W. and Huang, Michael and Sarhangian, Vahid},
  journal={Operations Research},
  volume={69},
  number={6},
  pages={1936--1959},
  year={2021},
  publisher={INFORMS}
}

@article{ang2016accurate,
  title={Accurate emergency department wait time prediction},
  author={Ang, Erjie and Kwasnick, Sara and Bayati, Mohsen and Plambeck, Erica L. and Aratow, Michael},
  journal={Manufacturing \& Service Operations Management},
  volume={18},
  number={1},
  pages={141--156},
  year={2016},
  publisher={INFORMS}
}

@article{dai2019inpatient,
  title={Inpatient overflow: An approximate dynamic programming approach},
  author={Dai, Jim G. and Shi, Pengyi},
  journal={Manufacturing \& Service Operations Management},
  volume={21},
  number={4},
  pages={894--911},
  year={2019},
  publisher={INFORMS}
}

@article{feizi2023b,
  title={To batch or not to batch? {I}mpact of admission batching on emergency department boarding time and physician productivity},
  author={Feizi, Arshya and Carson, Anita and Jaeker, Jillian Berry and Baker, William Evan},
  journal={Operations Research},
  volume={71},
  number={3},
  pages={939--957},
  year={2023},
  publisher={INFORMS}
}

@article{chan2017queues,
  title={Queues with time-varying arrivals and inspections with applications to hospital discharge policies},
  author={Chan, Carri W. and Dong, Jing and Green, Linda V.},
  journal={Operations Research},
  volume={65},
  number={2},
  pages={469--495},
  year={2017},
  publisher={INFORMS}
}

@article{armony2018critical,
  title={Critical care capacity management: Understanding the role of a step down unit},
  author={Armony, Mor and Chan, Carri W. and Zhu, Bo},
  journal={Production and Operations Management},
  volume={27},
  number={5},
  pages={859--883},
  year={2018},
  publisher={Wiley Online Library}
}

@article{hu2021prediction,
  title={Prediction-Driven Surge Planning with Application to Emergency Department Nurse Staffing},
  author={Hu, Yue and Chan, Carri W. and Dong, Jing},
  journal={Management Science},
  volume={71},
  number={3},
  pages={2079--2126},
  year={2024},
}

@article{gonzalez2019proactive,
  title={A proactive transfer policy for critical patient flow management},
  author={Gonz{\'a}lez, Jaime and Ferrer, Juan-Carlos and Cataldo, Alejandro and Rojas, Luis},
  journal={Health Care Management Science},
  volume={22},
  number={2},
  pages={287--303},
  year={2019},
  publisher={Springer}
}

@article{dong2020structural,
  title={Capacity Management in Networks: A Structural Estimation Approach for Hospital Inpatient Wards},
  author={Dong, Jing and Shi, Pengyi and Zheng, Fanyin and Jin, Xin},
  journal={Manufacturing \& Service Operations Management},
  year={2026},
  doi={10.1287/msom.2023.0700}
}

@article{zhalechian2020personalized,
  title={ Data-Driven Hospital Admission Control: A Learning Approach},
  author={Zhalechian, Mohammad and Keyvanshokooh, Esmaeil and Shi, Cong and Van Oyen, Mark P.},
  journal={Operations Research},
  volume={71},
  number={6},
  pages={2111--2129},
  year={2023}
}

@article{batt2017early,
  title={Early task initiation and other load-adaptive mechanisms in the emergency department},
  author={Batt, Robert J. and Terwiesch, Christian},
  journal={Management Science},
  volume={63},
  number={11},
  pages={3531--3551},
  year={2017},
  publisher={INFORMS}
}

@article{kc2017benefits,
  title={Benefits of surgical smoothing and spare capacity: an econometric analysis of patient flow},
  author={Kc, Diwas Singh and Terwiesch, Christian},
  journal={Production and Operations Management},
  volume={26},
  number={9},
  pages={1663--1684},
  year={2017},
  publisher={Wiley Online Library}
}

@article{ding2019patient,
  title={Patient prioritization in emergency department triage systems: An empirical study of the {C}anadian triage and acuity scale ({CTAS})},
  author={Ding, Yichuan and Park, Eric and Nagarajan, Mahesh and Grafstein, Eric},
  journal={Manufacturing \& Service Operations Management},
  volume={21},
  number={4},
  pages={723--741},
  year={2019},
  publisher={INFORMS}
}

@article{soltani2022does,
  title={Does what happens in the {ED} stay in the {ED}? {The} effects of emergency department physician workload on post-{ED} care use},
  author={Soltani, Mohamad and Batt, Robert J. and Bavafa, Hessam and Patterson, Brian W.},
  journal={Manufacturing \& Service Operations Management},
  volume={24},
  number={6},
  pages={3079--3098},
  year={2022},
  publisher={INFORMS}
}

@article{batt2019effects,
  title={The effects of discrete work shifts on a nonterminating service system},
  author={Batt, Robert J. and Kc, Diwas S. and Staats, Bradley R. and Patterson, Brian W.},
  journal={Production and Operations Management},
  volume={28},
  number={6},
  pages={1528--1544},
  year={2019},
  publisher={Wiley Online Library}
}

@article{song2015diseconomies,
  title={The diseconomies of queue pooling: An empirical investigation of emergency department length of stay},
  author={Song, Hummy and Tucker, Anita L. and Murrell, Karen L.},
  journal={Management Science},
  volume={61},
  number={12},
  pages={3032--3053},
  year={2015},
  publisher={INFORMS}
}

@article{lin1992self,
  title={Self-improving reactive agents based on reinforcement learning, planning and teaching},
  author={Lin, Long-Ji},
  journal={Machine Learning},
  volume={8},
  pages={293--321},
  year={1992},
  publisher={Springer}
}

@misc{flusurge,
  author = {{CDC}},
  title = {{FluSurge 2.0}},
  year = {2016},
  note = {Accessed: 2023-6-23}
}

@book{porteus,
  title={Foundations of Stochastic Inventory Theory},
  author={Porteus, Evan L.},
  year={2002},
  publisher={Stanford Business Books}
}

@article{reid2016effectiveness,
  title={The effectiveness and variation of acute medical units: a systematic review},
  author={Reid, Lindsay E. M. and Dinesen, Lotte C. and Jones, Michael C. and Morrison, Zoe J. and Weir, Christopher J. and Lone, Nazir I.},
  journal={International Journal for Quality in Health Care},
  volume={28},
  number={4},
  pages={433--446},
  year={2016},
  publisher={International Society for Quality in Health Care and Oxford University Press}
}

@article{moloney2005impact,
  title={Impact of an acute medical admission unit on length of hospital stay, and emergency department ‘wait times’},
  author={Moloney, E. D. and Smith, D. and Bennett, K. and O'riordan, D. and Silke, B.},
  journal={QJM},
  volume={98},
  number={4},
  pages={283--289},
  year={2005},
  publisher={Oxford University Press}
}

@article{mehrotra2017,
    title={Starting with a clear endpoint: Development of a tool to predict admissions at triage.},
    author={Travers, Debbie and Mehrotra, Abhi and Chen, Wanyi and Lopiano, Kenneth and Bohrmann, Thomas and Argon, Nilay T. and Ziya, Serhan and Strickler, Jeffrey and Ring, Jackie and Linthicum, Benjamin},
    year={2017},
    journal={Academic Emergency Medicine},
    volume={volume 24 of 1},
    pages={S13-S14},
    note={Society for Academic Emergency Medicine}
}

@article{brown2005statistical,
  title={Statistical analysis of a telephone call center: A queueing-science perspective},
  author={Brown, Lawrence and Gans, Noah and Mandelbaum, Avishai and Sakov, Anat and Shen, Haipeng and Zeltyn, Sergey and Zhao, Linda},
  journal={Journal of the American Statistical Association},
  volume={100},
  number={469},
  pages={36--50},
  year={2005},
  publisher={Taylor \& Francis}
}

@article{shi2014patient,
  title={Patient flow from emergency department to inpatient wards: Empirical observations from a {Singaporean} hospital},
  author={Shi, Pengyi and Dai, Jim G. and Ding, Ding and Ang, Soo Keng and Chou, Mabel and Jin, Xin and Sim, Joe},
  year={2014},
note={http://dx.doi.org/10.2139/ssrn.2517050},
   journal={SSRN}
}

@book{lapan2018deep,
  title={{Deep reinforcement learning hands-on: Apply modern {RL} methods, with deep {Q}-networks, value iteration, policy gradients, {TRPO}, {AlphaGo Zero} and more}},
  author={Lapan, Maxim},
  year={2018},
  publisher={Packt Publishing Ltd.}
}

@misc{Arena,
  title = {{Arena simulation software}},
  author = {{Rockwell Automation}},
  year = {2023},
  url={https://www.rockwellautomation.com/en-us/products/software/arena-simulation.html},
  note = {Accessed: 2026-06-22}
}

@book{hastie2009,
title={The Elements of Statistical Learning: Data Mining, Inference, and Prediction},
author={Hastie, Trevor and Tibshirani, Robert and Friedman, Jerome},
year={2009},  
edition={second},
publisher={New York: Springer}
}

@article{Canellas2024,
title = {Measurement of Cost of Boarding in the Emergency Department Using Time-Driven Activity-Based Costing},
journal = {Annals of Emergency Medicine},
volume = {84},
number = {4},
pages = {376-385},
year = {2024},
author = {Maureen M. Canellas and Marcella Jewell and Jennifer L. Edwards and Danielle Olivier and Adalia H. Jun-O’Connell and Martin A. Reznek}
}

@article{pearce2023,
  title={Emergency department crowding: an overview of reviews describing measures, causes, and harms},
  author={Pearce, Sarah and Marchand, Thomas and Shannon, Thomas and Lang, Eddy},
  journal={Internal and Emergency Medicine},
  volume={18},
  number={5},
  pages={1137--1158},
  year={2023}  
}

@article{sartini2022,
  title={Overcrowding in Emergency Department: Causes, Consequences, and Solutions—A Narrative Review},
  author={Sartini, Marco and Carbone, Antonio and Demartini, Alessandro and Giribone, Laura and Oliva, Maria and Spagnolo, Anna M. and Cremonesi, Paola and Canale, Francesco and Cristina, Maria L.},
  journal={Healthcare},
  volume={10},
  number={9},
  pages={1625},
  year={2022}
}

@article{kishore2023, 
title={Early prediction of hospital admission of emergency department patients},
  author={Kishore, Karthik and Braitberg, George and Holmes, Natasha E.},
  journal={Emergency Medicine Australasia},
  volume={35},
  number={4},
  pages={572 -- 588},
  year={2023}
}

@article{kadri2023,
title={Towards accurate prediction of patient length of stay at emergency department: a GAN-driven deep learning framework},
  author={Kadri, Farid and Dairi, Abdelkader and Harrou, Fouzi and Sun, Ying},
  journal={Journal of Ambient Intelligence and Humanized Computing},
  volume={14},
  number={9},
  pages={11481--11495},
  year={2023}}

@thesis{linthicum2018improving,
  author  = {Linthicum, Benjamin},
  title   = {{Improving Emergency Department Throughput by Adoption of an Admissions Predictor Tool at Triage}},
  school  = {University of North Carolina at Chapel Hill},
  year    = {2018},
  type    = {Doctoral thesis},
  doi     = {10.17615/a5rv-w510}
}

@article{Wu2025,
  title={Reinforcement learning for healthcare operations management: methodological framework, recent developments, and future research directions},
  author={Wu, Qihao and Han, Jiangxue and Yan, Yimo and Kuo, Yong-Hong and Shen, Zuo-Jun Max},
  journal={Health Care Management Science},
  number={2},
  volume={28},
  pages={298-333},
  year={2025}
}

@article{Shin2025,
author = {Shin, Hyun A. and Kang, Hyeonji and Choi, Mona},
title = {Triage Data-Driven Prediction Models for Hospital Admission of Emergency Department Patients: A Systematic Review},
journal = {Healthcare Informatics Research},
volume = {31},
number = {1},
pages = {23--36},
year = {2025}
}
